\documentclass[11pt,a4paper]{article}
\usepackage{mathptmx} 
\usepackage{avant}     

\usepackage[ruled,vlined]{algorithm2e}
\usepackage[affil-it]{authblk}
\usepackage[english]{babel}

\usepackage{latexsym,amsfonts,amssymb}
\usepackage{amssymb,amsmath,bm,mathrsfs,makeidx,amsfonts,graphicx,amsthm}
\usepackage[authoryear]{natbib}
\usepackage{epsfig}
\usepackage{setspace}
\usepackage{colordvi,multicol}
\usepackage{subfig}
\usepackage{graphicx}
\usepackage{graphics}
\usepackage{booktabs}
\usepackage{multirow}
\usepackage{lscape}
\usepackage{pdflscape}
\usepackage{rotating}
\usepackage{multirow}
\usepackage{indentfirst}
\usepackage{caption}
\usepackage{array}
\usepackage{xcolor}

\usepackage{enumitem}

\usepackage[colorlinks=true,citecolor=blue]{hyperref}

\usepackage{tabularx,booktabs}
\newcolumntype{Y}{>{\centering\arraybackslash}X}

\newtheorem{thm}{Theorem}

\newtheorem{lem}{Lemma}

\newtheorem{remark}{Remark}

\newtheorem{assu}{Assumption}
\numberwithin{equation}{section}

\def\ba{\begin{align*}}
\def\ea{\end{align*}}
\def\bao{\begin{align}}
\def\eao{\end{align}}
\def\begine{\begin{enumerate}}
\def\ende{\end{enumerate}}

\def\be{\begin{equation}}
\def\ee{\end{equation}}

\DeclareMathOperator*{\argmin}{argmin}

\newcommand{\bbY}{{\boldsymbol Y}}
\newcommand{\bby}{{\boldsymbol y}}

\newcommand{\bbA}{{\boldsymbol A}}
\newcommand{\bba}{{\boldsymbol a}}
\newcommand{\bbB}{{\boldsymbol B}}
\newcommand{\bbC}{{\boldsymbol C}}

\newcommand{\bbD}{{\boldsymbol D}}
\newcommand{\bbd}{{\boldsymbol d}}

\newcommand{\bbE}{{\boldsymbol E}}

\newcommand{\bbb}{{\boldsymbol b}}
\newcommand{\bbF}{{\boldsymbol F}}

\newcommand{\bbG}{{\boldsymbol G}}
\newcommand{\bbH}{{\boldsymbol H}}

\newcommand{\bbI}{{\boldsymbol I}}

\newcommand{\bbJ}{{\boldsymbol J}}
\newcommand{\bbk}{{\boldsymbol k}}
\newcommand{\bbK}{{\boldsymbol K}}

\newcommand{\bbM}{{\boldsymbol M}}
\newcommand{\bbm}{{\boldsymbol m}}
\newcommand{\bbL}{{\boldsymbol L}}

\newcommand{\bbQ}{{\boldsymbol Q}}

\newcommand{\bbP}{{\boldsymbol P}}

\newcommand{\bbS}{{\boldsymbol S}}

\newcommand{\bbu}{{\boldsymbol u}}

\newcommand{\Sig}{\boldsymbol{\Sigma}}
\newcommand{\Lam}{\boldsymbol{\Lambda}}

\newcommand{\bmu}{\boldsymbol{\mu}}

\newcommand{\bvar}{\boldsymbol{\varepsilon}}

\newcommand{\stT}{\sum_{t=1}^T}

\newcommand{\RN}[1]{%
  \textup{\uppercase\expandafter{\romannumeral#1}}%
}

\usepackage[margin=1in]{geometry}
\usepackage{setspace}
\begin{document}



\title{A Forecast-driven Hierarchical Factor Model with Application to Mortality Data} 
\author[1]{Lingyu He}
\author[2]{Fei Huang}
\author[3]{Yanrong Yang\thanks{All coauthors have equal contribution to this paper. \\
Correspondence to: Yanrong Yang, College of Business and Economics, The Australian National University, Australia. Email: yanrong.yang@anu.edu.au}}
\affil[1]{Hunan University, China}
\affil[2]{UNSW Sydney, Australia}
\affil[3]{The Australian National University, Australia}
\maketitle

\begin{abstract}
	Mortality forecasting plays a pivotal role in insurance and financial risk management of life insurers, pension funds, and social securities. Mortality data is  usually high-dimensional in nature and favors factor model approaches to modelling and forecasting. This
	paper introduces a new forecast-driven hierarchical factor model (FHFM) customized for mortality forecasting.
	Compared to existing models, which only capture the cross-sectional variation or  time-serial dependence in the dimension reduction step, the new model captures both features  efficiently under
	a hierarchical structure, and provides insights into the  understanding of dynamic variation
	of mortality patterns over time. By comparing with static PCA utilized in \citet{lee1992modeling}, dynamic PCA introduced in \citet{LYB2011}, as well as other existing mortality modelling methods, we find that this approach
	provides both better estimation results and superior out-of-sample forecasting performance.
	Simulation studies further  illustrate the advantages of the proposed model
	based on different data structures. Finally, empirical studies using the US mortality data demonstrate the implications and significance of this new model in  life expectancy forecasting and life annuities pricing.

\end{abstract}

\noindent%
{\it Keywords:} Hierarchical factor model; high dimensional time series; life expectancy; mortality forecasting.


\section{Introduction}
The age-specific human mortality data consists of observations on either the death numbers or the death rates of a population under each age, measured for each historical year. Accurate forecasting of mortality data plays a crucial role in insurance and financial risk management of life insurers, pension funds, and social securities. For instance, the life expectancy of policyholders and the present values of life annuities are highly related to the future mortality rates. According to the life table published by \citet{perlifetable}, from $2016$ to $2095$, the life expectancy, which is the average remaining years of life, for a male aged $66$ in the US will rise from $17.2$ to $21.7$ years. Meanwhile, the present value (premium) of the corresponding life annuity, which pays annuities beginning from the year of age $66$ until death, will change from $\$13.94$ to $\$16.70$ per $\$1$ annual payment. Even a small  change in mortality forecasting may have  dramatic impact financially  on insurance companies and social security.  Therefore, a better mortality forecasting method, which guarantees more accurate estimations of life expectancy and premiums        of life annuity, is crucial for risk management and financial planning.

This paper aims to model and forecast the age-specific mortality data of the US population from the Human Mortality Database (HMD) (\citenum{HMD}). After preprocessing, the annual age-specific death rates under study consist of a matrix data with $84$ yearly observations (1933-2018) for $91$ ages ($0-90+$). Modelling and forecasting mortality data pose a challenge for traditional statistical analysis and multivariate time series analysis, as the dimension $91$ is comparable to the sample size (or time length) $86$. This high dimensional setting incurs the curse of dimensionality. Dimension reduction is a remedy method that extracts representative features or patterns of available high dimensional data. Statistical analysis on extracted features and recovery of corresponding inference on original data are common techniques in high dimensional data analysis. However, forecast-driven feature selection and statistical inference are rarely studied in high dimensional data analysis. This paper contributes to seeking forecast-driven linear features of mortality data by proposing a hierarchical factor model. Roughly speaking, a linear feature is a linear combination of annual death rates over the total $91$ ages, which is a univariate time series that summarizes the $91$-dimensional time series linearly. 
Before introducing the formal statistical model, we first analyze the US mortality data and interpret the features in pursuit intuitively.  

\begin{table}[!htbp] \centering 
  \small
  \caption{The Log Central Death Rates of the US} 
  \label{logdr} 
\begin{tabularx}{\textwidth}{c *{8}{Y}} 
\toprule
& \multicolumn{5}{c}{Historical data} 
& \multicolumn{3}{c}{Forecasts}\\ 
\cmidrule(lr){2-6} \cmidrule(l){7-9}
& 1933 & 1934 & 1935 & \dots & 2018 & 2019 & 2020 & \dots \\ 
\midrule 
0 & $-2.792$ & $-2.681$ & $-2.789$ & \dots & \dots &  ?& ?& ?\\ 
1 & $-4.661$ & $-4.551$ & $-4.720$ & \dots & \dots &  ?& ?& ?\\ 
2 & $-5.437$ & $-5.328$ & $-5.486$ & \dots & \dots &  ?& ?& ?\\ 
3 & $-5.775$ & $-5.735$ & $-5.816$ & \dots & \dots &  ?& ?& ?\\ 
4 & $-6.038$ & $-6.011$ & $-6.031$ & \dots & \dots &  ?& ?& ?\\ 
5 & $-6.227$ & $-6.200$ & $-6.210$ & \dots & \dots &  ?& ?& ?\\
\dots & \dots  & \dots & \dots  & \dots & \dots &  ?& ?& ?\\ 
90+ & \dots  & \dots & \dots  & \dots & \dots & ?& ?& ?\\
\bottomrule 
\end{tabularx} 
\end{table} 

We consider the logarithms of the death rates, because this transformation makes the positive-valued original data spread over total real-value set without loss of generality \citep{booth2008mortality}. Table \ref{logdr} shows the structure of the historical log death rates as well as the purpose of forecasting. It demonstrates a classical problem: modelling and forecasting a high dimensional time series. Figure \ref{logdr:a} illustrates the time-serial trend for each age; while Figure \ref{logdr:b} exhibits the age structure of mortality data at different years. From these plots, we can see that the death rates are generally decreasing over the years for most ages; and mortality rates at different ages has strong relations, especially consecutive ages. 
These characteristics  motivate us designing customized dimension reduction methods for mortality forecasting. More specifically, we can extract linear features that capture both the time-serial trend and  cross-sectional variations, which represent the  two most important characteristics of mortality data. 
In fact, the cross-sectional features help improve model fitting; while the time-serial features can enhance forecasting accuracy.

\begin{figure}
\begin{minipage}[t]{0.49\textwidth}
\includegraphics[width=\linewidth]{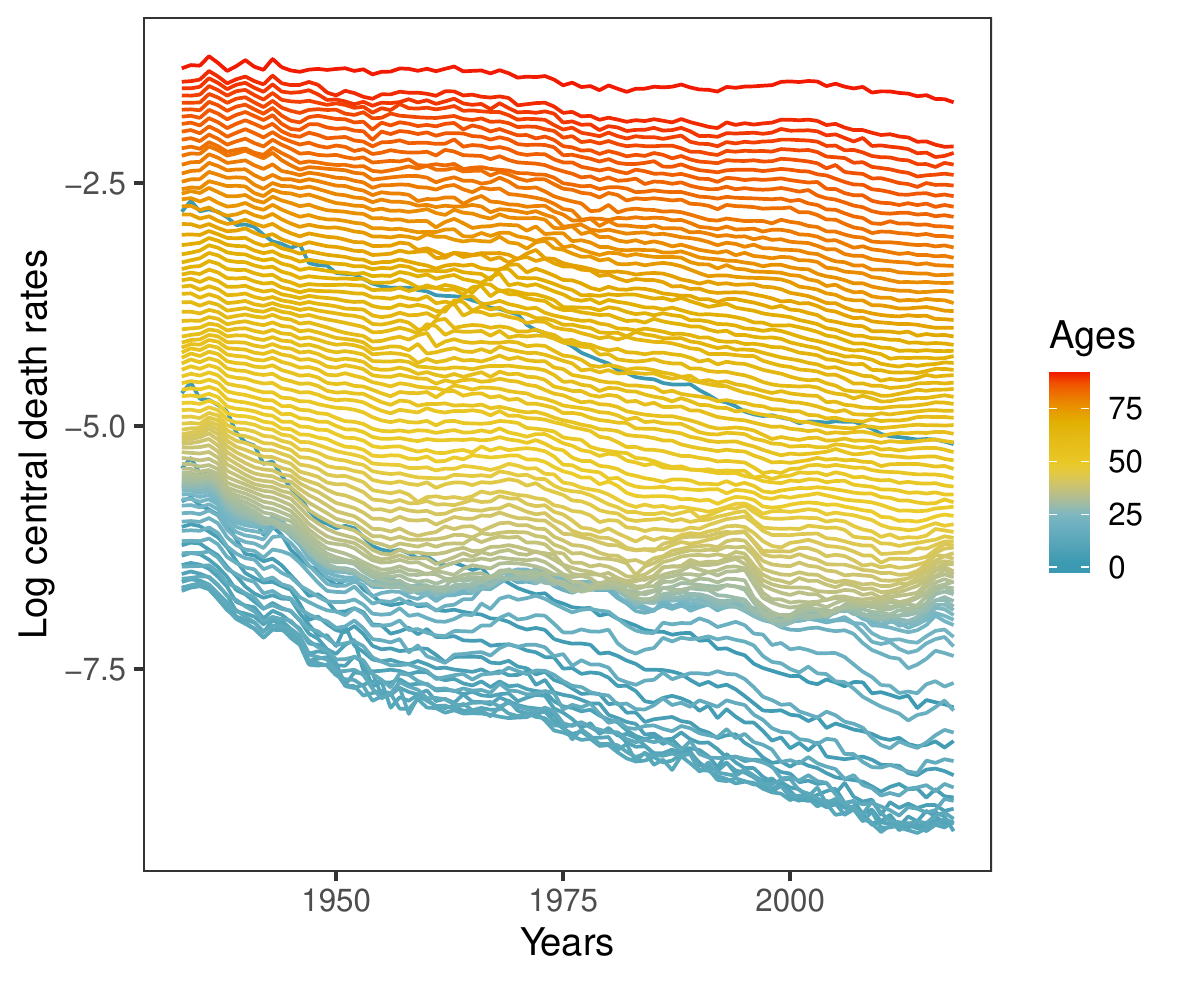}
\caption{The Log Central Death Rates, years $1933-2018$ for ages $0$ to $90+$.}
\label{logdr:a}
\end{minipage}
\hspace*{0.01cm}
\begin{minipage}[t]{0.49\textwidth}
\includegraphics[width=\linewidth]{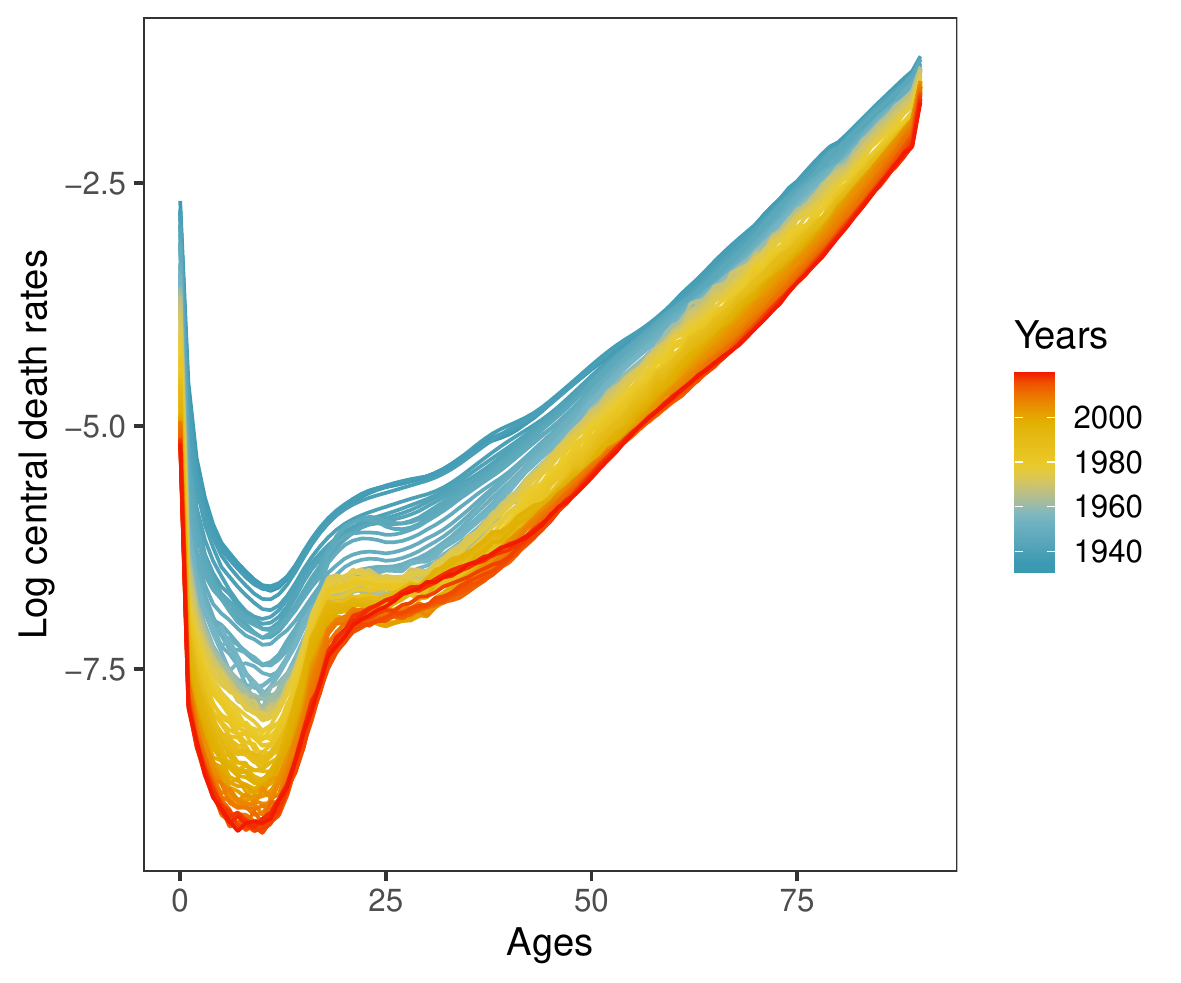}
\caption{The Log Central Death Rates, ages $0$ to $90+$ for years $1933-2018$.}
\label{logdr:b}
\end{minipage}
\end{figure}

There has been a large body of literature studying mortality modelling and forecasting, see a detailed review in \citet{booth2008mortality}. One seminal paper on the US mortality forecasting is \citet{lee1992modeling}. The Lee-Carter model is one of the most prominent methods for mortality forecasting, which has been used by the US Bureau of the Census as the benchmark model \citep{hollmann1999methodology}.  Later, many papers, such as \cite{booth2002applying}, \cite{renshaw2003lee}, \cite{hyndman2007robust}, \cite{yang2010modeling}, and \cite{HE202114} extended the Lee-Carter framework for mortality modeling in different countries. However, the feature extraction (dimension reduction) of existing models is usually not driven by forecast accuracy, which may not lead to optimal forecasting performance, consequently.  To this end,  this paper proposes a forecast-driven hierarchical factor model (FHFM) customized for mortality forecasting. It is noteworthy that hierarchical factor model appears in literature with various purposes (for example, see \cite{MNP2013}), which are different from the aim of optimal dimension reduction for forecasting in this paper.


As an instance of high dimensional time series,  mortality data faces both the curse of dimensionality and challenges in time series forecasting. We apply factor analysis (or dimension reduction) to overcome the challenges, which can extract common features (including both time-serial trends and cross-sectional variations) among all cross-sections (ages). The Lee-Carter model can be regarded as applying the traditional dimension-reduction method (static PCA) \citep{Anderson2003, MR2036084} to pursue common features that retain the largest variations of mortality data.  As a popular dimension reduction technique, PCA can be traced to that of \citet{Anderson2003, MR2036084}. \cite{HJLTZ2021} proposed a new scaled PCA approach which scales each predictor with its predictive slope on the target to be forecast and then improves forecasting results.  
Different from the static PCA, several papers search for common features that drive the time-serial dependence of the original high-dimensional time series. \citet{BL1975} and \citet{HKH2015} extended the static PCA to dynamic PCA, which extract features from a Fourier transformation on covariance and auto-covariances with different time-lags. \citet{LYB2011} and \citet{chang2018principal} extracted dynamic features by assembling auto-covariances in another way while excluding the covariance. Different from the existing factor analysis which focuses on either the static features or dynamic features, the proposed FHFM combines the dynamic and static features together via a two-step procedure, with the first step extracting factors with the most predictability and the second step extracting factors capturing the largest variations. Theoretically, we show that these two types of factors are both indispensable in producing optimal forecasting for mortality data. 

The factor-based forecasting procedure and error analysis are illustrated in Figure \ref{illustration}.  The original mortality data $\{y_t\}$ ($t=1,2, \dots, T$) is modelled by two types of factors $\bbk_t^{(1)}$ and $\bbk_t^{(2)}$, where $\bbk_t^{(1)}$ represents the features with the maximum predictability and $\bbk_t^{(2)}$ represents the features that capture the largest remaining variations. $A$ and $B$ are coefficient matrices. Time series models are applied to the extracted factors to forecast $\hat{\bbk}_{T+1}^{(1)}$ and $\hat{\bbk}_{T+1}^{(2)}$, which are then plugged in the estimated factor model to obtain the predicted mortality rate at $T+1$, $\hat{y}_{T+1}$. In doing so, we decompose the mortality data into three parts: a strong dynamic part driven by low-dimensional factor time series $\{\bbk_t^{(1)}\}$; a weak dynamic but strong variation part represented by another lower-dimensional factor time series $\{\bbk_t^{(2)}\}$; and an error part that is a high dimensional time series with weak serial dependence as well as small variations. In Figure \ref{illustration}, we see that $\textcircled{1}$ indicates the factor modelling process, which contains the model approximation error; $\textcircled{2}$ is model estimation, which contains the model estimation error; $\textcircled{3}$ is factor forecasting, which contains forecasting error; and $\textcircled{4}$ represents  the total mortality data forecasting error, which consists of the other three parts: modelling error $\textcircled{1}$, model estimation error $\textcircled{2}$, and factor forecasting error $\textcircled{3}$. The factors $\bbk_t^{(1)}$ could help to decrease the factor forecasting error $\textcircled{3}$ while the factors $\bbk_t^{(2)}$ is able to minimize the factor modelling error $\textcircled{1}$. 
To maximize the prediction accuracy, our aim is to capture both the common dynamic features and common static features that are most helpful for forecasting. 
\begin{figure} 
\centering
\includegraphics[width=0.8\linewidth]{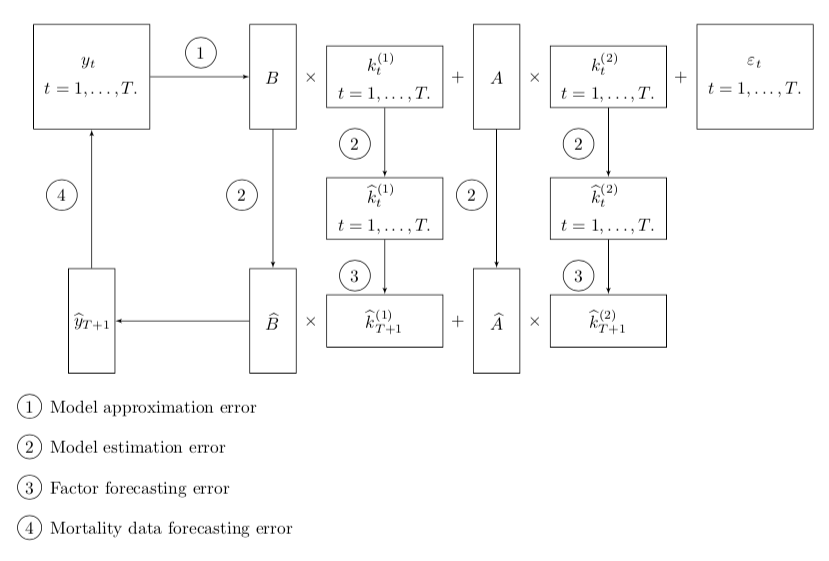}
\caption{Error Analysis for Factor-based Forecasting}
\label{illustration}
\end{figure}

How to obtain the two types of factors $\bbk_t^{(1)}$ and $\bbk_t^{(2)}$?  \cite{BT1977} applied canonical analysis to a stationary multivariate time series to identify and order linear features in terms of their predictability. Using a similar idea, we extract $\bbk_t^{(1)}$ as the  linear feature that explains the most predictability. However, to maximize the prediction accuracy, we also need to make sure the extracted features are a good low-rank approximation of the original mortality data, as $\bbk_t^{(1)}$ alone may not be sufficient to explain the data variation. To this end, we continue to extract the second type of factors ($\bbk_t^{(2)}$) that explains the most remaining variations. Estimation of this hierarchical factor model is carried out by a two-step eigenanalysis for an auto-covariance matrix and a covariance matrix, respectively. Through investigating asymptotic properties of the proposed method, we find that the estimation of the two types of factors have equal rates of convergence. It is noteworthy that the two steps to obtain the factors are both necessary to achieve optimal prediction accuracy. This model is superior compared to factor models based on the static PCA or dynamic PCA alone, which can only capture one of the two features considered in the FHFM and hence may ignore the other information to further reduce the forecasting error.

In this paper, we show that the proposed new method can improve forecasting accuracy compared to the benchmark methods. Moreover, empirical analysis using the US mortality data shows that the pricing error of a life annuity project using the Lee-Carter method is around \$$0.154$ per $\$1$ annual payment; however, the pricing error using the proposed FHFM  is only about $\$0.041$ per $\$1$. Although the difference looks very small, it will actually have significant financial impact on life insurers or social security. To illustrate the financial impact on the industry, we consider an annuity product as an example. Suppose the annual payment of the annuity is $\$20,000$ per person and there are $5,000$ individuals under cover, then the $\$0.113$ per dollar improvement in annuity pricing of the FHFM over Lee-Carter model will account for $\$1.13$ million overall improvement in annuity valuation of the whole group, which is a remarkable improvement. 

The rest of the paper is organized as follows. In Section \ref{dpca:model}, the details of the model are described, including model interpretations, the estimation and forecasting methods. A practical algorithm is provided as well. 
In Section \ref{dpca:relation}, we discuss the similarity and differences of our method compared with static PCA and dynamic PCA methods. 
The asymptotic properties of the proposed method are presented in Section \ref{dpca:asymptotic} and the corresponding proofs are in Appendix \ref{append:cha:stepPCA.2}.
Simulation studies demonstrating the advantages of the new method are presented in Section \ref{dpca:simu} and Appendix \ref{append:cha:stepPCA.1}. In Section \ref{dpca:empirical}, we compare the forecasting performance of the proposed new method with traditional methods using the US age-specific mortality data. The differences in   life expectancy calculations and life annuity prices are also assessed. Finally, the conclusion is presented in Section \ref{dpca:conclu}.

The notations in this paper are summarized here. For an $p \times n$ matrix $\bbC$, we denote its transpose as $\bbC^{\top}$, the square root of the maximum eigenvalue of $\bbC\bbC^{\top}$ as $\|\bbC\|$, and the square root of the smallest nonzero eigenvalue of $\bbC\bbC^{\top}$ as $\|\bbC\|_{\text{min}}$.
For a $k\times k$ matrix $\bbF$, $\lambda_i(\bbF)$ indicates the $i$-th largest eigenvalue of the matrix $\bbF$. For a non-symmetric matrix $\bbS$, we use $\sigma_j\left(\bbS\right)$ to denote the singular value of the matrix $\bbS$, which corresponds to the $j$-th largest eigenvalue of the matrix $\bbS\bbS^{\top}$. $\bbI_p$ represents $p-$dimensional identity matrix. All vectors are column vectors. 
The notation $a\asymp b$ means that $a=O(b)$ and $b=O(a)$.
$\stackrel{i.p.}{\longrightarrow}$ denotes convergence in probability. We use $P, T \rightarrow \infty$ to denote that $P$ and $T$ go to infinity jointly.

\section{Model and Estimation}
\label{dpca:model}
Let $\bbm_t = (m_{1,t}, m_{2,t}, \dots, m_{P,t})^{\top}$ be the US age-specific death rates in year $t$, where $m_{p,t}$ is the death rate for age $p$ in year $t$ with $p = 1,2, \dots, P$ and $t = 1,2, \dots, T$. The historical mortality data is available annually from the year $1933$ to the year $2018$ for ages from $0$ to $90+$. For high dimensional time series $\{\bbm_t, t=1, 2, \ldots, T\}$, the time-serial length and the dimension are $T = 86$ and $P = 91$, respectively. We will propose a two-step dimension reduction model on the log transformation of $\bbm_t$ which is denoted by $\bby_t = (\ln (m_{1,t}), \ln (m_{2,t}), \dots, \ln (m_{P,t}))^{\top}$. It is worth noting that building a model is much easier on $\bby_t$ than that on $\bbm_t$ because $\bbm_t$ takes non-negative values.   

In this section, we first introduce the FHFM, then propose a two-step estimation method for the model, and lastly provide the  forecasting procedure.

\subsection{The Forecast-driven Hierachical Factor Model}

As analyzed in last section, death rates for all the $91$ ages possess common features that drive common time-serial trend and common variations, respectively. This leads us to the following FHFM: for any $t=1, 2, \ldots, T \ (T=86)$, 
\begin{eqnarray}
&&\bby_t=\bbB\bbk_t^{(1)}+\bbu_t, \label{dpca:eq:y1}\\
&&\bbu_t=\bbA\bbk_t^{(2)}+\boldsymbol{\varepsilon}_t, \label{dpca:eq:y2} 
\end{eqnarray}
where $\bbk_t^{(1)}=\left(k_{1t}^{(1)}, k_{2t}^{(1)}, \ldots, k_{r_1t}^{(1)}\right)^{\top}$ is an $r_1\times 1$ latent process with $r_1<P$, which represents common temporal trends; $\bbB=(\bbb_1, \bbb_2, \ldots, \bbb_{r_1})$ is a $P\times r_1$ unknown deterministic coefficients matrix; similarly, $\bbk_t^{(2)}=\left(k_{1t}^{(2)}, k_{2t}^{(2)}, \ldots, k_{r_2t}^{(2)}\right)^{\top}$ is an $r_2\times 1$ latent process with $r_2<P$, which indicates common variations among all ages; $\bbA=\left(\bba_1, \bba_2, \ldots, \bba_{r_2}\right)$ is the corresponding $P\times r_2$ unknown deterministic coefficients matrix; and $\boldsymbol{\varepsilon}_t$ is an error component. 

Here we assume $r_1$ and $r_2$ are both unknown positive integers. Once $P$ is much larger than $(r_1+r_2)$, an effective dimension reduction is achieved because the original time series $\bby_t$ is driven by a much lower dimensional time series $\left(\bbk_t^{(1)}, \bbk_t^{(2)}\right)$. We also call model (\ref{dpca:eq:y1}) and (\ref{dpca:eq:y2}) the factor models with $\bbk^{(1)}_t$ and $\bbk^{(2)}_t$ being common factors, respectively. Thus $\bbB$ and $\bbA$ are the corresponding factor loadings, respectively. Factor model is a popular dimension reduction model in high dimensional statistics, which is investigated in huge amounts of literature, including \cite{B2002}, \cite{LYB2011}, and \cite{LY2012}. 

It is noted that FHFM  involves two kinds of common factors $\bbk^{(1)}_t$ and $\bbk^{(2)}_t$, which represent common temporal trends and common variations among all the $p$ ages, respectively. These two kinds of common factors are necessary in producing good forecasting results. 

Because all elements in the model are unknown, including factor loadings and common factors, we should impose identification conditions to make the model well-defined. 
First, we assume that the rank of factor loadings $\bbB$ and $\bbA$ are equal to $r_1$ and $r_2$, respectively. Otherwise, the two parts $\bbB\bbk_t^{(1)}$ and $\bbA\bbk_t^{(2)}$ can be represented in terms of factor models with even lower dimension. Moreover, as factors and factor loadings are all unknown, for any $r_1\times r_1$ invertible matrix $\bbH$, if we substitute the factor model part $\left(\bbB, \bbk_t^{(1)}\right)$ with $\left(\bbB\bbH, \bbH^{-1}\bbk_t^{(1)}\right)$, the term $\bbB\bbk_t^{(1)}$ is unchanged. It is also true for the term $\bbA\bbk_t^{(2)}$. To avoid such matters, we impose the following assumption.  
\begin{assu}\label{assu1}
{\color{blue}{Orthogonal Condition}}.  $\bbB^{\top}\bbB=\bbI_{r_1}$, 
$\bbA^{\top}\bbA=\bbI_{r_2}$, where $\bbI_{r_1}$ and $\bbI_{r_2}$ are $r_1\times r_1$ and $r_2\times r_2$ identity matrices, respectively. 
\end{assu}
Under Assumption \ref{assu1}, the factor loading $\bbB$ and the common factor $\bbk_t^{(1)}$ are determined up to an orthogonal matrix $\bbH$, and the same for the pair $\left(\bbA, \bbk_t^{(2)}\right)$. In this way, Assumption \ref{assu1} provides identification conditions between common factors and the corresponding factor loadings. It is also a common identification condition for factor models used in literature including \cite{B2002}, \cite{LY2012}.

Secondly, we consider the identification between the two kinds of common factors. As mentioned earlier, the two kinds of factor parts represent common temporal trends and common variations of the data, respectively. Intuitively, we think the first common factor part possesses stronger time serial dependence than the second factor part. Formally, we use the auto-covariance to distinguish the two parts, which is reasonable since auto-covariance can describe strength of time-serial dependence. Before introducing Assumption \ref{assu2}, we specify some notations. $\boldsymbol{\Sigma}_k^{(1)}(\ell):=cov\left(\bbk_t^{(1)}, \bbk_{t+\ell}^{(1)}\right)$ and $\boldsymbol{\Sigma}_k^{(2)}(\ell):=cov\left(\bbk_t^{(2)}, \bbk_{t+\ell}^{(2)}\right)$ are auto-covariance of $\bbk_t^{(1)}$ with lag $\ell$ and that of $\bbk_t^{(2)}$ with lag $\ell$, respectively. For any matrix $\bbC$, let $\left|\left|\bbC\right|\right|$ be the square root of the maximum eigenvalue of $\bbC\bbC^{\top}$ and $\left|\left|\bbC\right|\right|_{\min}$ be the square root of the smallest nonzero eigenvalue of $\bbC\bbC^{\top}$. 
\begin{assu}\label{assu2}
{\color{blue}{Identification between $\bbk_{t}^{(1)}$ and $\bbk_{t}^{(2)}$}}. $\left|\left|\boldsymbol{\Sigma}_k^{(1)}(\ell)\right|\right|\asymp P^{1-\delta_1}\asymp\left|\left|\boldsymbol{\Sigma}_k^{(1)}(\ell)\right|\right|_{\min}$, 
$\left|\left|\boldsymbol{\Sigma}_k^{(2)}(\ell)\right|\right|\asymp P^{1-\delta_2}\asymp\left|\left|\boldsymbol{\Sigma}_k^{(2)}(\ell)\right|\right|_{\min}$, where $0\leq\delta_1<\delta_2\leq 1$. 
\end{assu}
Assumption \ref{assu2} imposes different orders for eigenvalues of the auto-covariance matrices for the two kinds of common factors. The order $P^{1-\delta_1}$ for $\bbk_t^{(1)}$ is larger than the order $P^{1-\delta_2}$ for $\bbk_t^{(2)}$ as $\delta_1<\delta_2$, which ensures that the time-serial dependence of $\bbk_t^{(1)}$ is stronger than that of $\bbk_t^{(2)}$. In view of this, Assumption \ref{assu2} identifies the two kinds of factor parts via their time-serial dependence. In other words, the first factor model part extracts common factors with stronger time-serial dependence, which also takes higher priority in the forecasting. This kind of identification condition is utilized in \cite{LY2012}.

Thirdly, we distinguish the second kind of factor part from the error component of the model. After extracting the common temporal trends in the first part of the model, the aim on better forecasting stimulates us to pursue further necessary features in the data. Compared to the factor $\bbk_t^{(1)}$, the factor $\bbk_t^{(2)}$ has weaker time-serial dependence. It has little interest for forecasting improvement. However, it implies large amounts of common variations of the data. Neglecting it will result in unsatisfied model fitting of the original data. As better model fitting also plays an important role in improving forecasting, we would like to keep them in the dimension reduction as well.    
\begin{assu}\label{assu3}
{\color{blue}{Identification between $\bbk_{t}^{(2)}$ and $\boldsymbol{\varepsilon}_t$}}.\\ $\frac{1}{T}\sum^{T}_{t=1}\bbk_{t}^{(2)}\bbk_{t}^{(2)\top}\stackrel{i.p.}{\longrightarrow}\boldsymbol{\Sigma}_{k}^{(2)}(0)>0$, as $P, T\rightarrow\infty$. Here $\boldsymbol{\Sigma}_k^{(2)}(0)$ is a deterministic $r_2\times r_2$ positive definite matrix. 
\end{assu}
Assumption \ref{assu3} is a common condition on factor model in the sense that the factors represent most variations of the data. It is used in \cite{B2002}.

At last, we impose some conditions on the error component. 
\begin{assu}\label{assu4}
{\color{blue}{Error components}}. 
\begin{enumerate}
\item 
$\mathbb{E}\left(\varepsilon_{it}\right)=0$. 
$\{\boldsymbol{\varepsilon}_t: t\geq 1\}$ is strictly stationary. 
\item $\sum^{P}_{i=1}\sum^{P}_{j=1}\sum^{T}_{t=1}\sum^{T}_{s=1}\left|\mathbb{E}\left(\varepsilon_{it}\varepsilon_{js}\right)\right|=O(PT)$ and $\sum^{P}_{i=1}\sum_{j\neq i}\left|\sigma_{\varepsilon, ij}\right|=O(P)$, where $\sigma_{\varepsilon, ij}:=\mathbb{E}\left(\varepsilon_{it}\varepsilon_{jt}\right)$. 
\end{enumerate}
\end{assu}
Condition 2 of Assumption \ref{assu4} ensures that only weak cross-sectional dependence and time-serial dependence exist in the error component. This condition indicates that no obvious common variations and common temporal trends are involved in the error component. 

In summary, Assumptions \ref{assu1}-\ref{assu4} create a well-defined FHFM (\ref{dpca:eq:y1}) and (\ref{dpca:eq:y2}). Next, we will consider how to estimate the two kinds of common factors for further forecasting.

\subsection{Hierarchical Factors}
Models (\ref{dpca:eq:y1}) and (\ref{dpca:eq:y2}) propose hierarchical factors which capture the time-serial dependence and cross-sectional variations, respectively. The factors $\bbk_t^{(1)}$ could help to decrease the factor forecasting error, while the factors $\bbk_t^{(2)}$ is able to minimize the factor modelling error. Combining these two types of factors, we could achieve optimal forecasting performance.  In this section, we discuss how to obtain these two types of factors from the perspective of optimal factor analysis.  

First, we review the linear combinations with the maximum predictability defined in \cite{BT1977}. 
We consider the forecasting model for the high dimensional time series $\bby_t$
\begin{eqnarray}\label{y04(1)}
\bby_t=\widehat{\bby}_{t-1}(1)+\boldsymbol{\eta}_t, 
\end{eqnarray}
where 
\begin{eqnarray*}
\widehat{\bby}_{t-1}(1)=\mathbb{E}\left(\bby_t|\bby_{t-1}, \bby_{t-2}, \ldots\right)
\end{eqnarray*}
is the expectation of $\bby_t$ conditional on past history up to the time $t-1$. So it is also the one-step ahead prediction. $\boldsymbol{\eta}_t$ is the forecast error. 

Since $\bby_t$ is a high dimensional time series, we 
pursue the linear combination $k_t=\bbb^{\top}\bby_t$ which has the maximum predictability. In terms of (\ref{y04(1)}), the univariate time series $k_t$ satisfies
\begin{eqnarray*}
k_t=\widehat{k}_{t-1}(1)+\zeta_t, 
\end{eqnarray*}
where $\widehat{k}_{t-1}(1)=\bbb^{\top}\widehat{\bby}_t(1)$ and $\zeta_t=\bbb^{\top}\boldsymbol{\eta}_t$. 

Because $k_t$ is stationary, 
\begin{eqnarray}\label{y04(2)}
var\left(k_t^2\right)=var\left(\widehat{k}_{t-1}^2(1)\right)+var\left(\zeta_t^2\right). 
\end{eqnarray} 
\cite{BT1977} proposed the quantity $\lambda$ to measure the predictability of a univariate stationary time series as 
\begin{eqnarray}\label{y04(3)}
\lambda:=\frac{var\left(\widehat{k}_{t-1}(1)\right)}{var(k_t)}=\frac{\bbb^{\top}\boldsymbol{\Gamma}^{*}_y(0)\bbb}{\bbb^{\top}\boldsymbol{\Gamma}_{y}(0)\bbb}, 
\end{eqnarray}
where $\boldsymbol{\Gamma}_y=\mathbb{E}\left(\bby_t\bby_t^{\top}\right)$ and $\boldsymbol{\Gamma}_y^{*}=\mathbb{E}\left(\widehat{\bby}_{t-1}(1)\widehat{\bby}_{t-1}^{\top}(1)\right)$. 

It follows from (\ref{y04(3)}) that the maximal predictability $\lambda$ should be the largest eigenvalue of the matrix $\boldsymbol{\Gamma}_y^{-1}\boldsymbol{\Gamma}_y^{*}$ and the linear combination vector $\bbb$ is the eigenvector corresponding to the largest eigenvalue. In other words, the eigenvector corresponds to the smallest eigenvalue will yield the least predictable combination of $\bby_t$.  

To make the matrix $\boldsymbol{\Gamma}_y^{*}$ clearer, we consider the linear forecasting model 
\begin{eqnarray}
\bby_t=\bbG\bby_t^{*}+\boldsymbol{\delta}_t, 
\end{eqnarray}
where $\bby_t^{*}=\left(\bby_{t-1}^{\top}, \bby_{t-2}^{\top}, \ldots, \bby_{t-\tau}\right)$ is the past information; $\bbG$ is an unknown $P\times (P\tau)$ coefficient matrix; and $\boldsymbol{\delta}_t$ is the forecast error. 

As the least-square estimator for $\bbG$ is $\widehat{\bbG}=\boldsymbol{\Gamma}_{yy^{*}}\boldsymbol{\Gamma}_{y^{*}}^{-1}$ with $\boldsymbol{\Gamma}_{yy^{*}}=\mathbb{E}\left(\bby_t\bby_t^{*\top}\right)$ and $\boldsymbol{\Gamma}_{y^{*}}=\mathbb{E}\left(\bby^{*}\bby^{*\top}\right)$, we have 
\begin{eqnarray}\label{y1105(1)}
\boldsymbol{\Gamma}_{y}^{*}=\boldsymbol{\Gamma}_{yy^{*}}\boldsymbol{\Gamma}_{y^{*}}^{-1}\boldsymbol{\Gamma}_{yy^{*}}^{\top}. 
\end{eqnarray}
Based on (\ref{y1105(1)}), $\lambda$ should be the largest eigenvalue of the canonical correaltion matrix 
\begin{eqnarray}
\bbS_{yy^{*}}:=\boldsymbol{\Gamma}_y^{-1}\boldsymbol{\Gamma}_{yy^{*}}\boldsymbol{\Gamma}_{y^{*}}^{-1}\boldsymbol{\Gamma}_{yy^{*}}^{\top}. 
\end{eqnarray}
In view of this point, the linear combination that results in the maximal predictability should be the eigenvector of the matrix $\bbS_{yy^{*}}$ corresponding to the largest eigenvalue. This conclusion provides us with insights into recovering the factors with the maximum predictability, which is the first step of the hierarchical factor modelling. 

It is noteworthy that we estimate the maximum predictable factors based on the eigen-decomposition of the matrix $\boldsymbol{\Gamma}_{yy^{*}}\boldsymbol{\Gamma}^{\top}_{yy^{*}}$ instead of the canonical correlation matrix $\bbS_{yy^{*}}$, which will be shown in Section \ref{sec: est}. Actually, the matrix $\bbS_{yy^{*}}$ is a normalized version of the symmetric matrix $\boldsymbol{\Gamma}_{yy^{*}}\boldsymbol{\Gamma}^{\top}_{yy^{*}}$, by normalizing the different variances of mortality data across all ages. However, as our FHFM model also pursues factors with largest variations in the second step,  we utilize the non-scaled version $\boldsymbol{\Gamma}_{yy^{*}}\boldsymbol{\Gamma}^{\top}_{yy^{*}}$, for simplicity.  

After recovering the factors with the maximum predictability, we extract the second type of factors to explain the largest variations using PCA as the second step. The implementation of PCA is based on the eigen-decomposition of the covariance matrix $\boldsymbol{\Gamma}_y$ of the residuals from the first step.

\subsection{Estimation Approach} \label{sec: est}
Based on the identification between $\bbk_t^{(1)}$ and $\bbk_t^{(2)}$, the factor $\bbk_t^{(1)}$ will play a leading role in the auto-covariance matrix $\boldsymbol{\Sigma}_y(\ell):=cov\left(\bby_t, \bby_{t+\ell}\right)$, with $\ell$ being a positive integer. To see this point clearly, we do some calculations under the case of $\bbk_t^{(1)}$ and $\bbk_t^{(2)}$ being independent. It follows from (\ref{dpca:eq:y1}) and (\ref{dpca:eq:y2}) that
\begin{eqnarray}
\boldsymbol{\Sigma}_y(\ell)=\bbB\boldsymbol{\Sigma}_k^{(1)}(\ell)\bbB^{\top}
+\bbA\boldsymbol{\Sigma}_k^{(2)}(\ell)\bbA^{\top}
+\boldsymbol{\Sigma}_{\varepsilon}(\ell), 
\end{eqnarray}
where $\boldsymbol{\Sigma}_{\varepsilon}(\ell)=cov\left(\boldsymbol{\varepsilon}_t, \boldsymbol{\varepsilon}_{t+\ell}\right)$. 

With Assumption \ref{assu2} and Assumption \ref{assu4}, $\bbB\boldsymbol{\Sigma}_k^{(1)}(\ell)\bbB^{\top}$ is the leading term of $\boldsymbol{\Sigma}_y(\ell)$ in the sense of spectral norm. 
As auto-covariance matrices are not symmetric, we consider the matrix 
\begin{eqnarray}
\bbL(\ell):=\boldsymbol{\Sigma}_y(\ell)\boldsymbol{\Sigma}_{y}(\ell)^{\top}. 
\end{eqnarray} 
It is easy to obtain that the columns of $\bbB$ are the eigenvectors of the matrix
$\bbB\boldsymbol{\Sigma}_k^{(1)}(\ell)\boldsymbol{\Sigma}_k^{(1)}(\ell)^{\top}\bbB^{\top}$ corresponding to its non-zero eigenvalues. In fact, if $\bbC=\left(\bbb_1, \ldots, \bbb_{P-r_1}\right)$ is a $P\times (P-r_1)$ matrix for which $\left(\bbB, \bbC\right)$ forms a $P\times P$ orthogonal matrix, that is $\bbC^{\top}\bbB=\textbf{0}$ and $\bbC^{\top}\bbC=\bbI_{P-r_1}$, then we have $\left(\bbB\boldsymbol{\Sigma}_k^{(1)}(\ell)\boldsymbol{\Sigma}_k^{(1)}(\ell)^{\top}\bbB^{\top}\right)\bbC=\textbf{0}$. That is, the columns of $\bbC$ are eigenvectors of $\bbB\boldsymbol{\Sigma}_k^{(1)}(\ell)\boldsymbol{\Sigma}_k^{(1)}(\ell)^{\top}\bbB^{\top}$ corresponding to zero eigenvalues. 

Furthermore, the matrix $\bbB\boldsymbol{\Sigma}_k^{(1)}(\ell)\boldsymbol{\Sigma}_k^{(1)}(\ell)^{\top}\bbB^{\top}$ is the leading term of the matrix $\bbL(\ell)$ in the sense of spectral norm. Hence the columns of $\bbB$ are close to the eigenvectors of the matrix $\bbL(\ell)$ corresponding to non-zero eigenvalues, approximately. 

In terms of the analysis above, the eigendecomposition of $\bbL(\ell)$ provides a recovery method of the factor loading matrix $\bbB$. Note that we use $\ell = 1$ in the estimation step, because the estimation of $\bbB$ is not sensitive to $\ell$ and the correlation is often at its strongest at the small time lag (\citet{LY2012}). Besides, after analyzing the US mortality data, we also find $\ell = 1$ is enough for the forecasting.

Back to the FHFM (\ref{dpca:eq:y1}) and (\ref{dpca:eq:y2}), given an estimator for the first factor part, the recovery of the second factor part is more straightforward. In fact, the model is reduced to a simpler form 
\begin{eqnarray}\label{dpca:eq:y3}
\bby_t-\bbB\bbk_t^{(1)}=\bbA\bbk_t^{(2)}+\boldsymbol{\varepsilon}_t. 
\end{eqnarray} 
Under Assumption \ref{assu3}, (\ref{dpca:eq:y3}) is a classical factor model which can be estimated by the standard (static) PCA. See \cite{FLM2013}.

In summary, the estimation of the FHFM has two dimension reduction steps. The first step is to extract features that can enhance forecasting accuracy via a dynamic PCA procedure. The second step is to extract features, that retain variations to the largest extent for each age by performing static PCA. Next let us discuss in details about the two steps.


\subsubsection*{The First Step}
Firstly, we assume that $\{\bby_t\}_{t = 1,2, \dots, T}$ is covariance stationary and consider the following matrix 
\begin{align*}
\bbL_1 =  \Sig_{\bby}(1) \Sig_{\bby}(1)^{\top},
\end{align*}
where $\Sig_{\bby}(1) = \text{cov}{(\bby_t, \bby_{t+1})}$. As $\bbL_1$ is a symmetric matrix, it can be decomposed as $\bbL_1 = \bbQ\Lam\bbQ^{\top}$. The $P \times P$ matrix $\bbQ$ consists of the orthogonal eigenvectors of $\bbL_1$ in the columns and the columns are arranged such that the corresponding eigenvalues are in descending order. $\Lam$ is a $P \times P$ diagonal matrix with eigenvalues of $\bbL_1$ as the diagonal elements in descending order. As $\bbQ$ is an orthogonal matrix, we have $\bbQ^{\top}\bbQ = \bbQ\bbQ^{\top} = \bbI$. 

Let $\bmu_{\bby} = \mathbb{E}(\bby_t)$, and $\bbb_i$ be the $i^{th}$ column of $\bbQ$, which is the eigenvector corresponding to the $i^{th}$ largest eigenvalue of $\bbL_1$. Without loss of generality, we assume $\bmu_{\bby} = \mathbf{0}$ for convenience. Then by some simple rearrangement, we have  
\begin{align}
\bby_t &= \sum_{i = 1}^{r_1} \bbb_i \bbb_i^{\top}\bby_t + \sum_{i = r_1 +1}^{P} \bbb_i \bbb_i^{\top}\bby_t.		\label{dpca:eq:ex1}			
\end{align}
Let $k_{it}^{(1)} = \bbb_i^{\top}\bby_t$, and $\bbu_t = \sum_{i = r_1 +1}^{P} \bbb_i \bbb_i^{\top}\bby_t$. Then we can rewrite equation~(\ref{dpca:eq:ex1})  as
\begin{align}
\bby_t &= \sum_{i = 1}^{r_1} \bbb_i k_{it}^{(1)} + \bbu_t. \label{dpca:eq:ex2}
\end{align}
Then the linear combination $k_{it}^{(1)}, i = 1, 2, \dots, r_1$, are the features representing the time-serial trends, which are supposed to have good forecasting behaviors.

\subsubsection*{The Second Step}
The second step is equivalent to do a static PCA on $\bbu_t$ in equation~(\ref{dpca:eq:ex2}). Let $\Sig_{\bbu}(0) = var{(\bbu_t)}$, then the desired matrix for the second step is:
\begin{align*}
\bbL_2 = \Sig_{\bbu}(0) \Sig_{\bbu}(0)^{\top}.
\end{align*}
We conducting eigendecomposition on $\bbL_2$ and let $\bba_i$ be the eigenvector corresponding to the $i^{th}$ largest eigenvalue of $\bbL_2$. Then similar to that in the first step,
$\bbu_t$ can be expressed as:
\begin{align}
\bbu_t = \sum_{i = 1}^{r_2} \bba_i k_{it}^{(2)} +\bvar_t \label{dpca:eq:ex3},
\end{align}
where $k_{it}^{(2)} = \bba_i^{\top}\bbu_t$ and $\bvar_t = \sum_{i = r_2 +1}^{P} \bba_i \bba_i^{\top}\bbu_t$. $k_{it}^{(2)} = \bba_i^{\top}\bbu_t, i = 1, 2, \dots, r_2$ is the features extracted from the second step, which capture most of the common variations.


Finally combining equation (\ref{dpca:eq:ex2}) and (\ref{dpca:eq:ex3}), we have:
\begin{align*}
\bby_t = \sum_{i = 1}^{r_1} \bbb_i k_{it}^{(1)} + \sum_{i = 1}^{r_2} \bba_i k_{it}^{(2)} + \bvar_t.
\end{align*}
We can choose $r_1$ and $r_2$ such that $r_1 + r_2 < P$ and $\mathbb{E}(\|\bvar_t^{\top}\bvar_t\|)$ is small enough.

Replace the matrices with their sample version, we get the sample estimation of the model
\begin{align}
\widetilde{\bby}_t = \sum_{i = 1}^{\widehat{r}_1} \widehat{\bbb}_i \widehat{k}_{it}^{(1)} + \sum_{i = 1}^{\widehat{r}_2} \widehat{\bba}_i \widehat{k}_{it}^{(2)}, \quad t = 1, \dots, T \label{dpca:eq:ex4}
\end{align}
which is a low-dimensional representation of the original data via two-step dimension reduction.
 
\subsection{Forecasting}
Recall that after the two-step dimension reduction, we get the estimation (\ref{dpca:eq:ex4}). Following \citet{lee1992modeling}, we can forecast $\bby_{T+h}$ by forecasting the features $k_{i,T+h}^{(1)}$ and $k_{i,T+h}^{(2)}$first. In order to get the forecasts $\widehat{k}_{i,T+h}^{(1)}$ and $\widehat{k}_{i,T+h}^{(2)}$, we model $\{\widehat{k}_{it}^{(1)}: i  = 1,2,\dots, r_1\}_{t = 1,2,\dots, T}$ and $\{\widehat{k}_{it}^{(2)}:i = 1,2,\dots,r_2\}_{t=1,2,\dots,T}$ with standard time series models and conduct $h$-steps ahead forecasting with them.
Then together with (\ref{dpca:eq:ex4}), the $h$-steps ahead forecasting for $\bby_{T+h}$ is 
\begin{align*}
\widetilde{\bby}_{T+h} = \sum_{i = 1}^{r_1} \widehat{\bbb}_i \widehat{k}_{i,T+h}^{(1)} + \sum_{i = 1}^{r_2} \widehat{\bba}_i \widehat{k}_{i,T+h}^{(2)},
\end{align*}
where $\widehat{k}_{i,T+h}^{(1)}$ and $\widehat{k}_{i,T+h}^{(2)}$ are predicted values of the features in $h$ years after time $T$, $h = 1,2, \dots$. 

Consequently, instead of conducting $P$ forecasting models, we only need $\widehat{r}_1 + \widehat{r}_2 < P$ forecasting models.
In our simulations and application on the US mortality data, we choose $ARIMA(p, d, q)$ models to forecast the time series, and we use BIC to choose the parameters $p, d, q$ for each model. 

\subsection{Practical Algorithm}

The practical procedure for the two-steps dimension reduction and forecasting is summarized in Algorithm \ref{Algorithm}.
\begin{algorithm}
\SetAlgoLined
\KwIn{Data $\bbY = [\bby_1, \dots, \bby_T] \in \mathbb{R}^{P \times T}$; Desired rank $\le P$.}
\KwOut{Low-dimensional representation of $\bbY$; $h-$steps ahead forecasts of $\bby_T$, h = 1,2, \dots.}
\caption{{\bf FHFM  for Mortality Forecasting} \label{Algorithm}}
\textbf{Dimension Reduction Step 1}:\\
\nl Compute the sample mean $\overline{\bby} = T^{-1}\stT \bby_t$\;
\nl Compute the sample auto-covariance matrix $\widehat{\Sig}_{\bby}(1) = \frac{1}{T-1} \sum_{t=1}^{T-1} (\bby_{t+1} - \overline{\bby})(\bby_{t} - \overline{\bby})^{\top}$\;
\nl Compute sample matrix for the first step $\widehat{\bbL}_1 =  \widehat{\Sig}_{\bby}(1) \widehat{\Sig}_{\bby}(1)^{\top}$\;
\nl Conduct eigendecomposition on $\widehat{\bbL}_1$ and get $\widehat{\bbb}_1, \dots, \widehat{\bbb}_{\widehat{r_1}}$, the eigenvectors corresponding to the largest $\widehat{r_1}$ eigenvalues of $\widehat{\bbL}_1$\;
\nl Compute the first sets of features $\widehat{k}_{it}^{(1)} = \widehat{\bbb}_i^{\top}(\bby_t- \overline{\bby}), i = 1, \dots, \widehat{r_1}, t = 1,\dots, T$\;
\textbf{Dimension Reduction Step 2}:\\
\nl Compute $\widehat{\bbu}_t = (\bby_t - \overline{\bby}) - \sum_{i = 1}^{\widehat{r}_1} \widehat{\bbb}_i \widehat{k}_{it}^{(1)}$\;
\nl Compute sample the covariance matrix of $\widehat{\bbu}_t$, $\widehat{\Sig}_{\bbu}(0) = \frac{1}{T} \sum_{t=1}^{T} \widehat{\bbu}_t\widehat{\bbu}_t^{\top}$\;
\nl Compute sample matrix for second step $\widehat{\bbL}_2 = \widehat{\Sig}_{\bbu}(0)\widehat{\Sig}_{\bbu}(0)^{\top}$\;
\nl Conduct eigendecomposition on $\widehat{\bbL}_2$ and get $\widehat{\bba}_1, \dots, \widehat{\bba}_{\widehat{r_2}}$, the eigenvectors corresponding to the largest $\widehat{r_2}$ eigenvalues of $\widehat{\bbL}_2$\;
\nl Compute the second sets of features $\widehat{k}_{it}^{(2)} = \widehat{\bba}_i^{\top}\widehat{\bbu}_t, i = 1, \dots, \widehat{r_2}, t = 1,\dots, T$\;
\textbf{Estimation result}:\\
\nl Compute $\widehat{\bby}_t =  \overline{\bby} + \sum_{i = 1}^{\widehat{r}_1} \widehat{\bbb}_i \widehat{k}_{it}^{(1)} + \sum_{i = 1}^{\widehat{r}_2} \widehat{\bba}_i \widehat{k}_{it}^{(2)}, t = 1, \dots, T,$ and the estimated low-dimensional representation of $\bbY$ is $\widehat{\bbY} = [\widehat{\bby}_1, \dots, \widehat{\bby}_T]$\;
\textbf{Forecasting Step}:\\
\nl Fit $\widehat{k}_{it}^{(1)}, i = 1, \dots, \widehat{r_1}, t = 1,\dots, T$ and $\widehat{k}_{it}^{(2)}, i = 1, \dots, \widehat{r_2}, t = 1,\dots, T$ with standard ARIMA(p,d,q) models respectively\;
\nl Compute $\widehat{k}_{i,T+h}^{(1)}, \widehat{k}_{j,T+h}^{(2)}$, the $h-$step ahead forecasts of the features, with the fitted ARIMA models; $i = 1, \dots, \widehat{r}_1$; $j  = 1, \dots, \widehat{r}_2$\; 
\nl Compute the $h-$steps ahead forecasts after $\bby_T$ by $\widehat{\bby}_{T+h} =  \overline{\bby} + \sum_{i = 1}^{\widehat{r}_1} \widehat{\bbb}_i \widehat{k}_{i,T+h}^{(1)} + \sum_{i = 1}^{\widehat{r}_2} \widehat{\bba}_i \widehat{k}_{i,T+h}^{(2)}$.
\end{algorithm}

In our simulations and analysis of the US mortality data, we estimate the values of $r_1$ and $r
_2$ by the criterion,
\begin{equation}\label{rts}
\widehat{r} = \argmin_{1\le i \le R} \frac{\widehat{\lambda}_{i+1}}{\widehat{\lambda}_i},
\end{equation} 
where $\widehat{\lambda}_i, i = 1, 2, \dots, R$ are the eigenvalues of $\widehat{\bbL}_1$ or $\widehat{\bbL}_2$ in descending order, and $\max(r_1, r_2)<R<P$. This ratio criterion is commonly used in literature, like \cite{LY2012} and \cite{AH2013} for auto-covariance matrices and covariance matrices on high dimensional data, respectively. 
As mentioned in \cite{LY2012}, in practice, the parameter $R$ is chosen as $\frac{1}{2}\min (P, T)$. It is worthy being mentioned that the number of nonzero eigenvalues of the sample matrices $\widehat{\bbL}_1$ and $\widehat{\bbL}_2$ is no larger than $\min(P, T)$.

\section{Relationship with Existing Methods}
\label{dpca:relation}
The methods which forecast mortality via features' forecasting date back to the Lee-Carter model (\citenum{lee1992modeling}). For general comparison, consider the following one factor model
\begin{align*}
y_{x,t} = \ln(m_{x,t}) = a_x + b_x k_t + u_{x,t},
\end{align*}
where $a_x$ is a constant for each $x$, $k_t$ is an unobserved time series (the feature summarizing the original high-dimensional time series), $b_x$ is the loading of the feature $k_t$ to each age $x$, and $u_{x,t}$ is the error term. One can recover the $h-$steps ahead forecasting of $y_{x,t}$ via the forecasting of $k_t$. Therefore, how we extract the feature $k_t$ by dimension reduction is the main difference.
\subsection{Static PCA Method}


The most popular method for mortality modelling, the Lee-Carter model, utilizes static PCA to estimate $k_t$. The extracted feature $k_t$, which performs the most important role in the forecasting, is the first principal component that explains the most of the variance of the original data. 
Mathematically, the static PCA solves the following objective
\begin{align*}
\max_{\bbb} \text{var}{(\bbb^{\top}\bby_t)}
\end{align*}
to get $\widetilde{k}_t = \widetilde{\bbb}^{\top}(\bby_t - \bar{\bby}),\ t = 1, \dots, T$. This solution has the smallest average squared reconstruction error $\mathbb{E}(\|\bbu_t^{\top}\bbu\|^2)$ \citep{BL1975}. It, however, does not seem to have enough forecasting ability. For example, suppose $\bby_t = (y_{1t}, y_{2t}, \dots, y_{pt})^{\top},\ t = 1,2 \dots, T$, where $y_{1t}$ and $y_{2t}$ have huge variances and very weak time serial dependence, while the rest have small variances but strong time serial dependence. Performing static PCA on $\bby_t$ will get a feature which puts most of the loading on $y_{1t}$ and $y_{2t}$. The forecasting based on this feature will heavily depend on the pattern of $y_{1t}$ and $y_{2t}$ and not make use of the strong serial dependence information contained in the rest, which may lead to a misleading forecasting.

On the other hand, the first step of our proposed method extracts $k_t$ from the auto-covariance matrix, which contains most of the time-serial dependence information. Thus it is expected to have stronger forecasting ability comparing to the feature extracted by the static PCA.
  
\subsection{Dynamic PCA Method}

There are several dynamic PCA methods, which usually involve the auto-covariance matrices and also utilize the time-serial dependence information, including \citet{BL1975}, \citet{LYB2011}, \citet{HKH2015}, and \citet{chang2018principal}. Those methods can be used to extract feature $k_t$ as well.
For comparison purpose with our method, we consider one of them described in the following. 

Define $\Sig_{\bby}(\ell) = \text{cov}(\bby_{t}, \bby_{t+\ell}),\ \ell = 0, 1, 2, \dots$, and consider the nonnegative definite matrix
\begin{align}
\bbL = \sum_{\ell = 0}^{\ell_0} \Sig_{\bby}(\ell) \Sig_{\bby}(\ell)^{\top}. 
\label{dpca:eq:ll}
\end{align}
With matrix (\ref{dpca:eq:ll}), the coefficients of the feature $k_t$ can be estimated by the eigenvector of the sample matrix $\widehat{\bbL}$ corresponding to its largest eigenvalue.
This is similar to \citet{BL1975} and \citet{HKH2015} by assigning different weights on those covariances, while in \citet{LYB2011}, they exclude $\Sig_{\bby}(0)$.
If $\ell_0 = 0$, it is the same with the static PCA. If $\ell_0 = 1$, $\bbL$ can be seen as the mixing of the two steps in our method. If $\ell_0 > 1$, $\bbL$ aggregates more lagged covariances than our method. 

There are similarities and advantages of our method compared to the aforementioned dynamic PCA methods. On one hand, the first step of our method is motivated by the dynamic PCA that auto-covariance matrices are used to obtain forecasting ability for the features. While from the empirical and simulating studies, we find the lag 1 auto-covariance is enough for the mortality data and data with similar structure. To make the method simple and easy to apply, our method only involves the most useful auto-covariance. If the data structure changes, one can include more lagged covariances which adapts to the data. On the other hand, we intend to maximize the forecasting ability instead of balancing several characteristics of the features. The dynamic PCA provides only one set of features which mix the information of the temporal trend and the variation, while our proposed method extracts two types of features separately via a two-steps procedure. The first type of features represent the temporal trend, which benefits the forecasting, and the second type ones capture variations which are good for model fitting as well as the forecasting. 
The features are linear combinations of the original data and the coefficients of them represent the directions which the original data are projected to. Therefore, we can visualize the features by the directions. 
Figure \ref{dpca:fg:fs} shows a simple example of the estimated directions for our method and the dynamic PCA. Data is generated the same as in \textit{Example} 4 described in Appendix \ref{append:cha:stepPCA.1} with $P = 3$ and $T = 20$ . The red arrow is the direction of the feature for the dynamic PCA (DPCA), and the blue ones are those for the first step and the second step of our method (FHFM), respectively. It is clear that they project the data into different directions and the red one can be seen as a direction which mixes the other two.
\begin{figure} 
\centering
\includegraphics[width=0.5\linewidth]{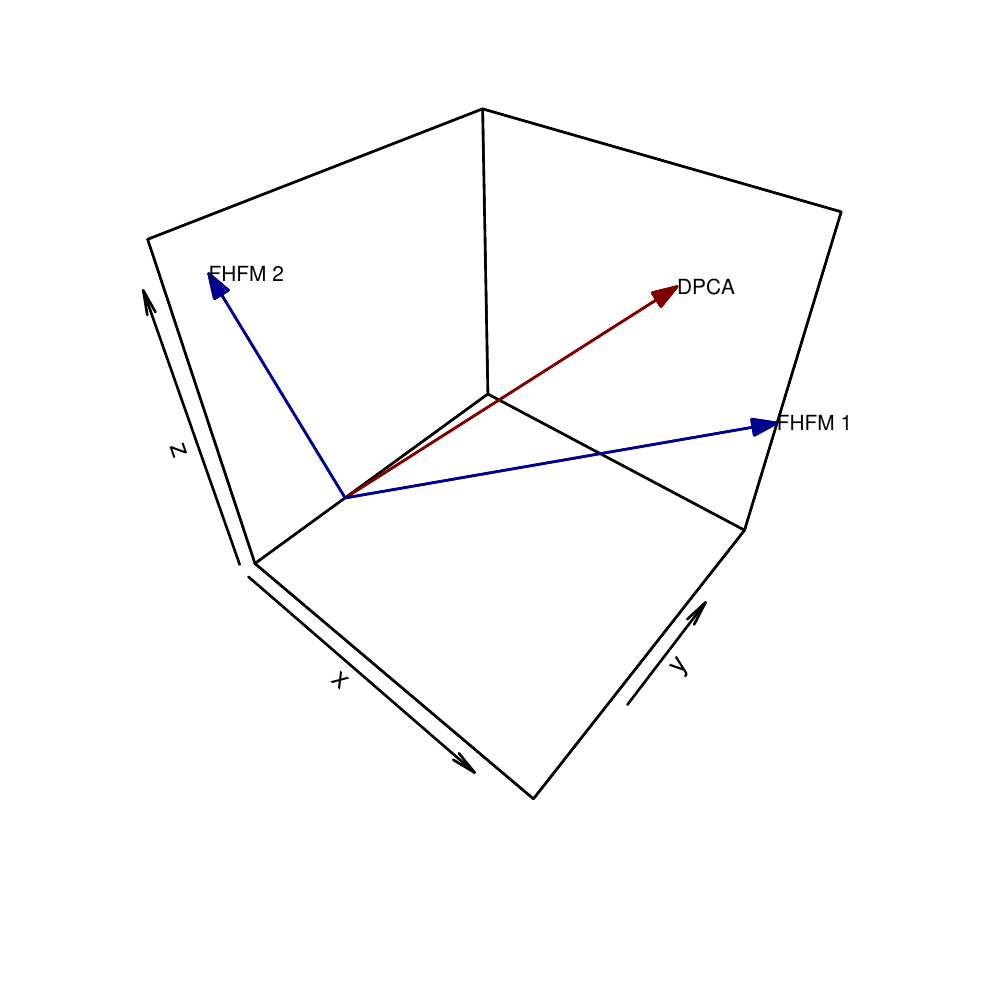}
\caption{The directions of the features; A comparison between the dynamic PCA and FHFM.}
\label{dpca:fg:fs}
\end{figure}
In addition, if we take a detailed look at the matrix $\bbL$, we find that it is actually hard to tell and explain what information is contained in this matrix, as it is a mix of several auto-covariances and the variance matrix. Our method, on the other hand, is stepwise, thus it has a clear goal for each step. 

\section{Asymptotic Properties}
\label{dpca:asymptotic}
In this section, we establish the rates of convergence for the two-steps estimators of the factor loadings. Additional to the Assumptions \ref{assu1}-\ref{assu4}, we impose the following assumptions for the asymptotic theory. 
\begin{assu}\label{assu5}
{\color{blue}{Relation between the error and the factors}}. $\boldsymbol{\varepsilon}_t$ independent of $\bbk_{t}^{(1)}$ and $\bbk_{t}^{(2)}$. 
\end{assu}
\begin{remark}
For simplicity of techniques and without loss of generality, the Assumption \ref{assu5} assumes independent relationship between the error component and the two kinds of factors. 
\end{remark}
\begin{assu}\label{assu6}
{\color{blue}{Relation between $\bbk_{t}^{(1)}$ and $\bbk_{t}^{(2)}$}}. Suppose that $\left|\left|\boldsymbol{\Sigma}_{k}^{(21)}(\ell)\right|\right|\asymp\left|\left|\boldsymbol{\Sigma}_{k}^{(12)}(\ell)\right|\right|_{\min}$, $\left|\left|\boldsymbol{\Sigma}_{k}^{(12)}(\ell)\right|\right|=O\left(P^{1-\frac{\delta_2}{2}}\right)$, where $\boldsymbol{\Sigma}_k^{(21)}(\ell)=cov\left(\bbk_{t+\ell}^{(2)}, \bbk_{t}^{(1)}\right)$, $\boldsymbol{\Sigma}_k^{(12)}(\ell)=cov\left(\bbk_{t+\ell}^{(1)}, \bbk_{t}^{(2)}\right)$ and $\delta_2$ is defined in Assumption \ref{assu2}. 
\end{assu}
\begin{remark}
The order of the eigenvalues of $\boldsymbol{\Sigma}_k^{(21)}(\ell)$ is not specified in Assumption \ref{assu6}. The reason is that the information involved in $\boldsymbol{\Sigma}_k^{(21)}$ participate in the recovery of the factor $\bbk^{(1)}_t$. 
The order of $\left|\left|\boldsymbol{\Sigma}_{k}^{12}(\ell)\right|\right|$ is restricted in order to make it not involved in the leading term when recovering $\bbk_t^{(1)}$. 
\end{remark}
\begin{assu} \label{assu7}
{\color{blue}{Dimension Condition}}. $\frac{P}{T}\rightarrow c\in (0, \infty)$. 
\end{assu}
\begin{remark}
The setting of the dimension $P$ and the sample size $T$ being comparable is under consideration because the number of ages is comparable to the length of time series for the US mortality data. Note that when $P$ and $T$ are on the same order, the estimators for the eigenvalues and the eigenvectors may be no longer consistent. See \cite{LY2012}, \cite{AH2013}. However, the ratio based estimators for $r_1$ and $r_2$ can still work well.   
\end{remark}
\begin{assu}\label{assu8}
$\left\{\left(\bbk_t^{(1)}, \bbk_t^{(2)}, \boldsymbol{\varepsilon}_t\right): t\geq 1\right\}$ is strictly stationary with finite fourth moments.   
\end{assu}

\begin{thm}\label{thm1}
In addition to Assumptions \ref{assu1} - \ref{assu8}, we assume that 
\begin{eqnarray}
\frac{P^{1-\delta_1}}{T}=o(1), \ \ as \ \ P, T\rightarrow\infty. 
\end{eqnarray}
Then we have the following convergent rates 
\begin{eqnarray}
\left|\left|\widehat{\bbB}-\bbB\right|\right|=O_p\left(\frac{1}{T^{1/2}}\right), \ \ 
\left|\left|\widehat{\bbA}-\bbA\right|\right|=O_p\left(\frac{1}{T^{1/2}}\right). 
\end{eqnarray}
\end{thm}
\begin{remark}
Two kinds of factors are strong factors in the sense of auto-correlation and variance, respectively. It is reasonable to obtain fast rates of convergence for both of them. In view of this, 
our proposed two-steps estimators have good statistical performance, which is an advantage for forecasting improvement. Based on identification condition between factors and factor loadings, $\bbB$ is determined up to an orthogonal matrix. 
Due to technical proofs (some techniques in Lemma \ref{lem2}), the estimator $\widehat{\bbB}$ here is the estimator up to an identity matrix.  
\end{remark}

\begin{thm}\label{thm2}
In addition to assumptions in Theorem \ref{thm1}, we denote the one-step-ahead forecasting error of the common factors as
\begin{eqnarray*}
\left|\left|\widehat{\bbk}^{(j)}_{t-1}(1)-\bbk_t^{(j)}\right|\right|=O_p\left(\gamma_{PT}^{(j)}\right), \ \ \ j=1, 2, 
\end{eqnarray*}
where $\widehat{\bbk}^{(j)}_{t-1}$ is the forecasting for the common factor $\bbk^{(j)}_t$ at time $t$. 

Then we have
\begin{eqnarray}\label{y1105(3)}
\frac{1}{P}\left|\left|\widehat{\bby}_{t-1}(1)-\bby_t\right|\right|
=O_p\left(\max\left\{\frac{1}{\sqrt{T}}, \gamma_{PT}^{(1)}, \gamma_{PT}^{(2)}, \theta_{PT}\right\}\right), 
\end{eqnarray}
where $\theta_{PT}=\frac{\sum^{P}_{i=r_2+1}\lambda_i(\boldsymbol{\Sigma}_u(0))}{\sum^{P}_{i=1}\lambda_i(\boldsymbol{\Sigma}_u(0))}$. 
\end{thm}
\begin{remark}
As the one-step-ahead forecasting for $\bby_t$, the forecasting error of $\widehat{\bby}_{t-1}$ is influenced by three aspects: the estimation error $\frac{1}{\sqrt{T}}$, the forecasting error from the forecast factors $\gamma_{PT}^{(j)}$, $j=1,2$, and the model error $\theta_{PT}$ which comes from the error component $\boldsymbol{\varepsilon}_t$ in (\ref{dpca:eq:y2}). 
\end{remark}

\section{Simulations}
\label{dpca:simu}
In this section, we use simulated data to illustrate the advantages of our method. For descriptive convenience, we use ``FHFM'' to represent our method, ``CPCA'' to represent the static PCA method which was described in Section \ref{dpca:relation}, and ``DPCA'' to represent the dynamic PCA method described in Section \ref{dpca:relation} with $\ell_0 = 1$.

For all the three examples, we first examine the variance and serial dependence (lag $1$ auto-covariance) of the first estimated factor by the three methods, respectively. Secondly, we evaluate the serial dependence and the variations remained in the error terms. Finally, we compare the forecasting performance for the $1$ step and $5$ steps ahead forecasting with the root mean squared forecasting error (FRMSE). 

We show that our method extracts the feature with the largest auto-covariance and leaves the least information in the error terms. As a result, our method provides the best forecasting results for all the three examples. The details of the simulations are described in the rest of this section. More simulation studies which are special cases can be find in Appendix \ref{append:cha:stepPCA.1}.

\subsection{Data Generating Processes}

We generate three examples according to the following two-factors model:
\begin{eqnarray*}
\bby_t = \bbb k_{t} + \bba w_{t} + \bvar_t,
\end{eqnarray*}
where $\bba$ and $\bbb$ are two independent $P \times 1$ vectors with elements generated from a uniform distribution $U(0, 1)$ and $\bvar_t$ is a $P \times 1$ error term with elements independently generated from a normal distribution $N(0, 0.2^2)$. For all the three examples, $\{k_t\}_{t =1,2, \dots, T}$ is generated from $AR(1)$ model with coefficient $0.8$, while $\{w_t\}_{t =1,2, \dots, T}$ are different for each example. 
\begin{itemize}[noitemsep]
	\item For \textit{Example 1}, elements in $\{w_t\}_{t =1,2, \dots, T}$ are independently generated from standard normal distribution $N(0, 1)$, which indicates the series of $w_t$ are independent; 

	\item For \textit{Example 2}, we add time-serial dependence to the feature $w_t$, hence $\{w_t\}_{t =1,2, \dots, T}$ is generated from $AR(1)$ model with coefficient $0.05$; 

	\item At last in \textit{Example 3}, we increase the dependence in the series of $w_t$ and generate it from $AR(1)$ model with coefficient $0.2$.
\end{itemize}

\subsection{Performance Evaluation Criterion}

Firstly, we show the variance and serial dependence (lag $1$ auto-covariance) of the first estimated factor of the three methods, respectively. The variance and serial dependence of the first estimated factor are computed as follows:
\begin{eqnarray*}
\text{Time variance}\left(\widehat{k}_t\right) &= \frac{1}{T-1} \sum_{t =1} ^{T} \left(\widehat{k}_t - \frac{1}{T}\sum_{j = 1}^T \widehat{k}_j\right)^2 ,
\end{eqnarray*}
\begin{eqnarray*}
\text{Time dependence}\left(\widehat{k}_t\right) &=  \frac{1}{T-2} \sum_{t =1} ^{T-1} \left(\widehat{k}_t - \frac{1}{T}\sum_{j = 1}^T \widehat{k}_j\right)\left(\widehat{k}_{t+1} - \frac{1}{T}\sum_{j = 1}^T \widehat{k}_j\right),
\end{eqnarray*}
where $\widehat{k_t}$ is the estimated first feature at time $t$. Especially, for our method, we compare the estimated first feature from the first step as it is the feature which intends to improve the forecasting power. 
Besides, we also report the sum of the aforementioned quantities:
\begin{eqnarray*}
\text{Mix}(\widehat{k}_t) =  \text{Time variance}(\widehat{k}_t) + \text{Time dependence}(\widehat{k}_t).
\end{eqnarray*}

Secondly, we investigate the dependence and the variation remained in the error terms as follows:
\begin{eqnarray*}
\text{Time variance}\left(\widehat{\bvar}_{\cdot t}\right) &= \frac{1}{P} \sum_{p =1} ^{P} \left(\frac{1}{T-1} \sum_{t =1} ^{T} \left(\widehat{\varepsilon}_{pt} - \frac{1}{T}\sum_{j = 1}^T \widehat{\varepsilon}_{pj}\right)^2\right),
\end{eqnarray*}
\begin{eqnarray*}
\text{Time dependence}\left(\widehat{\bvar}_{\cdot t}\right) &=  \frac{1}{T(T-1)} \sum_{t_1 =1} ^{T} \sum_{t_2 = 1, t_2 \ne t_1}^{T} \left|cov\left(\widehat{\bvar}_{\cdot t_1}, \widehat{\bvar}_{\cdot t_2}\right)\right|,
\end{eqnarray*}
\begin{eqnarray*}
\text{Cross-sectional variance}\left(\widehat{\bvar}_{p\cdot}\right) &= \frac{1}{T} \sum_{t =1} ^{T} \left(\frac{1}{P-1} \sum_{p =1} ^{P} \left(\widehat{\varepsilon}_{pt} - \frac{1}{P}\sum_{j = 1}^P \widehat{\varepsilon}_{jt}\right)^2\right),
\end{eqnarray*}
\begin{eqnarray*}
\text{Cross-sectional dependence}\left(\widehat{\bvar}_{p\cdot}\right) &=  \frac{1}{P(P-1)} \sum_{p_1 =1} ^{P} \sum_{p_2 = 1, p_2 \ne p_1}^{P} \left|cov\left(\widehat{\bvar}_{p_1 \cdot}, \widehat{\bvar}_{p_2 \cdot}\right)\right|,
\end{eqnarray*} 
where $\widehat{\varepsilon}_{pt}$ is the the error term for age $p$ at time $t$, $\widehat{\bvar}_{\cdot t}$ is error terms for all ages at time $t$, $\widehat{\bvar}_{p\cdot}$ is the error terms across all time for age $p$, and 
 \begin{eqnarray*}
 cov\left(\widehat{\bvar}_{\cdot t_1}, \widehat{\bvar}_{\cdot t_2}\right) &= \frac{1}{P} \sum_{p =1} ^{P} \left(\widehat{\varepsilon}_{pt_1} - \frac{1}{P}\sum_{j = 1}^P \widehat{\varepsilon}_{jt_1}\right) \left(\widehat{\varepsilon}_{pt_2} - \frac{1}{P}\sum_{j = 1}^P \widehat{\varepsilon}_{jt_2}\right) ,
 \end{eqnarray*}
 \begin{eqnarray*}
 cov\left(\widehat{\bvar}_{p_1\cdot}, \widehat{\bvar}_{p_2\cdot}\right) &= \frac{1}{T} \sum_{t =1} ^{T} \left(\widehat{\varepsilon}_{p_1t} - \frac{1}{T}\sum_{j = 1}^T \widehat{\varepsilon}_{p_1j}\right) \left(\widehat{\varepsilon}_{p_2t} - \frac{1}{T}\sum_{j = 1}^T \widehat{\varepsilon}_{p_2j}\right).
 \end{eqnarray*}
To evaluate the forecasting performance, we show the the $1$ step and $5$ steps ahead root mean squared error, which is computed by 
\begin{eqnarray*}
\text{FRMSE}(h) = \left(\frac{\sum_{i = 0}^{h-1}\|\widehat{\bby}_{T-i} - \bby_{T-i}\|_2^2}{hP}\right)^{1/2}
\end{eqnarray*}
where $h = 1,5$ (the forecasting length), $\widehat{\bby}_{T-i}$ is obtained by forecasting with $\{\bby_1, \bby_2, \dots, \bby_{T-h}\}$, and $\bby_{T-i}$ is the true value in the forecasting horizon.

\subsection{Simulation Results}
We try different sets of $(P, T)$: $(50, 50),\ (50, 100),\ (100, 100),\ (100, 200),\ (200,200)$, as we would like to evaluate the performance under the situations that $P$ and $T$ are comparable. The results are shown in Table \ref{dpca:ts.kt} to Table \ref{dpca:ts.forecast}.

\begin{table}[!htbp] 
  \centering 
  \footnotesize
  \caption{Variance and Dependence of $\widehat{k}_t$} 
  \label{dpca:ts.kt} 
\begin{tabularx}{\textwidth}{c *{9}{Y}}
\toprule
 & \multicolumn{3}{c}{Time variance ($\widehat{k_t}$)} 
 & \multicolumn{3}{c}{Time dependence ($\widehat{k_t}$)}
 & \multicolumn{3}{c}{Mix ($\widehat{k_t}$)}\\
\cmidrule(lr){2-4} \cmidrule(lr){5-7} \cmidrule(l){8-10}
 $(P, T)$ & CPCA & DPCA & FHFM & CPCA & DPCA & FHFM & CPCA & DPCA & FHFM \\
\midrule
& \multicolumn{9}{c}{Example 1 (AR(1) 0.8 + N(0,1))}\\
\midrule
$(50, 50)$ & $\bf{51.102}$ & $51.008$ & $48.750$ & $29.015$ & $29.437$ & $\bf{30.174}$ & $80.117$ & $\bf{80.445}$ & $78.923$ \\ 
$(50, 100)$ & $\bf{53.436}$ & $53.330$ & $51.577$ & $31.843$ & $32.304$ & $\bf{32.938}$ & $85.279$ & $\bf{85.635}$ & $84.515$ \\ 
$(100, 100)$ & $\bf{107.483}$ & $107.263$ & $103.799$ & $64.119$ & $65.070$ & $\bf{66.341}$ & $171.601$ & $\bf{172.333}$ & $170.139$ \\ 
$(100, 200)$ & $\bf{110.269}$ & $110.037$ & $107.091$ & $67.532$ & $68.517$ & $\bf{69.651}$ & $177.801$ & $\bf{178.554}$ & $176.742$ \\  
$(200, 200)$ & $\bf{221.091}$ & $220.619$ & $214.682$ & $135.760$ & $137.762$ & $\bf{140.053}$ & $356.851$ & $\bf{358.381}$ & $354.735$ \\ 
\midrule
& \multicolumn{9}{c}{Example 2 (AR(1) 0.8 + AR(1) 0.05)}\\
\midrule
$(50, 50)$ & $\bf{51.000}$ & $50.909$ & $48.941$ & $29.409$ & $29.806$ & $\bf{30.429}$ & $80.409$ & $\bf{80.715}$ & $79.371$ \\ 
$(50, 100)$ & $\bf{53.085}$ & $52.986$ & $51.430$ & $31.958$ & $32.382$ & $\bf{32.940}$ & $85.043$ & $\bf{85.368}$ & $84.371$ \\ 
$(100, 100)$ & $\bf{107.666}$ & $107.466$ & $104.384$ & $65.591$ & $66.440$ & $\bf{67.541}$ & $173.257$ & $\bf{173.906}$ & $171.925$ \\  
$(100, 200)$ & $\bf{110.047}$ & $109.838$ & $107.249$ & $68.630$ & $69.497$ & $\bf{70.471}$ & $178.677$ & $\bf{179.335}$ & $177.719$ \\  
$(200, 200)$ & $\bf{221.705}$ & $221.278$ & $216.268$ & $139.463$ & $141.213$ & $\bf{143.102}$ & $361.168$ & $\bf{362.492}$ & $359.369$ \\
\midrule
& \multicolumn{9}{c}{Example 3 (AR(1) 0.8 + AR(1) 0.2))}\\
\midrule
$(50, 50)$ & $\bf{51.572}$ & $51.498$ & $50.033$ & $31.392$ & $31.685$ & $\bf{32.093}$ & $82.964$ & $\bf{83.183}$ & $82.126$ \\ 
$(50, 100)$ & $\bf{54.963}$ & $54.889$ & $53.942$ & $35.265$ & $35.556$ & $\bf{35.873}$ & $90.228$ & $\bf{90.445}$ & $89.815$ \\ 
$(100, 100)$ & $\bf{108.933}$ & $108.778$ & $106.789$ & $69.386$ & $69.996$ & $\bf{70.666}$ & $178.319$ & $\bf{178.774}$ & $177.455$ \\ 
$(100, 200)$ & $\bf{110.179}$ & $110.024$ & $108.342$ & $71.475$ & $72.086$ & $\bf{72.697}$ & $181.655$ & $\bf{182.110}$ & $181.039$ \\ 
$(200, 200)$ & $\bf{224.522}$ & $224.200$ & $220.766$ & $147.018$ & $148.273$ & $\bf{149.494}$ & $371.541$ & $\bf{372.472}$ & $370.260$ \\ 
\bottomrule
\end{tabularx}
\end{table}
From Table \ref{dpca:ts.kt}, we can see that the CPCA method provides feature with the largest variance, while the first step of our method (FHFM) captures the feature with the largest lag $1$ auto-covariance, which summarizes the most of the time serial dependence of the original data. These simulated results corroborate the analysis in Section \ref{dpca:relation}.

On the other hand, we can see that the DPCA method provides a feature with the largest sum of variance and lag $1$ auto-covariance. Our method utilizes the same information with the DPCA, while we have two steps. When we compare the feature from our first step with that of the DPCA method, it is not surprising that DPCA one has larger Mix$(\widehat{k}_t)$. However, our second step provides features that capture the remaining variance, which is a necessary supplement for the first step to ensure that the final sets of features provide good fitting to the original data.

From Table \ref{dpca:ts.error} and \ref{dpca:ts.error.depend}, we can see that our method always provides the error terms with the smallest time and cross-sectional variance and dependence. It shows that FHFM can capture most of the time-serial dependence and variation information of all the ages among the three methods. The better model fitting performance of our method is supported by these results. 
\begin{table}[!htbp] \centering 
  \footnotesize
  \caption{Variance across Time and Sections of the error terms} 
  \label{dpca:ts.error} 
\begin{tabularx}{\textwidth}{c *{6}{Y}}
\toprule
 & \multicolumn{3}{c}{Time Variance ($\widehat{\bvar}_{\cdot t}$)} 
 & \multicolumn{3}{c}{Cross-sectional Variance ($\widehat{\bvar}_{p\cdot}$)}\\
\cmidrule(lr){2-4} \cmidrule(l){5-7}
 $(P, T)$ & CPCA & DPCA & FHFM & CPCA & DPCA & FHFM \\
\midrule
& \multicolumn{6}{c}{Example 1 (AR(1) 0.8 + N(0,1))}\\
\midrule
$(50, 50)$ & $0.141$ & $0.142$ & $\bf{0.038}$ & $0.145$ & $0.146$ & $\bf{0.038}$ \\ 
$(50, 100)$ & $0.146$ & $0.146$ & $\bf{0.038}$ & $0.148$ & $0.151$ & $\bf{0.038}$   \\ 
$(100, 100)$ & $0.147$ & $0.147$ & $\bf{0.039}$ & $0.151$ & $0.153$ & $\bf{0.039}$  \\ 
$(100, 200)$ & $0.151$ & $0.151$ & $\bf{0.039}$ & $0.154$ & $0.156$ & $\bf{0.039}$ \\ 
$(200, 200)$ & $0.152$ & $0.152$ & $\bf{0.039}$ & $0.156$ & $0.158$ & $\bf{0.039}$ \\
\midrule
& \multicolumn{6}{c}{Example 2 (AR(1) 0.8 + AR(1) 0.05)}\\
\midrule
$(50, 50)$ & $0.142$ & $0.142$ & $\bf{0.038}$ & $0.145$ & $0.147$ & $\bf{0.038}$ \\ 
$(50, 100)$ & $0.148$ & $0.148$ & $\bf{0.038}$ & $0.150$ & $0.152$ & $\bf{0.038}$  \\ 
$(100, 100)$ & $0.147$ & $0.148$ & $\bf{0.039}$ & $0.151$ & $0.153$ & $\bf{0.039}$  \\ 
$(100, 200)$ & $0.150$ & $0.150$ & $\bf{0.039}$ & $0.153$ & $0.155$ & $\bf{0.039}$ \\ 
$(200, 200)$ & $0.152$ & $0.152$ & $\bf{0.039}$ & $0.155$ & $0.158$ & $\bf{0.039}$ \\
\midrule
& \multicolumn{6}{c}{Example 3 (AR(1) 0.8 + AR(1) 0.2))}\\
\midrule
$(50, 50)$ & $0.144$ & $0.144$ & $\bf{0.038}$ & $0.148$ & $0.149$ & $\bf{0.038}$ \\ 
$(50, 100)$ & $0.150$ & $0.150$ & $\bf{0.038}$ & $0.152$ & $0.153$ & $\bf{0.038}$  \\ 
$(100, 100)$ & $0.151$ & $0.151$ & $\bf{0.039}$ & $0.154$ & $0.156$ & $\bf{0.039}$  \\ 
$(100, 200)$ & $0.154$ & $0.154$ & $\bf{0.039}$ & $0.157$ & $0.159$ & $\bf{0.039}$ \\ 
$(200, 200)$ & $0.155$ & $0.155$ & $\bf{0.039}$ & $0.159$ & $0.161$ & $\bf{0.039}$ \\ 
\bottomrule
\end{tabularx}
\end{table}

\begin{table}[!htbp] \centering 
  \footnotesize
  \caption{Covariance across Time and Sections of error terms} 
  \label{dpca:ts.error.depend} 
\begin{tabularx}{\textwidth}{c *{6}{Y}}
\toprule
 & \multicolumn{3}{c}{Time dependence ($\widehat{\bvar}_{\cdot t}$)} 
 & \multicolumn{3}{c}{Cross-sectional dependence ($\widehat{\bvar}_{p\cdot}$)}\\
\cmidrule(lr){2-4} \cmidrule(l){5-7}
 $(P, T)$ & CPCA & DPCA & FHFM & CPCA & DPCA & FHFM \\
\midrule
& \multicolumn{6}{c}{Example 1 (AR(1) 0.8 + N(0,1))}\\
\midrule
$(50, 50)$ & $0.067$ & $0.067$ & $\bf{0.005}$ & $0.072$ & $0.073$ & $\bf{0.005}$  \\ 
$(50, 100)$ & $0.069$ & $0.069$ & $\bf{0.004}$ & $0.074$ & $0.075$ & $\bf{0.003}$ \\ 
$(100, 100)$ & $0.069$ & $0.069$ & $\bf{0.003}$ & $0.075$ & $0.076$ & $\bf{0.003}$ \\ 
$(100, 200)$ & $0.071$ & $0.071$ & $\bf{0.003}$ & $0.076$ & $0.078$ & $\bf{0.002}$ \\ 
$(200, 200)$ & $0.072$ & $0.072$ & $\bf{0.002}$ & $0.078$ & $0.079$ & $\bf{0.002}$ \\ 
\midrule
& \multicolumn{6}{c}{Example 2 (AR(1) 0.8 + AR(1) 0.05)}\\
\midrule
$(50, 50)$ & $0.067$ & $0.067$ & $\bf{0.005}$ & $0.072$ & $0.073$ & $\bf{0.005}$  \\ 
$(50, 100)$ & $0.070$ & $0.070$ & $\bf{0.004}$ & $0.075$ & $0.076$ & $\bf{0.003}$ \\ 
$(100, 100)$ & $0.069$ & $0.069$ & $\bf{0.003}$ & $0.075$ & $0.076$ & $\bf{0.003}$ \\ 
$(100, 200)$ & $0.071$ & $0.071$ & $\bf{0.003}$ & $0.076$ & $0.078$ & $\bf{0.002}$ \\ 
$(200, 200)$ & $0.072$ & $0.072$ & $\bf{0.002}$ & $0.078$ & $0.079$ & $\bf{0.002}$ \\ 
\midrule
& \multicolumn{6}{c}{Example 3 (AR(1) 0.8 + AR(1) 0.2))}\\
\midrule
$(50, 50)$ & $0.069$ & $0.068$ & $\bf{0.005}$ & $0.074$ & $0.075$ & $\bf{0.005}$  \\ 
$(50, 100)$ & $0.071$ & $0.071$ & $\bf{0.004}$ & $0.076$ & $0.077$ & $\bf{0.003}$ \\ 
$(100, 100)$ & $0.072$ & $0.072$ & $\bf{0.003}$ & $0.077$ & $0.078$ & $\bf{0.003}$ \\ 
$(100, 200)$ & $0.074$ & $0.074$ & $\bf{0.003}$ & $0.079$ & $0.080$ & $\bf{0.002}$ \\ 
$(200, 200)$ & $0.074$ & $0.074$ & $\bf{0.002}$ & $0.080$ & $0.081$ & $\bf{0.002}$ \\ 
\bottomrule
\end{tabularx}
\end{table}

Finally in Table \ref{dpca:ts.forecast}, we show the the $1$ step and $5$ steps ahead root mean square errors for the three examples. Overall, our method (FHFM) has the smallest FRMSE for all the examples while the CPCA method performs the worst on the forecasting. The phenomenon tells us that the features extracted via the auto-covariance matrix are better than the ones from the covariance matrix, in the view of the forecasting accuracy. Moreover, the FHFM has smaller forecasting error than the DPCA, which indicates that, extracting the different types of features sequentiality can benefit the forecasting more than mixing them together.
\begin{table}[!htbp] \centering 
  \footnotesize
  \caption{1 Step and 5 Steps Ahead Forecasting RMSE} 
  \label{dpca:ts.forecast} 
\begin{tabularx}{\textwidth}{c *{6}{Y}}
\toprule
 & \multicolumn{3}{c}{1 Step Ahead} 
 & \multicolumn{3}{c}{5 Steps Ahead}\\
\cmidrule(lr){2-4} \cmidrule(l){5-7}
 $(P, T)$ & FHFM & CPCA & DPCA & FHFM & CPCA & DPCA \\
\midrule
& \multicolumn{6}{c}{Example 1 (AR(1) 0.8 + N(0,1))}\\
\midrule
$(50, 50)$ & $\bf{0.808}$ & $0.829$ & $0.822$ & $\bf{1.046}$ & $1.053$ & $1.051$ \\  
$(50, 100)$ & $\bf{0.774}$ & $0.793$ & $0.787$ & $\bf{1.000}$ & $1.004$ & $1.004$ \\   
$(100, 100)$ & $\bf{0.789}$ & $0.812$ & $0.804$ & $\bf{1.046}$ & $1.054$ & $1.053$ \\ 
$(100, 200)$ & $\bf{0.790}$ & $0.814$ & $0.807$ & $\bf{1.029}$ & $1.035$ & $1.034$ \\ 
$(200, 200)$ & $\bf{0.800}$ & $0.819$ & $0.812$ & $\bf{0.986}$ & $0.996$ & $0.993$ \\ 
\midrule
& \multicolumn{6}{c}{Example 2 (AR(1) 0.8 + AR(1) 0.05)}\\
\midrule
$(50, 50)$ & $\bf{0.827}$ & $0.850$ & $0.844$ & $\bf{1.039}$ & $1.049$ & $1.047$ \\ 
$(50, 100)$ & $\bf{0.802}$ & $0.818$ & $0.813$ & $\bf{1.041}$ & $1.049$ & $1.046$ \\   
$(100, 100)$ & $\bf{0.804}$ & $0.826$ & $0.820$ & $\bf{1.025}$ & $1.028$ & $1.028$ \\  
$(100, 200)$ & $\bf{0.790}$ & $0.810$ & $0.802$ & $\bf{0.993}$ & $0.998$ & $0.995$ \\ 
$(200, 200)$ & $\bf{0.787}$ & $0.807$ & $0.800$ & $\bf{0.986}$ & $0.993$ & $0.991$ \\ 
\midrule
& \multicolumn{6}{c}{Example 3 (AR(1) 0.8 + AR(1) 0.2))}\\
\midrule
$(50, 50)$ & $\bf{0.791}$ & $0.809$ & $0.805$ & $\bf{1.039}$ & $1.045$ & $1.043$  \\ 
$(50, 100)$ & $\bf{0.799}$ & $0.812$ & $0.808$ & $\bf{1.034}$ & $1.037$ & $1.035$ \\  
$(100, 100)$ & $\bf{0.756}$ & $0.771$ & $0.766$ & $\bf{1.035}$ & $1.039$ & $1.040$ \\ 
$(100, 200)$ & $\bf{0.813}$ & $0.825$ & $0.822$ & $\bf{1.011}$ & $1.018$ & $1.015$ \\ 
$(200, 200)$ & $\bf{0.787}$ & $0.803$ & $0.799$ & $\bf{1.008}$ & $1.017$ & $1.015$ \\ 
\bottomrule
\end{tabularx}
\end{table}

In summary, our method (FHFM) provides more accurate forecasts for the high dimensional time series data simulated in this section. In the Appendix \ref{append:cha:stepPCA.1}, more special simulated cases are presented.


\section{Analysis of the US Mortality Data}
\label{dpca:empirical}
In this section, we apply the FHFM on age-specific mortality data of the US. The data is the mortality data of the US population in the Human Mortality Database (HMD) (\citenum{HMD}) obtained in December 2020. HMD contains original calculations of death rates and life tables for the populations in 40 countries and areas, as well as the input data used in constructing those tables. The data we originally obtained from HMD includes the annual age-sex-specific information of the number of exposures to risk, the number of deaths, and the central death rate, for ages from 0 to 110+ (age 100 and above) during the period from $1933$ to $2018$. We focus our analysis on the age-specific central death rates of the total sex population. As the mortality data for advanced ages are measured sparsely which is mentioned in \citet{lee1992modeling}, death rates for the older age groups (from age 91 to 110+) are summarized and incorporated into a modified death rate denoted as age $90+$. In view of this, the annual age-specific death rates under study consist of a matrix data with $86$ yearly observations (1933-2018) for $91$ ages (0-90$+$).
Following \citet{lee1992modeling}, we consider the log transformed central death rates $[\ln(m_{p,t})]_{P\times T}$, where $P = 91, T = 86$, for modeling purposes. By doing so, we can guarantee that the estimated and predicted central death rates are non-negative.
We show the better model fitting and forecasting performance of the FHFM compared to several other models in this section. Moreover, we explain by examining the factor loadings of the features that the FHFM and its two-steps dimension reduction estimation are necessary on the mortality data. At last, we illustrate with two applications that improving the accuracy of the predicted death rates is crucial.

\subsection{Stationarity}
\label{dpca:subsec:station}
As the log central death rates are not stationary time series, we modify the first estimation step in our method (FHFM) to deal with the non-stationary issue, which is summarized in Algorithm \ref{sta}.
\begin{algorithm}
\SetAlgoLined
\KwIn{Data $\bbY = [\bby_1, \dots, \bby_T] \in \mathbb{R}^{P \times T}$; $\bby_t = \left(\ln(m_{1,t}), \ln(m_{2,t}), \dots, \ln(m_{P,t})\right)^{\top}$.}
\KwOut{$\widehat{k}_{it}^{(1)}$, which is described in the first step of our method (see \ref{dpca:eq:ex2}).}
\caption{{\bf Modified Dimension Reduction Step 1 of FHFM} \label{sta}}
\textbf{Dimension Reduction Step 1}:\\
\nl Compute $\bbd_t = \bby_t - \bby_{t-1}$ and $\overline{\bbd} = T^{-1}\sum_{t=2}^{T-1} \bbd_t$, $t = 2, 3, \dots, T$\;
\nl Compute \begin{align*}
	&\widehat{\Sig}_{\bbd}(1) = \frac{1}{T-1} \sum_{t=2}^{T-1} (\bbd_{t+1} - \overline{\bbd})(\bbd_{t} - \overline{\bbd})^{\top}
	\end{align*} and $\widehat{\Sig}_{\bbd}(1) \widehat{\Sig}_{\bbd}(1)^{\top}$\;
\nl Compute $\widehat{\bbb}_i: P \times 1$ by the eigenvector corresponding to the $i^{th}$ largest eigenvalue of $\widehat{\Sig}_{\bbd}(1) \widehat{\Sig}_{\bbd}(1)^{\top}$,where $i =1,2, \dots, \widehat{r}_1$. $\widehat{r}_1$ is chosen by the method introduced in Section \ref{dpca:model}\;
\nl Compute $\overline{\bby} = T^{-1}\stT \bby_t$ and $\widehat{k}_{it}^{(1)} = \widehat{\bbb}_i^{\top}(\bby_t- \overline{\bby})$.
\end{algorithm}

The difference of Algorithm \ref{sta} and the first step in Algorithm \ref{Algorithm} is, instead of obtaining $\widehat{\bbb}_i$ by the eigenvectors of $\widehat{\Sig}_{\bby}(1) \widehat{\Sig}_{\bby}(1)^{\top}$, we get it from $\widehat{\Sig}_{\bbd}(1) \widehat{\Sig}_{\bbd}(1)^{\top}$. 
Because $\bby_t$ is not stationary, 
$\widehat{\Sig}_{\bby}(1) \widehat{\Sig}_{\bby}(1)^{\top}$ is not a good estimator for the population lag $1$ auto-covariance of $\bby_t$.

We now explain the reason for using $\widehat{\Sig}_{\bbd}(1) \widehat{\Sig}_{\bbd}(1)^{\top}$. From $\bbd_t = \bby_t - \bby_{t-1}$, we have $\bby_t = \bby_{t-1} + \bbd_t = \bby_{t-2} + \bbd_{t-1} + \bbd_{t} = \dots =\sum_{i = -\infty}^{t} \bbd_i$. With the coefficients $\bbb_i$, $\bbd_i$ can be expressed as $\bbd_i = \sum_{j = 1}^P \bbb_j \varphi_{ij}$ where $\varphi_{ij} = \bbb_j^{\top}\bbd_i$, thus 
\begin{align*}
\bby_t = \sum_{i = -\infty}^{t} \bbd_i = \sum_{i = -\infty}^{t}\left(\sum_{j = 1}^P \bbb_j \varphi_{ij}\right) = \sum_{j = 1}^P \left(\bbb_j \sum_{i = -\infty}^{t}\varphi_{ij}\right) = \sum_{j = 1}^P \bbb_j \psi_{tj},
\end{align*}
where $\psi_{tj} = \sum_{i = -\infty}^{t}\varphi_{ij}$. Thus when performing the dimension reduction, the coefficients to form a low-dimensional representation of $\bbu_t$ should be the same as those of $\bby_t$. If $\bbd_t$ is stationary, $\widehat{\Sig}_{\bbd}(1) \widehat{\Sig}_{\bbd}(1)^{\top}$ is a good estimator for $\Sig_{\bbd}(1) \Sig_{\bbd}(1)^{\top}$. Then eigenvectors of $\widehat{\Sig}_{\bbd}(1) \widehat{\Sig}_{\bbd}(1)^{\top}$ are better estimators of the factor loadings than those of $\widehat{\Sig}_{\bby}(1) \widehat{\Sig}_{\bby}(1)^{\top}$. We did stationary tests on the lag $1$ differenced series of each age separately, and more than $72\%$ of the ages have stationary results under significant level $0.1$. This might not be enough to say $\bbu_t$ is stationary, but for this dataset, it is better than the original log central death rates. 

Due to the same reason, we make the same modification for the static PCA and the dynamic PCA methods when comparing on the US mortality data. Thus from now on, ``CPCA'' and ``DPCA'' refers to the static PCA and dynamic PCA which deal with the non-stationary issue, respectively. In addition, for comparison purpose, we also apply the static PCA method without considering the non-stationary issue, which is exactly the same as the method in \citet{lee1992modeling} and we call it ``Lee-Carter'' in the following sections. 

\subsection{Revisit the Structure of the US Mortality Data}
Let us have a further discussion about the suitability of the FHFM for the US mortality data. 
We examine the variance and time serial dependence of the central death rates of the US. Because we modified the first step of the FHFM according to Section \ref{dpca:subsec:station}, instead of examining the original data, we check the first difference of log central death rates for each age. That is, for each age $p$, we compute the variance and lag $1$ autocorrelation (representing the time serial dependence) of $d_{p,2}, d_{p,3}, \dots, d_{p,T}$, where $d_{p,t} = \log{(m_{p,t})} - \log{(m_{p, t-1})}$, $p = 0,1, \dots, 90+$, and $T = 86$. The results are shown in Figure \ref{dpca:fg:f13}.

In Figure \ref{dpca:fg:f13}, the top red plot shows the variances of $d_{p,\cdot}$ for age $p = 0, 1, \dots, 90+$, and the bottom blue line shows the lag $1$ autocorrelation of $d_{p,\cdot}$. From the plot, we see that the variances of ages from $5$ to $13$ are larger than those of ages from $25$ to $40$, while the lag $1$ autocorrelations of ages from $5$ to $13$ are smaller than those of ages from $25$ to $40$. This is the same structure with the Example 5 in the simulation (in Appendix \ref{append:cha:stepPCA.1}). 
In addition, we have seen previously that the death rates of all ages have similar patterns, which indicates that information from part of the ages can be borrowed to help with the forecasting of other ages. 
Thus the first step of the proposed method would like to use information from the ages with powerful forecasting ability, such as ages $25$ to $40$, to help with the forecasting of other ages with weak correlations, ages $5$ to $13$ for instance. 
On the other hand, the parts with powerful forecasting ability do not contain sufficient variations. For example, most of the variation is contained in younger ages while they do not all have large correlations. Therefore, the second step of our method utilizes static PCA to help retain sufficient variation of the original data, which is necessary for the final recovery for forecasting. 
As a result, FHFM is particularly suitable for the US mortality data. In the next section, we examine the model fitting performance of our method and illustrate the necessity of both steps using the estimated factor loadings in the following section.

\begin{figure} 
\centering
\includegraphics[width=0.7\linewidth]{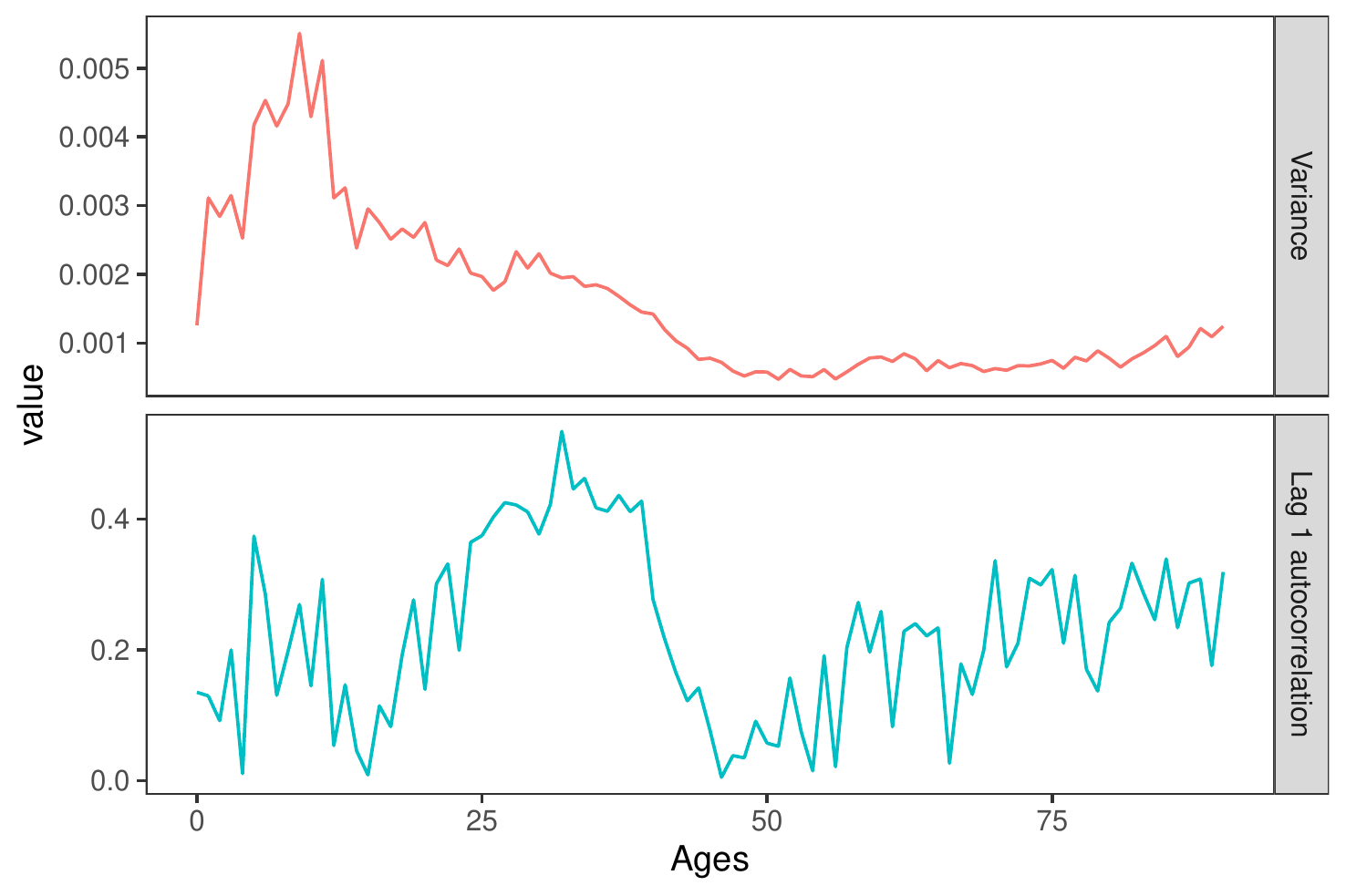}
\caption{Variance and Time serial dependence of ages}
\label{dpca:fg:f13}
\end{figure}

\subsection{Model Fitting Performance Comparison}

In this section, we check the performance of the FHFM on fitting the original data.
We apply FHFM, Lee-Carter, CPCA, and DPCA on the logarithm of central rates of death and compare the fitting performances using the root mean square error (RMSE). Based on the criteria in Section \ref{dpca:model}, all the methods choose only one factor. For FHFM, we have $\widehat{r}_1 = 1$ and $\widehat{r}_2 = 1$. Table \ref{dpca:tb:t1.1} and Table \ref{dpca:tb:t1.2} show the RMSEs of the four methods for selected ages ($5$, $25$, $50$, $65$ and $85$) and years ($1933$, $1953$, $1993$, and $2018$) respectively, along with the overall RMSE of the whole data. From the tables, the RMSEs of the FHFM are the lowest among the other three methods, which shows that FHFM fits the data the best. As we described before, the two steps of the FHFM guarantee that it captures sufficient variations of the original data and results in a good model fitting. In addition, we see that the RMSEs of Lee-Carter are smaller than those of CPCA. It implies that the fitting performance on the log central death rates is worse for static PCA if we revise the method to deal with non-stationarity. This phenomenon may be caused by special characteristics of the mortality data, which is interesting to explore further.

\begin{table}[!htbp] \centering 
  \footnotesize
  \caption{RMSE, for some specific ages} 
  \label{dpca:tb:t1.1} 
\begin{tabularx}{\textwidth}{c *{4}{Y}}
\toprule
 Age & FHFM & Lee-Carter & CPCA & DPCA \\
\midrule
5 & $\bf{0.049}$ & $0.062$ & $0.304$ & $0.274$ \\ 
25 & $\bf{0.061}$ & $0.126$ & $0.191$ & $0.189$ \\ 
50 & $\bf{0.051}$ & $0.063$ & $0.108$ & $0.119$ \\ 
65 & $\bf{0.038}$ & $0.086$ & $0.126$ & $0.152$ \\ 
85 & $\bf{0.046}$ & $0.067$ & $0.078$ & $0.119$ \\ 
\midrule 
RMSE & $\bf{0.055}$ & $0.083$ & $0.151$ & $0.155$ \\   
\bottomrule
\end{tabularx}
\end{table}

\begin{table}[!htbp] \centering 
  \footnotesize
  \caption{RMSE, for some specific years} 
  \label{dpca:tb:t1.2} 
\begin{tabularx}{\textwidth}{c *{4}{Y}}
\toprule
 Year & FHFM  & Lee-Carter & CPCA & DPCA \\
\midrule
1933 & $\bf{0.076}$ & $0.153$ & $0.186$ & $0.171$ \\ 
1953 & $\bf{0.047}$ & $0.092$ & $0.160$ & $0.165$ \\ 
1993 & $\bf{0.063}$ & $0.080$ & $0.142$ & $0.145$ \\ 
2018 & $\bf{0.083}$ & $0.140$ & $0.275$ & $0.289$ \\ 
\midrule
RMSE & $\bf{0.055}$ & $0.083$ & $0.151$ & $0.155$ \\ 
\bottomrule
\end{tabularx}
\end{table}

We can also visualize the fitting performances of the four methods via plots. Figure \ref{dpca:fig:side:a} shows the actual and fitted log central rates of death for selected ages ($5$, $25$, $50$, $65$ and $85$) over all historical years from $1933$ to $2018$, while Figure \ref{dpca:fig:side:b} shows the actual and fitted log central rates of death for selected years ($1933$, $1953$, $1993$ and $2018$) over all ages from $0$ to $90+$. The black lines represent the actual log central rates of death; the red, light blue, green and blue dashed lines show the fitted log central rates of death using FHFM, Lee-Carter, CPCA, and DPCA respectively. From Figure \ref{dpca:fig:side:a}, we see that the FHFM captures the time-serial patterns well for all selected ages even when there is curvature, such as the paths of ages $25$, $65$ and $85$. However, the Lee-Carter, CPCA and DPCA failed to recover the time-serial dependence appropriately and hence provide worse fitting results than the FHFM. From Figure \ref{dpca:fig:side:b}, we see that the four methods provide similar fittings for ages $0$ to $20$. For ages $20$ to $40$, the mortality patterns changed and FHFM shows a better fitting performance than the other three methods. The four models perform similarly again for ages $40$ and above with the FHFM's fitting performance slightly better than the other three, especially for the year $2018$. Hence, it shows that the FHFM captures both time-serial (time dimension) dependence and cross-sectional (age dimension) variation well and exhibits advantages over the other three methods especially when the mortality trends change.

\begin{figure}
\centering
\begin{minipage}[t]{0.48\textwidth}
\includegraphics[width=\linewidth]{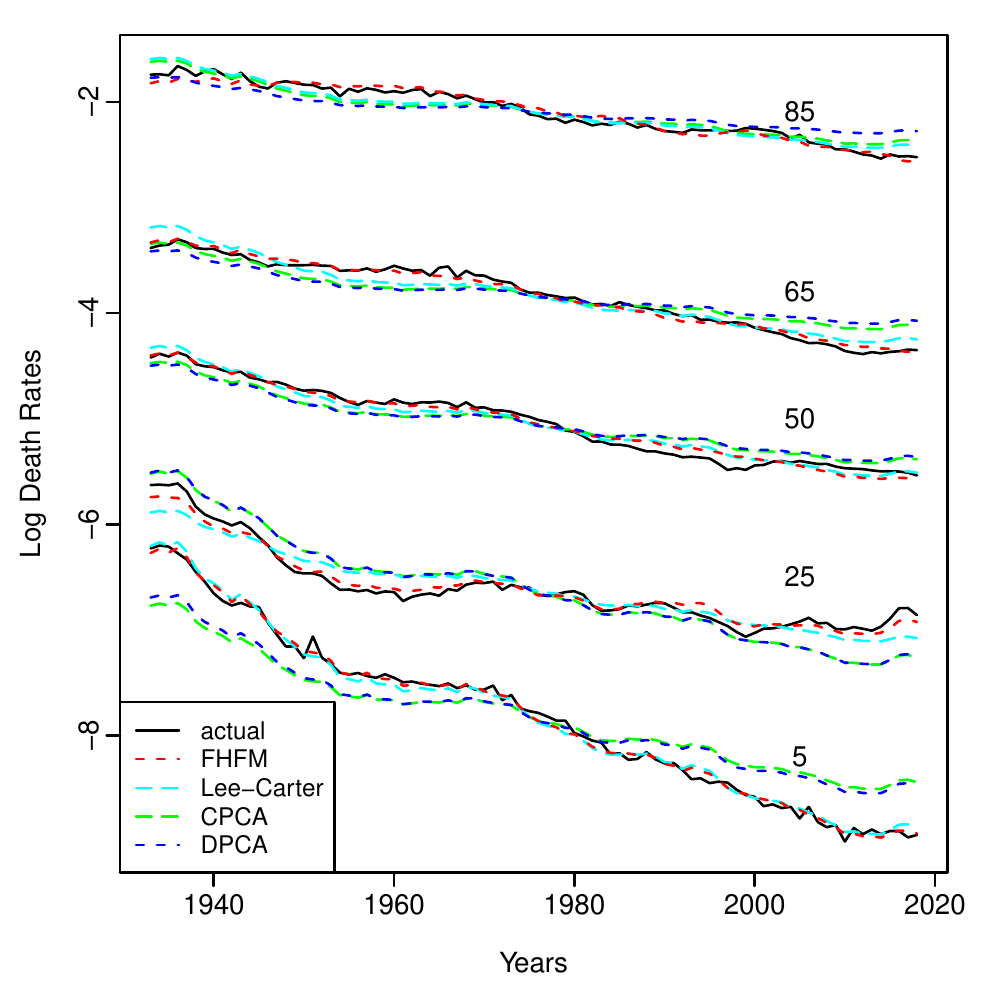}
\caption{Log Death Rates, $1933-2018$ for ages $5$, $25$, $50$, $65$, $85$; Actual and Fitted. }
\label{dpca:fig:side:a}
\end{minipage}
\hspace*{0.01cm}
\begin{minipage}[t]{0.48\textwidth}
\includegraphics[width=\linewidth]{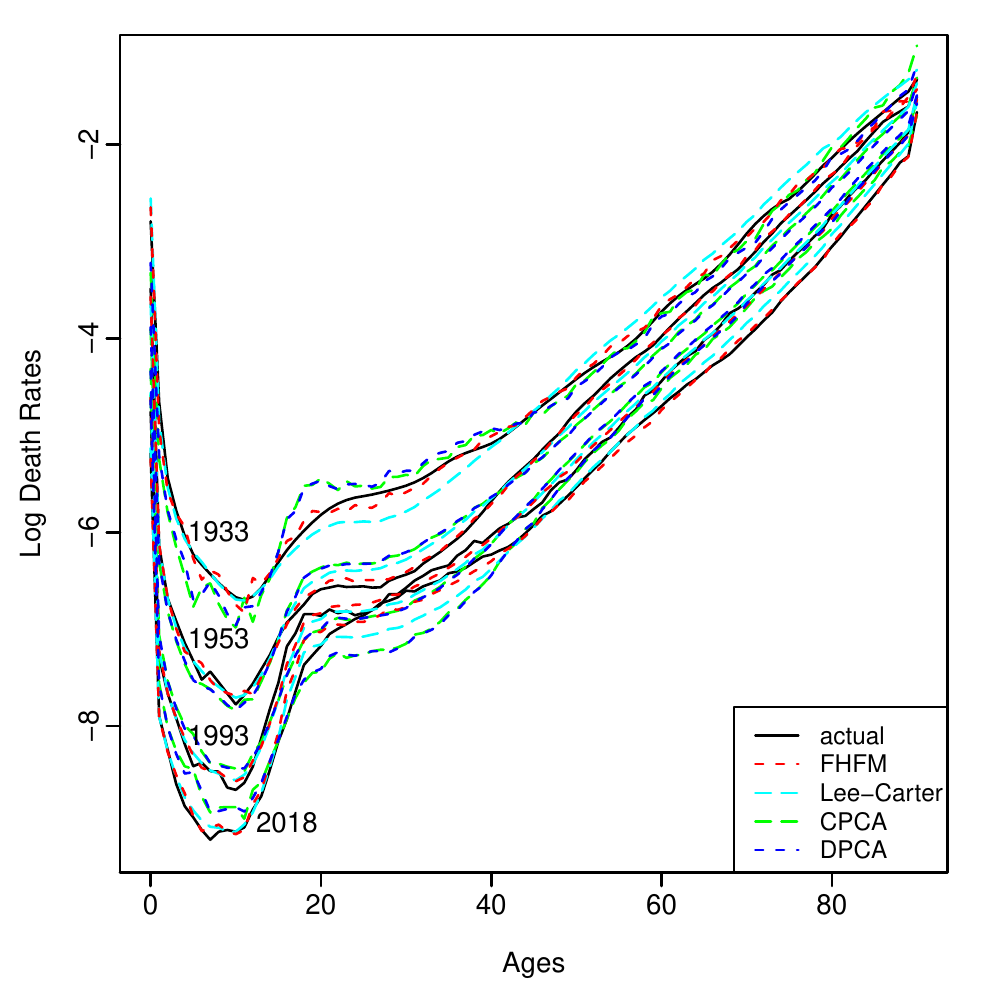}
\caption{Log Death Rates, ages $0$ to $90+$ for years $1933$, $1953$, $1993$, $2018$; Actual and Fitted.}
\label{dpca:fig:side:b}
\end{minipage}
\end{figure}

\subsection{Forecasting Performance Comparison}
In this section, we show that FHFM can forecast the central death rates of the US more accurately than the other methods. We compare the forecasting performance of FHFM with the Lee-Carter model in \citet{lee1992modeling} (`Lee-Carter'), the static PCA described in Section 3 (`CPCA'), the dynamic PCA described in Section 3 (`DPCA'), the Age-Period-Cohort model consided in \citet{currie2004smoothing} (`APC'), the two factor Cairns-Blake-Dowd model introduced by \citet{cairns2006two} (`CBD'), the quadratic CBD models with cohort effects considered in \citet{Cairns2009} (`M6', `M7', and `M8'), the functional mortality model of \citet{hyndman2007robust} (`FDM'), and the univariate ARIMA model (`Individual'). We use rolling window out-of-sample forecasting to evaluate the performance. The univariate ARIMA model is fitted with time series of each age independently and the model structure is selected based on Bayesian information criterion (BIC); we refer to this as ``individual'' model. The individual model is included for comparison, as we would like to show that conducting dimension reduction before the forecasting is necessary, especially in the long term. 

We use the data covering years from $2009$ to $2018$ as the test set and the historical data of previous years as the training set for modeling and testing purposes. Table \ref{dpca:tb:t2} shows the forecasting root mean square errors (FRMSEs) of $1$ to $25$ steps ahead forecasts using the five methods. For each forecast, we have $10$ rolling window sub-training sets for the $10$ test years and the values presented in the table are the averages of the $10$ rolling window FRMSEs. Figure \ref{dpca:fg:f9} plots the results shown in Table \ref{dpca:tb:t2} except for those of M6, M7, and M8, as the performance of M6, M7, and M8 are much worse than the others. We see that as the length of prediction steps increases, the performance of all methods gets worse. This is because the longer-term forecasting is always harder and contains more uncertainty. 

The individual model and FDM have better forecasting accuracy than the other models when $h \le 13$, while  perform worse in the long-term compared with the FHFM and the Lee-Carter. The individual model focuses on capturing the mortality pattern of each age vector, which ignores the dependence among different ages, overlooks the cross-sectional common information, and contains many more parameters than the FHFM. The FDM performs similarly to the individual model, which indicates its forecasts might put more weights on the individual patterns than the common information.  In the short term, individual factors dominate the forecasting performance, so the individual model and FDM perform better than others.  However, the long-term mortality forecasting is usually more important in practice as it provides important assumptions for various actuarial practices and government policymaking, such as life insurance and annuities pricing and reserving, asset liability management of pension funds, and the solvency analysis of social securities. In the long term, different ages share similar drivers of the mortality variation, such as technology innovation, health improvement, wars, and epidemics. 
So common factors dominate the forecasting performance in the long term, and dimension reduction plays a crucial role in recovering the common information from the high dimensional mortality data. 

FHFM performs the best overall (with the smallest average FRMSE, 0.181, over all forecasting horizons) and especially for long-term forecasting. Comparing to other methods, FHFM has the smallest FRMSEs for all $h>13$, and the smallest FRMSEs for all $h$ when excluding the individual model and FDM. The empirical analysis shows that FHFM successfully extracts features with powerful forecasting ability and provides a good low-rank representation.

\begin{table}[!htbp] \centering 
  \footnotesize
  \caption{Comparison of the forecasting performance on the US data: Rolling Window FRMSE. `FHFM' for our proposed model; `Lee-Carter' for the Lee-Carter model in \citet{lee1992modeling}; `CPCA' for the static PCA described in Section 3; `DPCA' for the dynamic PCA described in Section 3; `APC' for the Age-Period-Cohort model consided in \citet{currie2004smoothing}; `CBD' for the two factor Cairns-Blake-Dowd model introduced by \citet{cairns2006two}; `M6', `M7', and `M8' for the quadratic CBD models with cohort effects considered in \citet{Cairns2009}; `FDM' for the functional mortality model of \citet{hyndman2007robust}; `Individual' for fitting univariate ARIMA model with time series of each age independently.}
  \label{dpca:tb:t2} 
\begin{tabularx}{\textwidth}{c *{11}{Y}}
\toprule
 $h$ & FHFM & Lee-Carter & CPCA & DPCA & APC & CBD&  M6 & M7 & M8 & FDM & Indiv-idual\\
\midrule
1 & $0.092$ & $0.130$ & $0.245$ & $0.261$ & $0.282$ & $0.614$ & $0.848$ & $0.796$ & $0.840$ & $0.040$ & $\bf{0.039}$ \\ 
2 & $0.104$ & $0.139$ & $0.249$ & $0.265$ & $0.301$ & $0.632$ & $0.923$ & $0.929$ & $0.866$ & $\bf{0.059}$ & $0.060$ \\ 
3 & $0.118$ & $0.149$ & $0.255$ & $0.271$ & $0.321$ & $0.650$ & $1.005$ & $1.078$ & $0.890$ & $\bf{0.080}$ & $0.081$ \\ 
4 & $0.132$ & $0.161$ & $0.262$ & $0.277$ & $0.338$ & $0.669$ & $1.095$ & $1.242$ & $0.912$ & $0.099$ & $\bf{0.097}$ \\ 
5 & $0.146$ & $0.172$ & $0.270$ & $0.286$ & $0.354$ & $0.685$ & $1.193$ & $1.418$ & $0.932$ & $0.118$ & $\bf{0.115}$ \\ 
6 & $0.156$ & $0.180$ & $0.276$ & $0.293$ & $0.367$ & $0.701$ & $1.293$ & $1.600$ & $0.950$ & $0.131$ & $\bf{0.128}$ \\ 
7 & $0.164$ & $0.187$ & $0.282$ & $0.299$ & $0.380$ & $0.714$ & $1.405$ & $1.796$ & $0.967$ & $0.142$ & $\bf{0.137}$ \\ 
8 & $0.169$ & $0.191$ & $0.286$ & $0.304$ & $0.390$ & $0.727$ & $1.519$ & $1.991$ & $0.980$ & $0.150$ & $\bf{0.142}$ \\ 
9 & $0.174$ & $0.195$ & $0.290$ & $0.308$ & $0.400$ & $0.741$ & $1.644$ & $2.197$ & $0.991$ & $0.155$ & $\bf{0.145}$ \\ 
10 & $0.173$ & $0.196$ & $0.291$ & $0.310$ & $0.409$ & $0.755$ & $1.757$ & $2.397$ & $0.997$ & $0.158$ & $\bf{0.145}$ \\ 
11 & $0.177$ & $0.199$ & $0.295$ & $0.315$ & $0.417$ & $0.773$ & $1.868$ & $2.604$ & $1.000$ & $0.164$ & $\bf{0.152}$ \\ 
12 & $0.182$ & $0.203$ & $0.299$ & $0.319$ & $0.428$ & $0.788$ & $1.961$ & $2.801$ & $0.999$ & $0.170$ & $\bf{0.164}$ \\ 
13 & $0.190$ & $0.209$ & $0.304$ & $0.325$ & $0.437$ & $0.804$ & $2.065$ & $3.008$ & $0.998$ & $0.204$ & $\bf{0.188}$ \\ 
14 & $\bf{0.197}$ & $0.215$ & $0.309$ & $0.330$ & $0.446$ & $0.816$ & $2.157$ & $3.220$ & $0.996$ & $0.215$ & $0.202$ \\ 
15 & $\bf{0.205}$ & $0.221$ & $0.314$ & $0.335$ & $0.456$ & $0.823$ & $2.232$ & $3.432$ & $0.995$ & $0.233$ & $0.223$ \\ 
16 & $\bf{0.210}$ & $0.226$ & $0.317$ & $0.338$ & $0.464$ & $0.828$ & $2.323$ & $3.674$ & $0.996$ & $0.251$ & $0.244$ \\ 
17 & $\bf{0.217}$ & $0.233$ & $0.321$ & $0.342$ & $0.475$ & $0.833$ & $2.389$ & $3.910$ & $0.996$ & $0.272$ & $0.268$ \\ 
18 & $\bf{0.222}$ & $0.238$ & $0.324$ & $0.346$ & $0.487$ & $0.837$ & $2.465$ & $4.175$ & $0.997$ & $0.284$ & $0.283$ \\ 
19 & $\bf{0.226}$ & $0.243$ & $0.326$ & $0.348$ & $0.500$ & $0.839$ & $2.499$ & $4.417$ & $0.996$ & $0.303$ & $0.302$ \\ 
20 & $\bf{0.223}$ & $0.245$ & $0.325$ & $0.346$ & $0.512$ & $0.837$ & $2.538$ & $4.661$ & $0.994$ & $0.306$ & $0.310$ \\ 
21 & $\bf{0.217}$ & $0.246$ & $0.322$ & $0.344$ & $0.524$ & $0.835$ & $2.553$ & $4.848$ & $0.992$ & $0.311$ & $0.313$ \\ 
22 & $\bf{0.208}$ & $0.245$ & $0.318$ & $0.339$ & $0.532$ & $0.835$ & $2.605$ & $5.084$ & $0.993$ & $0.315$ & $0.307$ \\ 
23 & $\bf{0.201}$ & $0.247$ & $0.315$ & $0.336$ & $0.543$ & $0.841$ & $2.636$ & $5.308$ & $0.992$ & $0.247$ & $0.277$ \\ 
24 & $\bf{0.202}$ & $0.254$ & $0.317$ & $0.339$ & $0.556$ & $0.854$ & $2.642$ & $5.498$ & $0.988$ & $0.242$ & $0.276$ \\ 
25 & $\bf{0.216}$ & $0.268$ & $0.327$ & $0.348$ & $0.570$ & $0.878$ & $2.665$ & $5.691$ & $0.983$ & $0.249$ & $0.275$ \\ 
\midrule
Mean & $\bf{0.181}$ & $0.208$ & $0.298$ & $0.317$ & $0.435$ & $0.772$ & $1.931$ & $3.111$ & $0.970$ & $0.196$ & $0.195$ \\ 
\bottomrule
\end{tabularx}
\end{table}

\begin{figure} 
\centering
\includegraphics[width=0.7\linewidth]{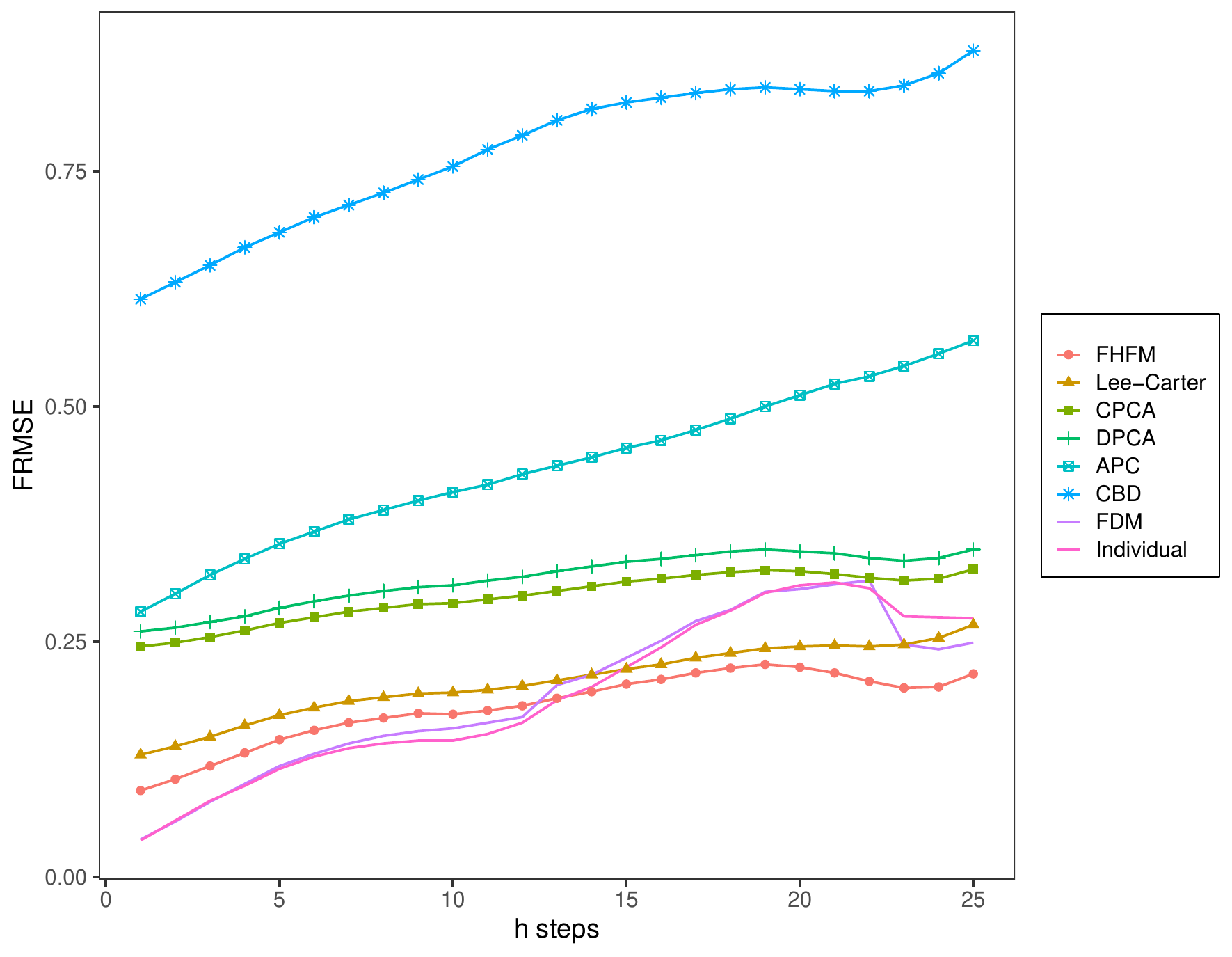}
\caption{Comparison of the forecasting performance on the US data: Rolling Window FRMSE}
\label{dpca:fg:f9}
\end{figure}

\subsection{Analysis of the FHFM on Mortality Data}
Recall that the FHFM intends to capture two kinds of common features for mortality data among all ages: common temporal trends and common variations. Now we would like to analysis the necessity and the behavior of the proposed model on the mortality data under study. 

Figure \ref{dpca:fg:new1} provides the estimation of $\{u_{p,t}: t\geq 1\}$ for all $p=1, \ldots, P$, which is the residual of the first step. After extracting the common temporal trends in the first step, it is expected that there is relatively weak common temporal trend existed in the residual $u_{p,t}$. 
Compared with Figure \ref{logdr:a} that illustrates the time-trends in original mortality data, Figure \ref{dpca:fg:new1} indeed demonstrates weak common time-trends for all ages, in view of different time-tendency for the young ages from that of the old ages. 
\begin{figure} 
\centering
\includegraphics[width=0.7\linewidth]{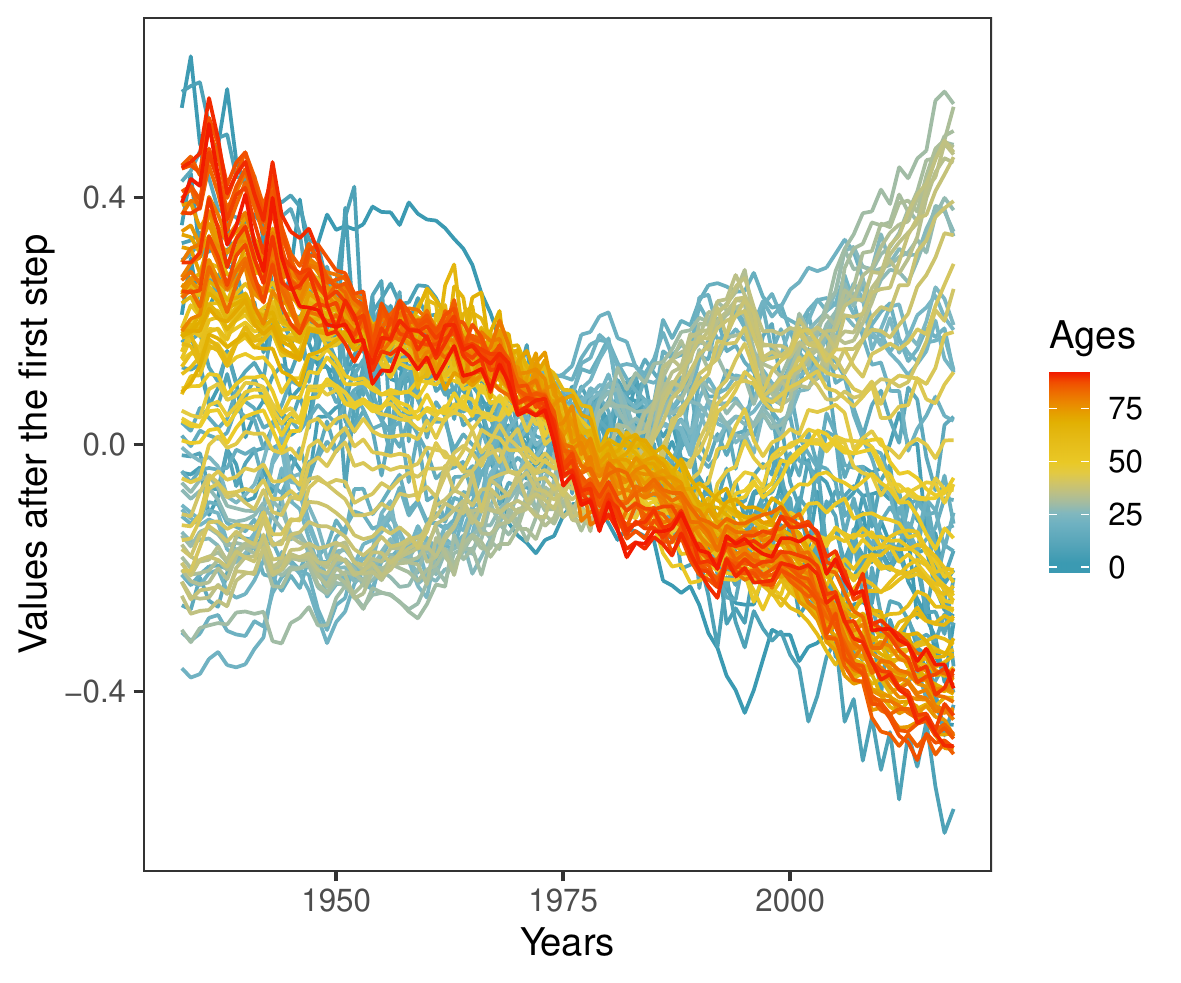}
\caption{Estimation of $u_t$ on the US mortality Data}
\label{dpca:fg:new1}
\end{figure}

Next, we investigate the extracted features from the two kinds of factor models, respectively. As analyzed earlier, the estimation for the two parts is based on the eigendecomposition of the two matrices $\widehat{\bbL}_1$ and $\widehat{\bbL}_2$, respectively. The first row of Figure \ref{dpca:fg:new3} shows all the eigenvalues of the two matrices. The spikeness is obvious and the ratio-based statistic will estimate $r_1$ and $r_2$ as $1$ intuitively. Then the bottom row of Figure \ref{dpca:fg:new3} provides the eigenvectors of $\widehat{\bbL}_1$ and $\widehat{\bbL}_2$ corresponding to their largest eigenvalues respectively. By comparing them, factor loadings from the two parts of factor models are quite different from each other. 
\begin{figure} 
\centering
\includegraphics[width=0.8\linewidth]{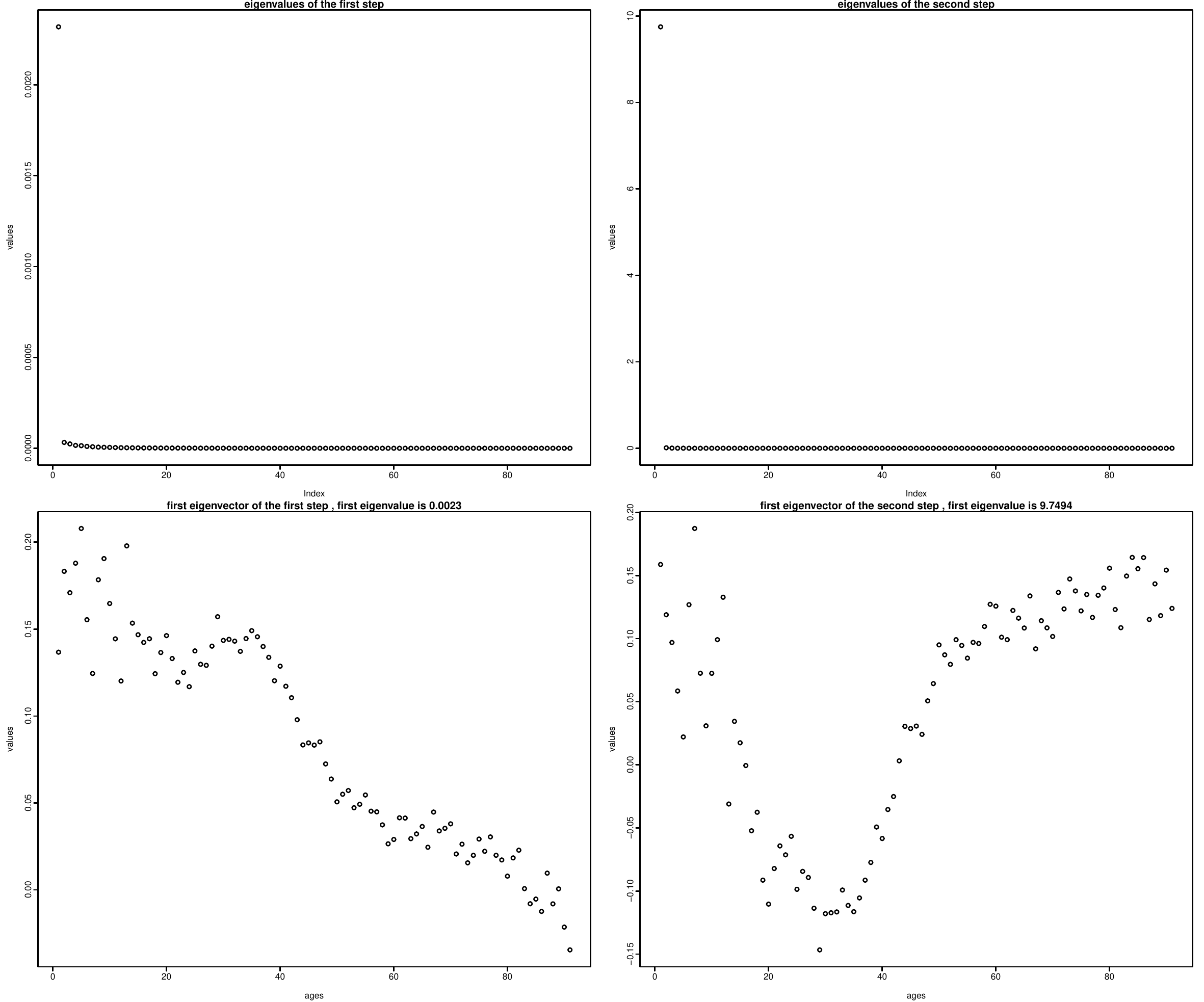}
\caption{$1st$ Principal Component in Two Steps}
\label{dpca:fg:new3}
\end{figure}

A natural question may arise: would it produce similar results if we included two factors in the first step instead of having the second-step factor modelling? Figure \ref{dpca:fg:new4} shows that the second eigenvector of $\widehat{\bbL}_1$ is different from the first eigenvector extracted by the second step, which ensures the necessity of the second-step factor modelling. Roughly speaking, the second principal component (PC) in the first step represents weaker common temporal trend than the first PC, but stronger than the remaining PCs. However, the aim of our second step is to pursue features possessing most common variations of the residuals out of the first step. Although the extracted factor in the second step also has weaker common temporal trend than that is extracted in the first step, it is different from the second PC of the first step in view of Figure \ref{dpca:fg:new4}.   
\begin{figure} 
\centering
\includegraphics[width=0.8\linewidth]{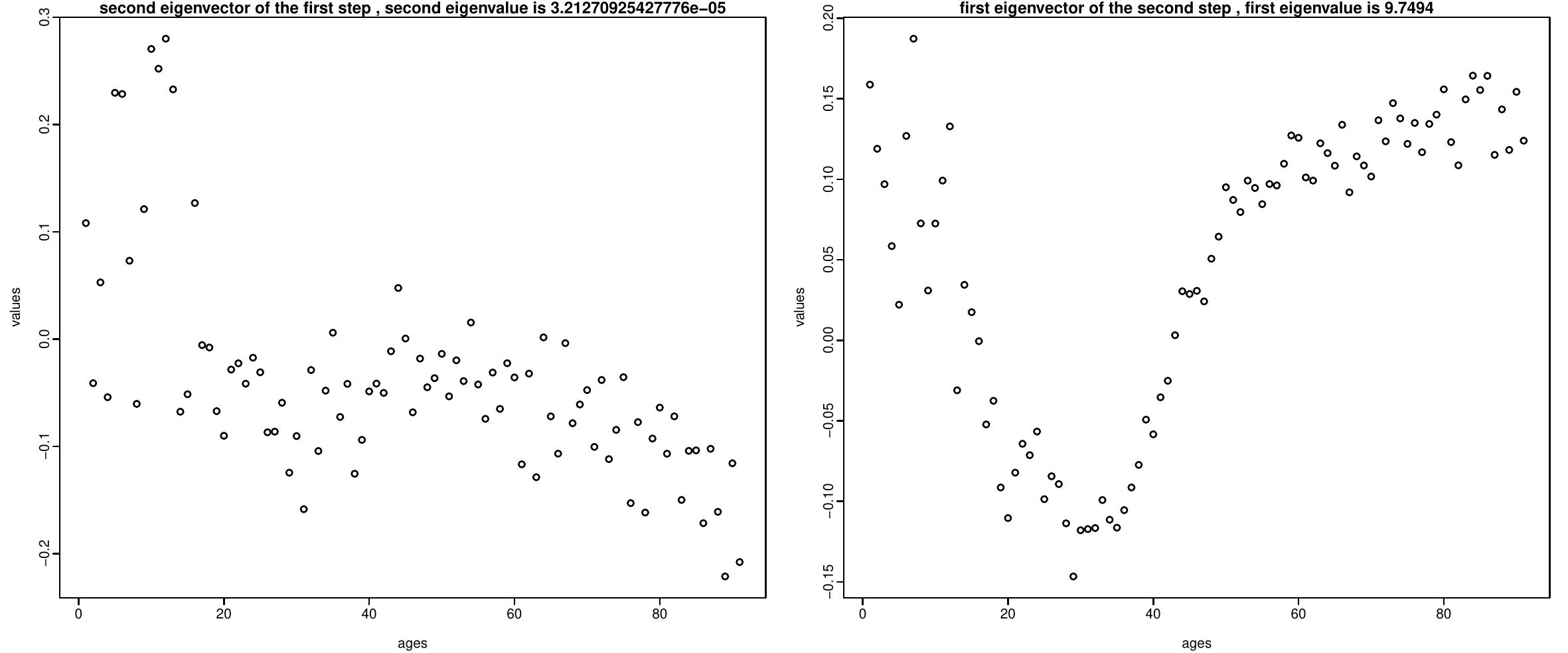}
\caption{$2nd$ PC in Step 1 and $1st$ PC in Step 2}
\label{dpca:fg:new4}
\end{figure}
Furthermore, the second eigenvalue of matrix $\widehat{\bbL}_1$, as shown in Figure \ref{dpca:fg:new5}, is not separable from the remaining ones (the third PC and above), and hence is not favored for forecasting purposes.

\begin{figure} 
\centering
\includegraphics[width=0.8\linewidth]{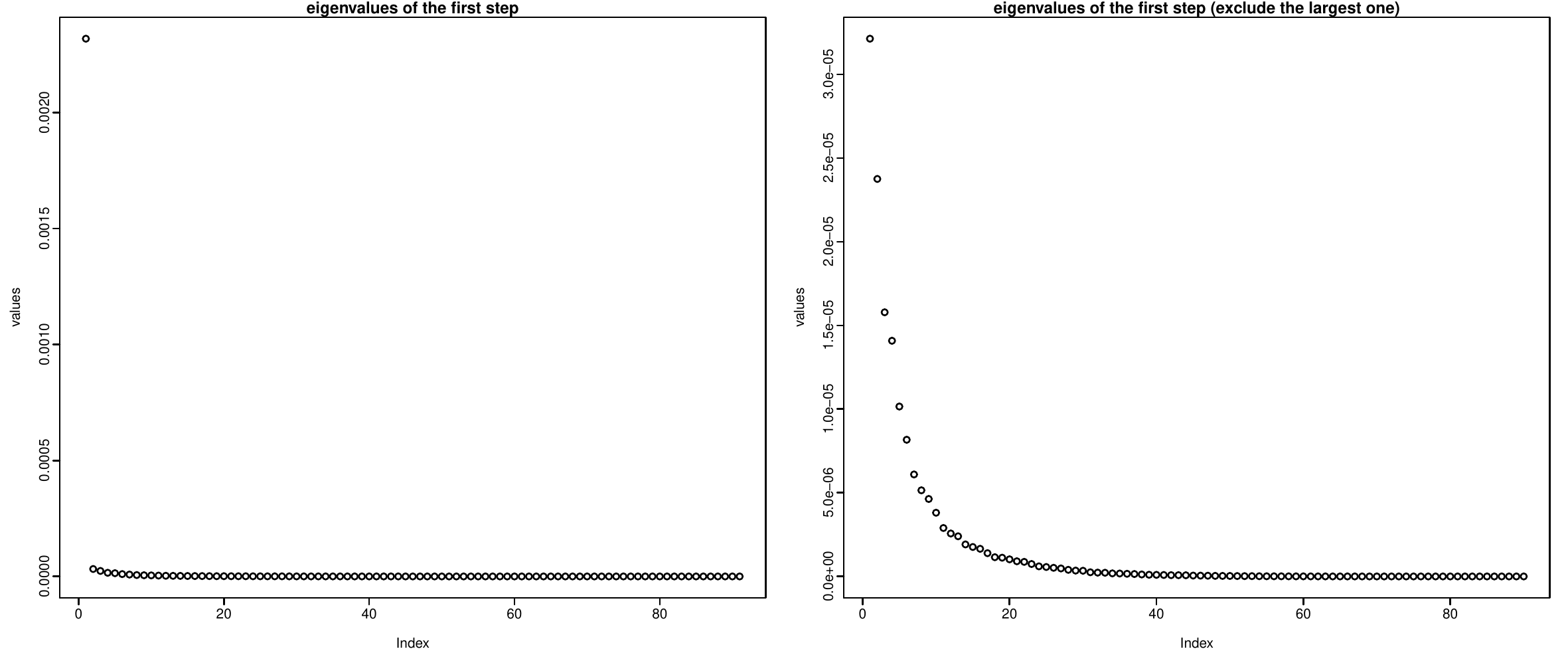}
\caption{Eigenvalues of Step 1}
\label{dpca:fg:new5}
\end{figure}



\subsection{Application of the Mortality Forecasting}
In this section, we use the forecasts of mortality to perform two applications: predicting the life expectancies and pricing the life annuities. The life expectancy describes the expected average remaining number of years prior to death for a person reached a specific age. Usually it can be reported in two different forms based on the mortality rates (period and cohort). The period life expectancy for a given year of each age is calculated based on the mortality rates for that single year, while the cohort life expectancy is estimated based on the mortality rates for the series of years in which the person will actually reach each succeeding age if the individual survives (\citet{usreport}).
For example, according to Table V.A4 and V.A5 in the 2019 report of \citet{usreport}, a male in the US aged $65$ in year $2018$ is expected to live another $18.1$ years before death on a period basis while $18.9$ years on a cohort basis. We will compare the estimated cohort and period life expectancy from our proposed method (FHFM) with those from the Lee-Carter model. In addition, related to the cohort life expectancy, another interesting and crucial problem is, how much would an individual pay for an insurance which provides annual payments after the retirement until the death? We will compare the present values (price, per $\$1$) of the life annuities based on the estimated cohort life expectancies from different methods.

In the following part, we compute the actuarial life expectancy for an individual aged $x$ at year $T$ ($e_{x, T}$) as follows,
\begin{align*}
e_{x, T} = 
\sum_{t = 1}^{w - x - 1} {}_t p_{x,T},
\end{align*}
where $w$ is the assumed maximum age, and ${}_t p_{x,T} = \prod_{j = 0}^{t-1} \left(1- {}_1 q_{x+j,T}\right)$ is the probability that a person aged $x$ at year $T$ will survive to age $x + t$. For the period life expectancy, ${}_1 q_{x+j,T} = m_{x+j, T}$, and for the cohort life expectancy, ${}_1q_{x+j,T} = m_{x+j, T+j}$, where $m_{x, t}$ is the death rate of a person aged $x$ at year $t$ from the mortality table. In addition, for simplicity, we assume $1 - m_{90+, T}$ represents the probability that a person age $90$ will survive to the maximum age $w$. Further, we calculate the present value of the life annuity ($PV_{x,T}$) for an individual purchased at age $x$ in year $T$ and beginning to make payments $\$ 1$ annually after age 66 until death or aged 90 (which one happens first) as below:
\begin{align*}
PV_{x, T} = \begin{cases}
    \sum_{t = 1}^{90 - x} {}_t p_{x,T}/(1+i)^t       & \quad \text{if } x \ge 66\\
    PV_{66, T + (66-x)}/(1+i)^{66-x}  & \quad \text{if } x < 66
  \end{cases},
\end{align*}
where $i = 2\%$ is the interest rate, ${}_t p_{x,T} = \prod_{j = 0}^{t-1} \left(1- {}_1 q_{x+j,T}\right)$ and ${}_1q_{x+j,T}$ $= m_{x+j, T+j}$, which is on a cohort basis and the same with the calculation for the cohort life expectancy. We let the life annuities end at age 90 for simplicity as the mortality rates for extreme older ages need more detailed analysis, which is beyond the scope of this paper. The age 66 is the retirement age for most individuals in the US. Hence, for an individual younger than 66, $PV_{x, T}$ is the price for a deferred life annuity. Similar calculation can be find in \citet{warshawsky1988private}, \cite{mccarthy2001assessing}, and \citet{cunningham2012models}.

In order to compare the out-of-sample performance of our method (FHFM) and the Lee-Carter model, we define the data for years 1933 to 1988 as the training set and the data for the last 30 years (1989 to 2018) as the test set. We first forecast the mortality rates of the test set with the training set using the FHFM and the Lee-Carter method, respectively. Then we calculate $e_{x, T}$ (cohort and period) and $PV_{x,T}$ using the actual mortality rates as well as the forecasting mortality rates from the two methods, respectively. 

With more accurate mortality forecasts, how much can the FHFM method improve the prediction of the life expectancies and the pricing of the life annuities? Table \ref{dpca:tb:ta1} shows the forecast mean squared error (FMSE) and the forecast mean absolute error (FMAE) for the FHFM method and the Lee-Carter model, which are computed as
\begin{align*}
\text{FMSE} &= \frac{1}{N}\sum_{x}\sum_{t} (\hat{y}_{x, t} - y_{x,t})^2,\\
\text{FMAE} &= \frac{1}{N}\sum_{x}\sum_{t} |\hat{y}_{x, t} - y_{x,t}|,
\end{align*}
where $\hat{y}_{x,t}$ is the estimated value (computed with forecast death rates from the FHFM or Lee-Carter), $y_{x,t}$ is the true value (computed with actual death rates), $N$ is the number of estimates (it is different for the period and cohort life expectancies). It can be seen from the table that for all the three applications, the estimations from the FHFM have smaller FMSEs and FMAEs comparing with those from Lee-Carter method. Particularly, from the FMAEs of the present values of life annuities, we can see that, on average, the pricing error is $\$0.154$ for Lee-Carter and only $\$0.041$ for FHFM with annual payment $\$1$. The better performance of FHFM is lead by the more accurately mortality forecasting.
\begin{table}[!htbp] \centering 
  \footnotesize
  \caption{FMSE and FMAE of life expectancies (cohort and period) and present values of annuities (annual payment \$1 and interest rate 2\%)} 
  \label{dpca:tb:ta1} 
\begin{tabularx}{\textwidth}{c *{6}{Y}}
\toprule
 & \multicolumn{3}{c}{FMSE} 
 & \multicolumn{3}{c}{FMAE}\\
\cmidrule(lr){2-4} \cmidrule(l){5-7}
  & period life expectancy & cohort life expectancy & pv of life annuity & period life expectancy & cohort life expectancy & pv of life annuity \\
\midrule
Lee-Carter & $0.768$ & $0.111$ & $0.040$ & $0.790$ & $0.251$ & $0.154$ \\ 
FHFM & $\bf{0.109}$ & $\bf{0.009}$ & $\bf{0.004}$ & $\bf{0.263}$ & $\bf{0.072}$ & $\bf{0.041}$ \\  
\bottomrule
\end{tabularx}
\end{table}

Figure \ref{dpca:fg:fa1} and Figure \ref{dpca:fg:fa2} show the cohort and period life expectancies for an individual aged $66$ at different years. The red line is the value computed from historical death rates, the green one is the value computed with the forecast from the FHFM and the blue line is that from the Lee-Carter method. From Figure \ref{dpca:fg:fa1}, we see that the three lines are close to each other before 1970, which is due to the less forecasts involved in the calculation for those years. After 1970, when involving more forecasts, the Lee-Carter method tends to estimate the life expectancies lower while the FHFM is close to the true value with slightly higher estimations for some years. From Figure \ref{dpca:fg:fa2}, we see that the output of FHFM is always more close to the the true values while both the FHFM and the Lee-Carter tend to underestimate the true values for the second half of the time horizon.
\begin{figure}
\begin{minipage}[t]{0.48\textwidth}
\includegraphics[width=\linewidth]{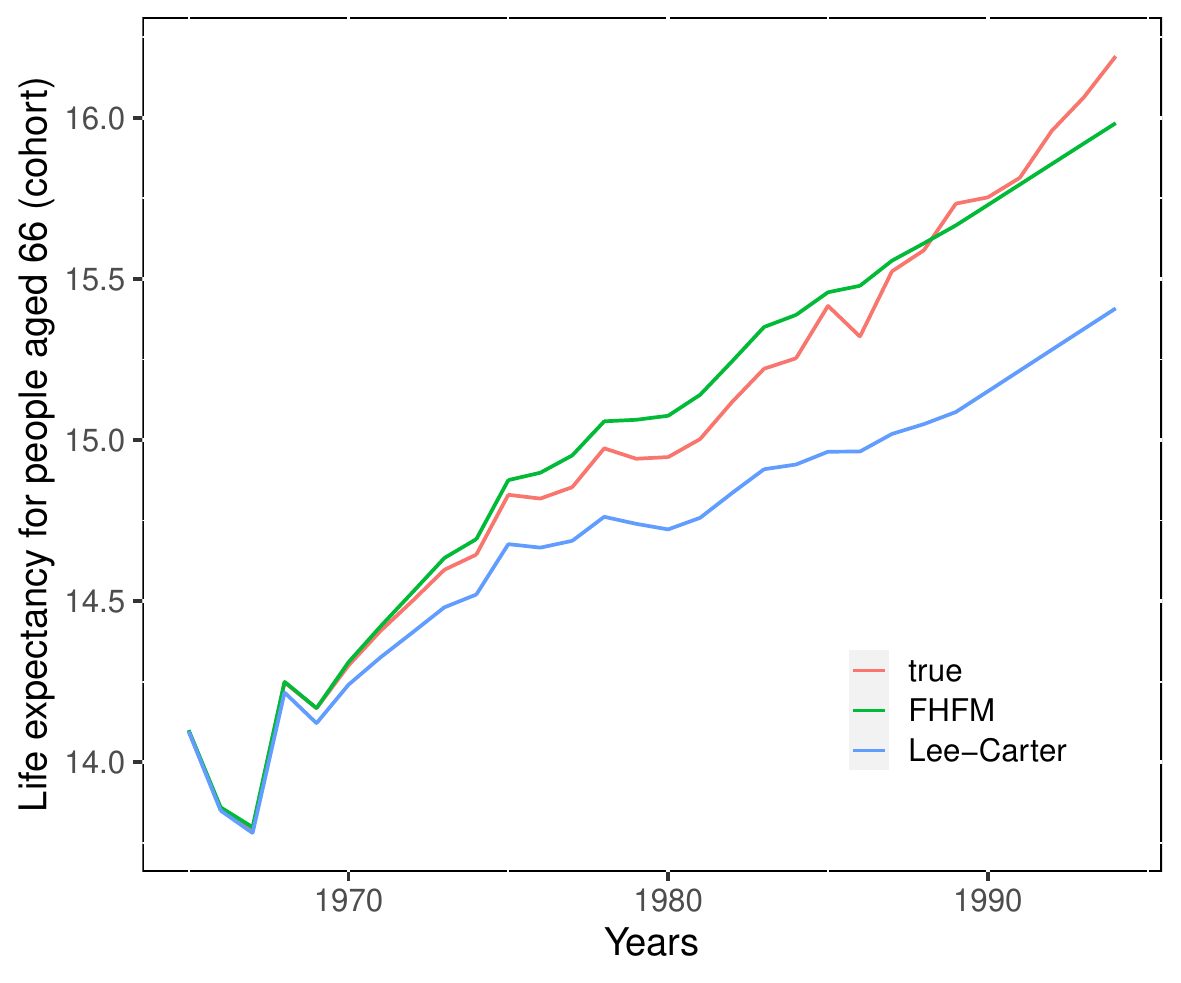}
\caption{Comparison of the predicted life expectancies from the FHFM and Lee-Carter with the true values (cohort)}
\label{dpca:fg:fa1}
\end{minipage}
\hspace*{0.01cm}
\begin{minipage}[t]{0.48\textwidth}
\includegraphics[width=\linewidth]{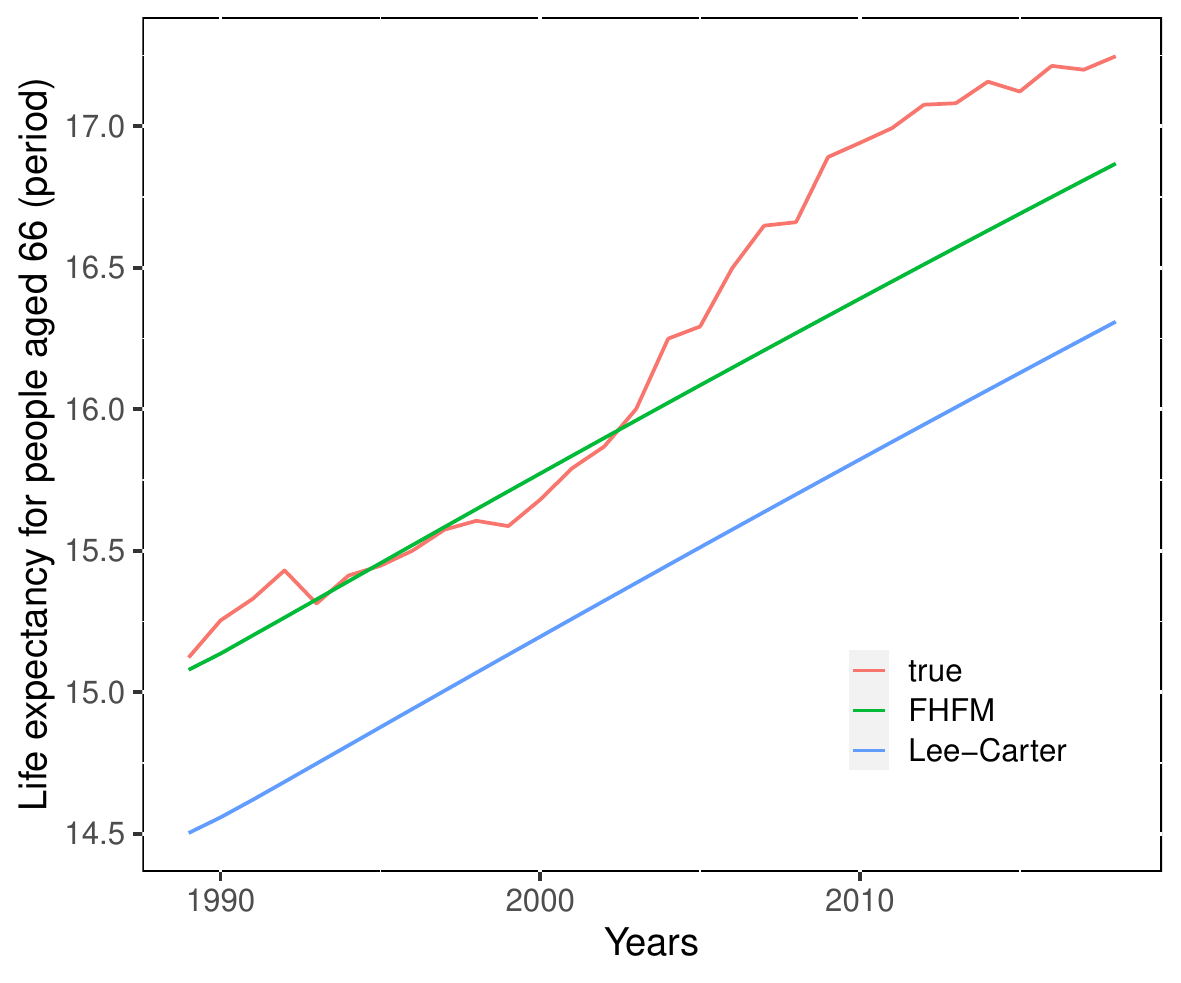}
\caption{Comparison of the predicted life expectancies from the FHFM and Lee-Carter with the true values (period)}
\label{dpca:fg:fa2}	
\end{minipage}
\end{figure}

Table \ref{dpca:tb:ta2} exhibits the life expectancies (cohort and period) and the present values of annuities with annual payment $\$1$ and interest rate $2\%$ for some selected ages and years (for some years and ages there are no forecast involves, hence we use $*$ to mark them). We can see that for the life expectancies, almost all the values from the FHFM are closer to the true values than those from the Lee-Carter method. The Lee-Carter method tends to price much lower than the empirical true values. The absolute pricing errors are around $\$0.20$ to $\$0.40$ per $\$1$ of the life annuity. On the other hand, FHFM provides very accurately pricing with a maximum $\$0.02$ error per $\$1$ annual payment. Although the difference looks very small, it is indeed a big risk for life insurers or social security. To illustrate the financial impact on the industry, we can consider the pricing for individuals aged $45$ in year $1970$. The price from the Lee-Carter method is $\$0.27$ lower and from the FHFM is $\$0.01$ lower per $\$1$ compared with the empirical true price. Suppose the annual payment for an individual is $\$20,000$ and the number of people purchased this insurance is $5,000$. Then according to the Lee-Carter method, the insurance company will have a $\$27$ million shortfall ($27 \text{million} = 0.27 \times 20000 \times 5000$), which is a huge risk. On the other hand, FHFM will only have a $\$1$ million shortfall. Although it also mis-priced the insurance, this amount of shortfall is much less risky for the company or the social security. In summary, our method improves the estimating of life expectancies and prices the life annuities more accurately by forecasting the mortality rates better. 

\begin{table}[!htbp] \centering 
  \footnotesize
  \caption{Selected life expectancies (cohort and period) and the present values of annuities (annual payment $\$1$ and interest rate $2\%$)} 
  \label{dpca:tb:ta2} 
\begin{tabularx}{\textwidth}{c *{9}{Y}}
\toprule
 & \multicolumn{3}{c}{period life expectancy} 
 & \multicolumn{3}{c}{cohort life expectancy}
 & \multicolumn{3}{c}{pv of life annuity}\\
\cmidrule(lr){2-4} \cmidrule(lr){5-7} \cmidrule(l){8-10}
 (year, age) & true & Lee-Carter & FHFM & true & Lee-Carter & FHFM & true & Lee-Carter & FHFM \\
\midrule
(1950, 25) & $45.81$ & * & * & $50.12$ & $49.60$ & $50.09$ & $5.72$ & $5.53$ & $5.71$ \\ 
(1960, 35) & $37.59$ & * & * & $40.84$ & $40.32$ & $40.81$ & $6.97$ & $6.74$ & $6.96$ \\ 
(1970, 45) & $29.19$ & * & * & $31.97$ & $31.43$ & $31.94$ & $8.49$ & $8.22$ & $8.48$ \\ 
(1980, 55) & $22.59$ & * & * & $23.85$ & $23.27$ & $23.81$ & $10.36$ & $10.02$ & $10.34$ \\ 
(1990, 65) & $15.95$ & $15.23$ & $15.83$ & $16.50$ & $15.87$ & $16.47$ & $12.62$ & $12.21$ & $12.60$ \\ 
(2000, 75) & $9.59$ & $9.40$ & $9.87$ & $10.05$ & $9.63$ & $10.09$ & $8.61$ & $8.29$ & $8.63$ \\ 
\bottomrule
\end{tabularx}
\end{table}

\section{Conclusion}
\label{dpca:conclu}
  Motivated by the unique feature of human mortality data, we propose a new forecast-driven hierarchical factor model  for mortality forecasting. 
 
To the best of our knowledge, this is the first  forecast-driven factor model for mortality forecasting. The proposed FHFM combines both the dynamic and static features together via a two-step procedure, with the first step extracting factors with the most predictability and the second step extracting factors capturing the largest variations. Theoretically, we show that these two types of factors are both indispensable in producing optimal forecasting for mortality data.

 Empirical analysis using the US mortality data shows that the proposed new method can improve the out-of-sample forecasting performance  compared to the benchmark methods, especially in the long-term. Moreover, we show that the pricing error of a life annuity project using the Lee-Carter method is around \$$0.154$ per $\$1$ annual payment; however, the pricing error using the proposed FHFM  is only about $\$0.041$ per $\$1$, which is a remarkable improvement. 

Like other existing mortality models, the FHFM proposed in this paper can only extract linear features from the original data.  For future research, we are interested in developing a general non-parametric factor model, which can capture non-linear features in mortality data. Modern statistical techniques, such as neural networks and non-parametric estimation methods, can be used to fulfil this target. In addition, we are  interested in extending the age range considered in this research to include more advanced ages  and exploring time-varying factor models for multi-population mortality forecasting \citep{HUANG202095,Dong2020}.

\bibliography{thesis}

@article{Dong2020,
author = {Yumo Dong and Fei Huang and Honglin Yu and Steven Haberman},
title = {Multi-population mortality forecasting using tensor decomposition},
journal = {Scandinavian Actuarial Journal},
volume = {2020},
number = {8},
pages = {754-775},
year  = {2020},
publisher = {Taylor & Francis},
doi = {10.1080/03461238.2020.1740314},

URL = { 
        https://doi.org/10.1080/03461238.2020.1740314
    
},
eprint = { 
        https://doi.org/10.1080/03461238.2020.1740314
    
}

}

@article{HUANG202095,
title = {Modelling life tables with advanced ages: An extreme value theory approach},
journal = {Insurance: Mathematics and Economics},
volume = {93},
pages = {95-115},
year = {2020},
issn = {0167-6687},
doi = {https://doi.org/10.1016/j.insmatheco.2020.04.004},
url = {https://www.sciencedirect.com/science/article/pii/S0167668720300482},
author = {Fei Huang and Ross Maller and Xu Ning}
}

@article{HE202114,
title = {Mortality forecasting using factor models: Time-varying or time-invariant factor loadings?},
journal = {Insurance: Mathematics and Economics},
volume = {98},
pages = {14-34},
year = {2021},
issn = {0167-6687},
doi = {https://doi.org/10.1016/j.insmatheco.2021.01.006},
url = {https://www.sciencedirect.com/science/article/pii/S0167668721000159},
author = {Lingyu He and Fei Huang and Jianjie Shi and Yanrong Yang}
}

@book{MR2036084,
	Author = {Jolliffe, I. T.},
	Edition = {Second},
	Isbn = {0-387-95442-2},
	Mrclass = {62-02 (62H25)},
	Mrnumber = {2036084},
	Pages = {xxx+487},
	Publisher = {Springer-Verlag, New York},
	Series = {Springer Series in Statistics},
	Title = {Principal component analysis},
	Year = {2002}}

@article{LYB2011,
	Author = {Lam, C. and Yao, Q. and Bathia, N.},
	Journal = {Biometrika},
	Number = {4},
	Pages = {901-918},
	Title = {Estimation of latent factors for high-dimensional time series},
	Volume = {98},
	Year = {2011}}

@article{chang2018principal,
	Author = {Chang, Jinyuan and Guo, Bin and Yao, Qiwei},
	Journal = {The Annals of Statistics},
	Number = {5},
	Pages = {2094--2124},
	Publisher = {Institute of Mathematical Statistics},
	Title = {Principal component analysis for second-order stationary vector time series},
	Volume = {46},
	Year = {2018}}

@article{lee1992modeling,
	Author = {Lee, Ronald D and Carter, Lawrence R},
	Journal = {Journal of the American Statistical Association},
	Number = {419},
	Pages = {659--671},
	Publisher = {Taylor \& Francis},
	Title = {Modeling and forecasting {US} mortality},
	Volume = {87},
	Year = {1992}}

@book{hollmann1999methodology,
	Author = {Hollmann, Frederick William and Mulder, Tammany J and Kallan, Jeffrey E},
	Publisher = {US Department of Commerce, Bureau of the Census, Population Division~{\ldots}},
	Title = {Methodology \& Assumptions for the Population Projections of the United States: 1999 to 2010},
	Year = {1999}}

@article{booth2002applying,
	Author = {Booth, Heather and Maindonald, John and Smith, Len},
	Journal = {Population Studies},
	Number = {3},
	Pages = {325--336},
	Publisher = {Taylor \& Francis},
	Title = {Applying {L}ee-{C}arter under conditions of variable mortality decline},
	Volume = {56},
	Year = {2002}}

@article{yang2010modeling,
	Author = {Yang, Sharon S and Yue, Jack C and Huang, Hong-Chih},
	Journal = {Insurance: Mathematics and Economics},
	Number = {1},
	Pages = {254--270},
	Publisher = {Elsevier},
	Title = {Modeling longevity risks using a principal component approach: A comparison with existing stochastic mortality models},
	Volume = {46},
	Year = {2010}}

@article{booth2008mortality,
	Author = {Booth, Heather and Tickle, Leonie},
	Journal = {Annals of Actuarial Science},
	Number = {1-2},
	Pages = {3--43},
	Publisher = {Cambridge University Press},
	Title = {Mortality modelling and forecasting: A review of methods},
	Volume = {3},
	Year = {2008}}

@article{Cairns2009,
	Author = {Andrew J. G. Cairns and David Blake and Kevin Dowd and Guy D. Coughlan and David Epstein and Alen Ong and Igor Balevich},
	Journal = {North American Actuarial Journal},
	Number = {1},
	Pages = {1-35},
	Publisher = {Routledge},
	Title = {A quantitative comparison of stochastic mortality models using data from {E}ngland and {W}ales and the {U}nited {S}tates},
	Volume = {13},
	Year = {2009}}

@article{currie2004smoothing,
	Author = {Currie, Iain D and Durban, Maria and Eilers, Paul HC},
	Journal = {Statistical Modelling},
	Number = {4},
	Pages = {279--298},
	Publisher = {Sage Publications Sage CA: Thousand Oaks, CA},
	Title = {Smoothing and forecasting mortality rates},
	Volume = {4},
	Year = {2004}}

@misc{HMD,
	Author = {{University of California, Berkeley (USA) and Max Planck Institute for Demographic Research (Germany)}},
	Note = {\url{www.mortality.org} [Accessed: 2018.07.10]},
	Title = {\emph{Human Mortality Database}}}

@article{hyndman2007robust,
	Author = {Hyndman, Rob J and Ullah, Md Shahid},
	Journal = {Computational Statistics \& Data Analysis},
	Number = {10},
	Pages = {4942--4956},
	Publisher = {Elsevier},
	Title = {Robust forecasting of mortality and fertility rates: a functional data approach},
	Volume = {51},
	Year = {2007}}

@article{AH2013,
	Author = {Ahn, S. C. and Horensten, A. R.},
	Journal = {Econometrica},
	Pages = {1203-1227},
	Title = {Eigenvalue ratio test for the number of factors},
	Volume = {81},
	Year = {2013}}

@article{LY2012,
	Author = {Lam, C. and Yao, Q.},
	Journal = {The Annals of Statistics},
	Number = {2},
	Pages = {694-726},
	Title = {Factor modeling for high-dimensional time series: inference for the number of factors},
	Volume = {40},
	Year = {2012}}

@article{B2002,
	Author = {Bai, J. and Ng, S.},
	Journal = {Econometrica},
	Number = {1},
	Pages = {191-221},
	Title = {Determine the number of factors in approximate factor models},
	Volume = {70},
	Year = {2002}}

@article{FLM2013,
	Author = {Fan, J. and Liao, Y. and Mincheva, M.},
	Journal = {Journal of the Royal Statistical Society: Series B},
	Number = {4},
	Pages = {603-680},
	Title = {Large covariance estimation by thresholding principal orthogonal complements},
	Volume = {75},
	Year = {2013}}

@article{HKH2015,
	Author = {H{\"o}rmann , S. and Kidzi{\'n}ski, {\L}. and Hallin, M.},
	Journal = {J. R. Stat. Soc. B},
	Pages = {319-348},
	Title = {Dynamic functional principal components},
	Volume = {77},
	Year = {2015}}

@book{BL1975,
	Author = {Brillinger, David R.},
	Publisher = {Holt, Rinehart, and Winston.},
	Title = {Time Series: Data Analysis and Theory},
	Year = {1975}}

@book{Anderson2003,
	Author = {Anderson, T. W.},
	Edition = {3},
	Publisher = {Wiley},
	Title = {An Introduction to Multivariate Statistical Analysis},
	Year = {2003}}

@article{cairns2006two,
	Author = {Cairns, Andrew JG and Blake, David and Dowd, Kevin},
	Journal = {Journal of Risk and Insurance},
	Number = {4},
	Pages = {687--718},
	Publisher = {Wiley Online Library},
	Title = {A two-factor model for stochastic mortality with parameter uncertainty: theory and calibration},
	Volume = {73},
	Year = {2006}}

@article{renshaw2003lee,
	Author = {Renshaw, Arthur and Haberman, Steven},
	Journal = {Journal of the Royal Statistical Society: Series C (Applied Statistics)},
	Number = {1},
	Pages = {119--137},
	Publisher = {Wiley Online Library},
	Title = {Lee--Carter mortality forecasting: A parallel generalized linear modelling approach for England and Wales mortality projections},
	Volume = {52},
	Year = {2003}}

@article{warshawsky1988private,
	Author = {Warshawsky, Mark},
	Journal = {Journal of Risk and Insurance},
	Pages = {518--528},
	Publisher = {JSTOR},
	Title = {Private annuity markets in the United States: 1919-1984},
	Year = {1988}}

@book{mccarthy2001assessing,
	Author = {McCarthy, David and Mitchell, Olivia S},
	Publisher = {Pension Research Council, the Wharton School, University of Pennsylvania},
	Title = {Assessing the impact of mortality assumptions on annuity valuation: Cross-country evidence},
	Year = {2001}}

@book{usreport,
	Author = {{The Board of Trustees of the Federal OASDI Trust Funds}},
	Publisher = {U.S. government publishing office},
	Title = {The 2019 annual report of the Board of Trustees of the Federal Old-Age and Survivors Insurance and Federal Disability Insurance Trust Funds},
	Year = {2019}}

@book{cunningham2012models,
	Author = {Cunningham, R.J. and Herzog, T.N. and London, R.L.},
	Isbn = {9781566989336},
	Lccn = {2012044398},
	Publisher = {Actex Publications},
	Series = {ACTEX Academic series},
	Title = {Models for Quantifying Risk},
	Url = {https://books.google.com.au/books?id=DpC3ojjaRKQC},
	Year = {2012}}

@misc{perlifetable,
	Author = {{Social Security Administration}},
	Note = {[\url{https://www.ssa.gov/oact/HistEst/PerLifeTables/2019/PerLifeTables2019.html}; accessed 5-June-2019]},
	Title = {Period Life Tables},
	Url = {https://www.ssa.gov/oact/HistEst/PerLifeTables/2019/PerLifeTables2019.html},
	Year = {2019}}

@article{BT1977,
 title = {A Canonical Analysis of Multiple Time Series},
 author = {G. E. P. Box and G. C. Tiao},
 journal = {Biometrika},
 volume = {64},
 number = {2},
 pages = {355--365},
 year = {1977},
 publisher = {Oxford University Press},
}

@article{HJLTZ2021,
author = {Huang, Dashan and Jiang, Fuwei and Li, Kunpeng and Tong, Guoshi and Zhou, Guofu},
title = {Scaled PCA: A New Approach to Dimension Reduction},
journal = {Management Science},
volume = {0},
number = {0},
pages = {null},
year = {2021},
doi = {10.1287/mnsc.2021.4020},
}

@article{MNP2013,
    author = {Moench, Emanuel and Ng, Serena and Potter, Simon},
    title = "{Dynamic Hierarchical Factor Models}",
    journal = {The Review of Economics and Statistics},
    volume = {95},
    number = {5},
    pages = {1811-1817},
    year = {2013},
}

\appendix

\section{Appendix A: Additional Simulations}
\label{append:cha:stepPCA.1}
In this appendix, we provide more simulation studies as a supplementary to Section \ref{dpca:simu}. Specially, \textit{Example} 5 and 6 are special cases to the examples in Section \ref{dpca:simu}, in which the two different types of features are completely separated.

Firstly, in \textit{Example} 4 and 5, we show that the first step of our method extracts features with strong time-serial dependence, which indicates a powerful forecasting ability. In addition, the low-dimensional representation has relative small reconstruction errors, which is necessary for recovering the data.

Secondly, with \textit{Example} 5 and 6, we show that our method performs better on forecasting compared to static PCA and dynamic PCA.

For descriptive convenience, we use ``FHFM'' to represent our method, ``CPCA'' to represent the static PCA which was described in Section \ref{dpca:relation}, and  ``DPCA'' to represent the dynamic PCA described in Section \ref{dpca:relation} with $\ell_0 = 1$. We also consider comparing with the method given in \citet{LYB2011}, and we use ``DPCA$(\ell)$'' to represent it, with $\ell = 1, 5, 10$. ``DPCA$(1)$'' is the same with the first step in our method, while ``DPCA$(5)$'' and ``DPCA$(10)$'' aggregate more auto-covariances but exclude the variance matrix. 

\subsection*{Data generating}

\begin{itemize}
\item \textit{Example 4}\\
$\{\bby_t\}_{t = 1, 2, \dots, T}: P \times 1$ is generated by 
\begin{align*}
 \bby_t = \bba + \bbb k_{t} + \bvar_t,
\end{align*}
where $\bba$ is a $P \times 1$ mean vector with elements independently generated from standard normal distribution, $\bbb$ is a $P \times 1$ vector obtained by the first column of a QR decomposition of a random generated matrix, $\{k_t\}_{t =1,2, \dots, T}$ is generated from $AR(1)$ model with coefficient $0.7$ and mean $0$, and $\bvar_t$ is a $P \times 1$ error term with elements independently generated from standard normal distribution. 

\item \textit{Example 5}\\
Construct $\bby_t =  (\bby_t^{(1)\top}, \bby_t^{(2)\top})^{\top}$.
$\{\bby_t^{(1)}\}_{t = 1, 2, \dots, T}: (dP) \times 1$ is generated by
\begin{align*}
\bby_t^{(1)} = \bbb k_{t} + \bvar_t^{(1)}.
\end{align*}
$\bbb$ is a $(dP) \times 1$ vector with elements generated from $U(0, 1)$, $\{k_t\}_{t =1,2, \dots, T}$ is generated from $AR(1)$ model with coefficient $0.8$, and $\bvar_t^{(1)}$ is a $(dP) \times 1$ error term with elements independently generated from $N(0, 0.2)$.
$\{\bby_t^{(2)}\}_{t = 1, 2, \dots, T}: ((1-d)P) \times 1$ is generated by 
\begin{align*}
\bby_t^{(2)} =  \bba w_{t} + \bvar_t^{(2)},
\end{align*}
where $\bba$ is a $((1-d)P) \times 1$ vector with elements generated from $U(0, 1)$, $\{w_t\}_{t =1,2, \dots, T}$ is generated from $N(0, 1.5)$, and $\bvar_t^{(2)}$ is a $((1-d)P) \times 1$ error term with elements independently generated from $N(0, 0.2)$. 
 
We call $\{\bby_t^{(1)}\}_{t = 1, 2, \dots, T}$ the dependent part as it has autocorrelations within observations (time dimension), and $\{\bby_t^{(2)}\}_{t = 1, 2, \dots, T}$ the independent part as it has independent generated observations, while the variance of it is larger than that of $\{\bby_t^{(1)}\}_{t = 1, 2, \dots, T}$. The parameter $d$ is the proportion of the dependent part among the whole dataset, with possible values among $(0,1)$.
As a result of this design, the whole data consists of two part. The dependent part has strong serial dependence with relatively small variance and the independent part has very weak dependence with relatively large variance. This is a special case for the FHFM in example 1 to 3 in which we can set $0$s in the coefficient vector $\bba$ and $\bbb$ to get the \textit{Example 5}.
\item \textit{Example 6}\\
The data structure is the same with \textit{Example} 5 with $d = 0.4$, except that $\{k_t\}_{t =1,2, \dots, T}$ is generated from $AR(1)$ model with coefficient $0.7$, $\bvar_t^{(1)}$ and $\bvar_t^{(2)}$ are $(dP) \times 1$ error terms with elements independently generated from $N(0, 0.5)$, and $\{w_t\}_{t =1,2, \dots, T}$ is generated from $N(0, 3)$. The main difference is we enlarge the variations in \textit{Example} 5, comparing with \textit{Example} 6. The purpose is to show that keeping sufficient information of the variation is necessary. 
\end{itemize}

Because in \textit{Example} 5 and 6, $\bby_t$ consists of the dependent part and independent part, we also report the FRMSE for the two parts separately, in addition to the overall FRMSE defined in Section \ref{dpca:simu}. Rewrite $\widehat{\bby}_{T-i}$ as $(\widehat{\bby}_{T-i}^{(1)\top}, \widehat{\bby}_{T-i}^{(2)\top})^{\top}$ and $\bby_{T-i}$ as $(\bby_{T-i}^{(1)\top}, \bby_{T-i}^{(2)\top})^{\top}$, then:
\begin{align*}
\text{Dependent FRMSE}(h) = \left(\frac{\sum_{i = 0}^{h-1}\|\widehat{\bby}^{(1)}_{T-i} - \bby^{(1)}_{T-i}\|_2^2}{hPd}\right)^{1/2},\\
\text{Independent FRMSE}(h) = \left(\frac{\sum_{i = 0}^{h-1}\|\widehat{\bby}^{(2)}_{T-i} - \bby^{(2)}_{T-i}\|_2^2}{hP(1-d)}\right)^{1/2},
\end{align*}  
where $d$ is the proportion of $\bby^{(1)}_t$ among $\bby_t$. 

\subsection*{Results}
We try different sets of $(P, T)$: $(50, 50),\ (50, 100),\ (100, 100),\ (100, 200),\ (200,200)$, as we would like to evaluate the performance under situations $P$ and $T$ are comparable. The results are showed in Table \ref{dpca:append:ts.kt} to Table \ref{dpca:append:ts.forecast.e6}.
\begin{table}[!htbp] \centering 
  \footnotesize
  \caption{Variance and Dependence of $\widehat{k}_t$} 
  \label{dpca:append:ts.kt} 
\begin{tabularx}{\textwidth}{c *{9}{Y}}
\toprule
 & \multicolumn{3}{c}{Time variance ($\widehat{k_t}$)} 
 & \multicolumn{3}{c}{Time dependence ($\widehat{k_t}$)}
 & \multicolumn{3}{c}{Mix ($\widehat{k_t}$)}\\
\cmidrule(lr){2-4} \cmidrule(lr){5-7} \cmidrule(l){8-10}
 $(P, T)$ & CPCA & DPCA & FHFM & CPCA & DPCA & FHFM & CPCA & DPCA & FHFM \\
\midrule
& \multicolumn{9}{c}{Example 4}\\
\midrule
$(50, 50)$ & $\bf{7.882}$ & $7.723$ & $6.073$ & $1.201$ & $1.607$ & $\bf{1.833}$ & $9.083$ & $\bf{9.331}$ & $7.905$ \\ 
$(50, 100)$ & $\bf{6.001}$ & $5.917$ & $4.595$ & $1.061$ & $1.328$ & $\bf{1.455}$ & $7.062$ & $\bf{7.245}$ & $6.050$ \\ 
$(100, 100)$ & $\bf{7.968}$ & $7.811$ & $6.094$ & $0.957$ & $1.392$ & $\bf{1.671}$ & $8.925$ & $\bf{9.204}$ & $7.765$ \\ 
$(100, 200)$ & $\bf{5.996}$ & $5.911$ & $4.567$ & $0.944$ & $1.241$ & $\bf{1.409}$ & $6.940$ & $\bf{7.151}$ & $5.976$ \\ 
$(200, 200)$ & $\bf{7.974}$ & $7.814$ & $6.088$ & $0.757$ & $1.130$ & $\bf{1.399}$ & $8.731$ & $\bf{8.943}$ & $7.487$ \\
\midrule
& \multicolumn{9}{c}{Example 5 ($d = 0.5$)}\\
\midrule 
$(50, 50)$ & $\bf{24.189}$ & $23.935$ & $21.005$ & $11.916$ & $13.405$ & $\bf{15.291}$ & $36.105$ & $\bf{37.340}$ & $36.296$ \\ 
$(50, 100)$ & $\bf{23.858}$ & $23.494$ & $21.389$ & $12.132$ & $14.311$ & $\bf{16.480}$ & $35.990$ & $37.805$ & $3\bf{7.869}$ \\ 
$(100, 100)$ & $\bf{47.927}$ & $47.217$ & $43.802$ & $25.852$ & $30.221$ & $\bf{33.975}$ & $73.780$ & $77.439$ & $\bf{77.777}$ \\ 
$(100, 200)$ & $\bf{47.730}$ & $46.898$ & $45.228$ & $27.808$ & $33.306$ & $\bf{35.603}$ & $75.539$ & $80.204$ & $\bf{80.831}$ \\ 
$(200, 200)$ & $\bf{94.918}$ & $93.398$ & $89.976$ & $55.774$ & $66.264$ & $\bf{70.893}$ & $150.692$ & $159.662$ & $\bf{160.869}$ \\
\bottomrule
\end{tabularx}
\end{table}

From Table \ref{dpca:append:ts.kt}, we can see that the CPCA provides feature $\widehat{k}_t$ with the largest variance, while the first step of our method (FHFM) captures $\widehat{k}_t$ with the largest lag $1$ auto-covariance.
In addition, in \textit{Example} $5$, our method has slightly larger Mix$(\widehat{k}_t)$ than DPCA, which shows that under certain data structure the dimension reduction of our first step is enough to represent sufficient information.

From Table \ref{dpca:append:ts.error} and \ref{dpca:append:ts.error.depend}, we can see that our method always provides the error terms with the smallest time and cross-sectional variance and dependence.
\begin{table}[!htbp] \centering 
  \footnotesize
  \caption{Variance across Time and Ages of error terms} 
  \label{dpca:append:ts.error} 
\begin{tabularx}{\textwidth}{c *{6}{Y}}
\toprule
 & \multicolumn{3}{c}{Time Variance ($\widehat{\bvar}_{\cdot t}$)} 
 & \multicolumn{3}{c}{Cross-sectional Variance ($\widehat{\bvar}_{p\cdot}$)}\\
\cmidrule(lr){2-4} \cmidrule(l){5-7}
 $(P, T)$ & CPCA & DPCA & FHFM & CPCA & DPCA & FHFM \\
\midrule
 & \multicolumn{6}{c}{Example 4}\\
\midrule
$(50, 50)$ & $0.778$ & $0.680$ & $\bf{0.337}$ & $0.794$ & $0.694$ & $\bf{0.343}$ \\ 
$(50, 100)$ & $0.775$ & $0.706$ & $\bf{0.373}$ & $0.782$ & $0.713$ & $\bf{0.377}$ \\ 
$(100, 100)$ & $0.796$ & $0.694$ & $\bf{0.351}$ & $0.804$ & $0.701$ & $\bf{0.355}$ \\ 
$(100, 200)$ & $0.778$ & $0.708$ & $\bf{0.381}$ & $0.782$ & $0.712$ & $\bf{0.383}$ \\ 
$(200, 200)$ & $0.803$ & $0.701$ & $\bf{0.357}$ & $0.807$ & $0.704$ & $\bf{0.359}$ \\
\midrule
 & \multicolumn{6}{c}{Example 5 ($d = 0.5$)}\\
\midrule
$(50, 50)$ & $0.071$ & $0.091$ & $\bf{0.037}$ & $0.094$ & $0.123$ & $\bf{0.037}$ \\ 
$(50, 100)$ & $0.056$ & $0.080$ & $\bf{0.038}$ & $0.067$ & $0.104$ & $\bf{0.038}$ \\ 
$(100, 100)$ & $0.055$ & $0.077$ & $\bf{0.039}$ & $0.064$ & $0.100$ & $\bf{0.039}$ \\ 
$(100, 200)$ & $0.046$ & $0.074$ & $\bf{0.039}$ & $0.050$ & $0.094$ & $\bf{0.039}$ \\ 
$(200, 200)$ & $0.046$ & $0.072$ & $\bf{0.039}$ & $0.051$ & $0.091$ & $\bf{0.039}$ \\
\bottomrule
\end{tabularx}
\end{table}

Table \ref{dpca:append:ts.forecast.1} and \ref{dpca:append:ts.forecast.5}, show the $1$ step and $5$ steps ahead forecasting root mean square errors for \textit{Example} $5$ with $d = 0.5, 0.4, 0.3$, respectively. Overall, FHFM performs better than the other two, as it has the smallest overall FRMSE for all the cases. Checking the Dependent FRMSE and Independent FRMSE separately, we can find that FHFM performs even better for the dependent part. As $d$ decreasing, the Dependent FRMSE of FHFM increases the least, although all the Dependent FRMSEs increase. And for the independent part, FHFM is better when $d = 0.5, 0.4$, but performs almost the same with others when $d = 0.3$. This result shows that FHFM extracts features with more forecasting power from the dependent part and uses it to help improve the forecasting of the independent part.
However, when the proportion of the dependent part is small, such as $d = 0.3$, forecasting the independent part cannot be blessed that much from the dependent part. Therefore, there will be very little difference among the three methods when comparing the performance for the independent part when $d= 0.3$. But FHFM will always provide better forecasting for the dependent part, which leads to better overall forecasting for all cases.
\begin{table}[!htbp] \centering 
  \footnotesize
  \caption{Covariance across Time and Ages of error terms} 
  \label{dpca:append:ts.error.depend} 
\begin{tabularx}{\textwidth}{c *{6}{Y}}
\toprule
 & \multicolumn{3}{c}{Time dependence ($\widehat{\bvar}_{\cdot t}$)} 
 & \multicolumn{3}{c}{Cross-sectional dependence ($\widehat{\bvar}_{p\cdot}$)}\\
\cmidrule(lr){2-4} \cmidrule(l){5-7}
 $(P, T)$ & CPCA & DPCA & FHFM & CPCA & DPCA & FHFM \\
\midrule
 & \multicolumn{6}{c}{Example 4}\\
\midrule
$(50, 50)$ & $0.107$ & $0.097$ & $\bf{0.059}$ & $0.108$ & $0.099$ & $\bf{0.060}$ \\ 
$(50, 100)$ & $0.106$ & $0.099$ & $\bf{0.064}$ & $0.086$ & $0.082$ & $\bf{0.058}$ \\ 
$(100, 100)$ & $0.077$ & $0.070$ & $\bf{0.043}$ & $0.077$ & $0.070$ & $\bf{0.043}$ \\ 
$(100, 200)$ & $0.075$ & $0.070$ & $\bf{0.046}$ & $0.061$ & $0.058$ & $\bf{0.041}$ \\ 
$(200, 200)$ & $0.054$ & $0.050$ & $\bf{0.031}$ & $0.055$ & $0.050$ & $\bf{0.031}$ \\ 
\midrule
 & \multicolumn{6}{c}{Example 5 ($d = 0.5$)}\\
\midrule
$(50, 50)$ & $0.025$ & $0.037$ & $\bf{0.004}$ & $0.032$ & $0.045$ & $\bf{0.004}$ \\ 
$(50, 100)$ & $0.016$ & $0.030$ & $\bf{0.004}$ & $0.015$ & $0.032$ & $\bf{0.003}$ \\ 
$(100, 100)$ & $0.013$ & $0.027$ & $\bf{0.003}$ & $0.014$ & $0.030$ & $\bf{0.003}$ \\ 
$(100, 200)$ & $0.007$ & $0.025$ & $\bf{0.003}$ & $0.007$ & $0.025$ & $\bf{0.002}$ \\ 
$(200, 200)$ & $0.007$ & $0.023$ & $\bf{0.002}$ & $0.007$ & $0.024$ & $\bf{0.002}$ \\ 
\bottomrule
\end{tabularx}
\end{table}

Table \ref{dpca:append:ts.forecast.e6} shows the $1$ step and $5$ steps ahead forecasting root mean square error of FHFM compared to DPCA$(\ell)$, $\ell = 1,5,10$, for \textit{Example} $6$. The reason for comparing DPCA$(\ell)$ separately is that it contains different information. The DPCA we compared with in \textit{Example} $5$ involves the same information (variance and lag $1$ auto-covariance of $\bby_t$) with FHFM, while DPCA$(\ell)$ aggregates more dependent information (lag $1$ to lag $\ell$ auto-covariances) but discards the variance $var\left({\bby_t}\right)$. In addition, DPCA$(1)$ is equivalent to only conduct the first step of FHFM. 
\begin{table}[!htbp] \centering 
  \footnotesize
  \caption{1 Step Ahead Forecasting RMSE} 
  \label{dpca:append:ts.forecast.1} 
\begin{tabularx}{\textwidth}{c *{9}{Y}}
\toprule
 & \multicolumn{3}{c}{Dependent FRMSE$(1)$} 
 & \multicolumn{3}{c}{Independent FRMSE$(1)$}
 & \multicolumn{3}{c}{Overall FRMSE$(1)$}\\
\cmidrule(lr){2-4} \cmidrule(lr){5-7} \cmidrule(l){8-10}
 $(P, T)$ & CPCA & DPCA & FHFM & CPCA & DPCA & FHFM & CPCA & DPCA & FHFM \\
\midrule
 & \multicolumn{9}{c}{Example 5 ($d = 0.5$)}\\
\midrule
$(50, 50)$ & $0.623$ & $0.608$ & $\bf{0.568}$ & $0.771$ & $0.768$ & $\bf{0.749}$ & $0.757$ & $0.748$ & $\bf{0.716}$ \\ 
$(50, 100)$ & $0.573$ & $0.567$ & $\bf{0.550}$ & $0.746$ & $0.747$ & $\bf{0.743}$ & $0.715$ & $0.711$ & $\bf{0.700}$ \\ 
$(100, 100)$ & $0.568$ & $0.558$ & $\bf{0.536}$ & $0.760$ & $0.746$ & $\bf{0.734}$ & $0.719$ & $0.705$ & $\bf{0.688}$ \\ 
$(100, 200)$ & $0.572$ & $0.568$ & $\bf{0.553}$ & $0.782$ & $0.772$ & $\bf{0.764}$ & $0.732$ & $0.725$ & $\bf{0.711}$ \\ 
$(200, 200)$ & $0.537$ & $0.531$ & $\bf{0.521}$ & $0.759$ & $0.752$ & $\bf{0.746}$ & $0.705$ & $0.696$ & $\bf{0.686}$ \\
\midrule
 & \multicolumn{9}{c}{Example 5 ($d = 0.4$)}\\
\midrule
$(50, 50)$ & $0.654$ & $0.642$ & $\bf{0.584}$ & $\bf{0.738}$ & $0.743$ & $0.741$ & $0.759$ & $0.756$ & $\bf{0.731}$ \\ 
$(50, 100)$ & $0.603$ & $0.595$ & $\bf{0.553}$ & $0.727$ & $0.731$ & $\bf{0.712}$ & $0.732$ & $0.731$ & $\bf{0.699}$ \\ 
$(100, 100)$ & $0.631$ & $0.607$ & $\bf{0.562}$ & $0.767$ & $0.768$ & $\bf{0.754}$ & $0.766$ & $0.755$ & $\bf{0.727}$ \\ 
$(100, 200)$ & $0.575$ & $0.580$ & $\bf{0.542}$ & $0.778$ & $0.774$ & $\bf{0.767}$ & $0.752$ & $0.752$ & $\bf{0.729}$ \\ 
$(200, 200)$ & $0.570$ & $0.568$ & $\bf{0.537}$ & $0.739$ & $0.742$ & $\bf{0.733}$ & $0.719$ & $0.720$ & $\bf{0.701}$ \\ 
\midrule
 & \multicolumn{9}{c}{Example 5 ($d = 0.3$)}\\
\midrule
$(50, 50)$ & $0.708$ & $0.686$ & $\bf{0.598}$ & $0.761$ & $0.764$ & $\bf{0.760}$ & $0.798$ & $0.793$ & $\bf{0.760}$ \\ 
$(50, 100)$ & $0.687$ & $0.653$ & $\bf{0.572}$ & $\bf{0.749}$ & $0.752$ & $\bf{0.749}$ & $0.782$ & $0.772$ & $\bf{0.742}$ \\ 
$(100, 100)$ & $0.695$ & $0.640$ & $\bf{0.556}$ & $\bf{0.773}$ & $0.774$ & $0.775$ & $0.805$ & $0.788$ & $\bf{0.757}$ \\ 
$(100, 200)$ & $0.682$ & $0.606$ & $\bf{0.533}$ & $\bf{0.742}$ & $0.747$ & $0.743$ & $0.772$ & $0.749$ & $\bf{0.719}$ \\ 
$(200, 200)$ & $0.680$ & $0.620$ & $\bf{0.549}$ & $\bf{0.719}$ & $0.720$ & $0.720$ & $0.759$ & $0.737$ & $\bf{0.711}$ \\  
\bottomrule
\end{tabularx}
\end{table}

In Table \ref{dpca:append:ts.forecast.e6}, we can see that FHFM and DPCA$(1)$ perform better than DPCA$(5)$ and DPCA$(10)$ for most $(P, T)$ cases. This shows that involving more lagged auto-covariances does not always provide more useful information for forecasting under certain situations. The performance of DPCA$(1)$ is worse than FHFM with $(P, T) = (50, 50), (50, 100), (100, 100)$ for $1$ step ahead forecasting and $(P, T) = (100, 100)$ for $5$ step ahead forecasting, and similar for other cases. These results show that when variation is large, it is necessary to conduct the second step in the FHFM in order to achieve more accurate forecasting.

\begin{table}[!htbp] \centering 
  \footnotesize
  \caption{5 Steps Ahead Forecasting RMSE} 
  \label{dpca:append:ts.forecast.5} 
\begin{tabularx}{\textwidth}{c *{9}{Y}}
\toprule
 & \multicolumn{3}{c}{Dependent FRMSE$(5)$} 
 & \multicolumn{3}{c}{Independent FRMSE$(5)$}
 & \multicolumn{3}{c}{Overall FRMSE$(5)$}\\
\cmidrule(lr){2-4} \cmidrule(lr){5-7} \cmidrule(l){8-10}
 $(P, T)$ & CPCA & DPCA & FHFM & CPCA & DPCA & FHFM & CPCA & DPCA & FHFM \\
\midrule
 & \multicolumn{9}{c}{Example 5 ($d = 0.5$)}\\
\midrule
$(50, 50)$ & $0.846$ & $\bf{0.839}$ & $\bf{0.839}$ & $0.891$ & $0.893$ & $\bf{0.879}$ & $0.899$ & $0.896$ & $\bf{0.888}$ \\ 
$(50, 100)$ & $0.833$ & $0.834$ & $\bf{0.827}$ & $0.876$ & $0.870$ & $\bf{0.864}$ & $0.883$ & $0.881$ & $\bf{0.874}$ \\ 
$(100, 100)$ & $0.823$ & $0.817$ & $\bf{0.811}$ & $0.885$ & $0.878$ & $\bf{0.867}$ & $0.883$ & $0.876$ & $\bf{0.866}$ \\ 
$(100, 200)$ & $0.802$ & $0.798$ & $\bf{0.794}$ & $0.872$ & $0.866$ & $\bf{0.862}$ & $0.866$ & $0.861$ & $\bf{0.857}$ \\ 
$(200, 200)$ & $0.790$ & $0.795$ & $\bf{0.788}$ & $0.862$ & $0.857$ & $\bf{0.854}$ & $0.852$ & $0.852$ & $\bf{0.847}$ \\
\midrule
 & \multicolumn{9}{c}{Example 5 ($d = 0.4$)}\\
\midrule
$(50, 50)$ & $0.864$ & $0.854$ & $\bf{0.838}$ & $0.876$ & $0.877$ & $\bf{0.867}$ & $0.901$ & $0.899$ & $\bf{0.886}$ \\ 
$(50, 100)$ & $0.826$ & $0.815$ & $\bf{0.803}$ & $0.875$ & $0.877$ & $\bf{0.868}$ & $0.885$ & $0.882$ & $\bf{0.872}$ \\ 
$(100, 100)$ & $0.815$ & $0.817$ & $\bf{0.807}$ & $0.870$ & $0.872$ & $\bf{0.860}$ & $0.874$ & $0.876$ & $\bf{0.865}$ \\ 
$(100, 200)$ & $0.787$ & $0.781$ & $\bf{0.772}$ & $0.866$ & $0.865$ & $\bf{0.857}$ & $0.859$ & $0.856$ & $\bf{0.849}$ \\ 
$(200, 200)$ & $0.792$ & $0.797$ & $\bf{0.786}$ & $0.865$ & $0.868$ & $\bf{0.852}$ & $0.861$ & $0.864$ & $\bf{0.851}$ \\ 
\midrule
 & \multicolumn{9}{c}{Example 5 ($d = 0.3$)}\\
\midrule
$(50, 50)$ & $0.877$ & $0.863$ & $\bf{0.835}$ & $\bf{0.885}$ & $0.889$ & $\bf{0.885}$ & $0.910$ & $0.909$ & $\bf{0.896}$ \\ 
$(50, 100)$ & $0.860$ & $0.841$ & $\bf{0.821}$ & $\bf{0.855}$ & $0.856$ & $\bf{0.855}$ & $0.883$ & $0.878$ & $\bf{0.870}$ \\ 
$(100, 100)$ & $0.890$ & $0.865$ & $\bf{0.833}$ & $\bf{0.847}$ & $0.848$ & $\bf{0.847}$ & $0.887$ & $0.880$ & $\bf{0.868}$ \\ 
$(100, 200)$ & $0.831$ & $0.804$ & $\bf{0.771}$ & $0.854$ & $0.856$ & $\bf{0.852}$ & $0.871$ & $0.864$ & $\bf{0.851}$ \\ 
$(200, 200)$ & $0.851$ & $0.818$ & $\bf{0.801}$ & $0.867$ & $0.868$ & $\bf{0.866}$ & $0.888$ & $0.877$ & $\bf{0.869}$ \\ 
\bottomrule
\end{tabularx}
\end{table}

\begin{table}[!htbp] \centering 
  \footnotesize
  \caption{1 Step and 5 Steps Ahead Forecasting RMSE, Example 6} 
  \label{dpca:append:ts.forecast.e6} 
\begin{tabularx}{\textwidth}{c *{8}{Y}}
\toprule
 & \multicolumn{4}{c}{Overall FRMSE$(1)$} 
 & \multicolumn{4}{c}{Overall FRMSE$(5)$}\\
\cmidrule(lr){2-5} \cmidrule(l){6-9}
 $(P, T)$ & DPCA$(1)$ & DPCA$(5)$ & DPCA$(10)$ & FHFM & DPCA$(1)$ & DPCA$(5)$ & DPCA$(10)$ & FHFM \\
\midrule
$(50, 50)$ & $1.371$ & $1.379$ & $1.376$ & $\bf{1.367}$ & $1.518$ & $\bf{1.518}$ & $\bf{1.518}$ & $1.527$ \\ 
$(50, 100)$ & $1.312$ & $1.332$ & $1.332$ & $\bf{1.309}$ & $\bf{1.495}$ & $\bf{1.496}$ & $\bf{1.495}$ & $1.497$ \\ 
$(100, 100)$ & $1.389$ & $1.403$ & $1.402$ & $\bf{1.384}$ & $1.490$ & $1.490$ & $1.490$ & $\bf{1.488}$ \\ 
$(100, 200)$ & $\bf{1.349}$ & $1.371$ & $1.370$ & $\bf{1.349}$ & $\bf{1.456}$ & $1.462$ & $1.462$ & $1.456$ \\ 
$(200, 200)$ & $\bf{1.329}$ & $1.355$ & $1.355$ & $\bf{1.329}$ & $\bf{1.481}$ & $1.489$ & $1.489$ & $\bf{1.481}$ \\ 
\bottomrule
\end{tabularx}
\end{table}


\section{Appendix B: Proof of Theorems}
\label{append:cha:stepPCA.2}
This section contains proof of Theorem \ref{thm1}, as well as some lemmas that are used in these proofs. Before introducing the proofs, we provide some notations. For a $k\times k$ matrix $\bbF$, $\lambda_i(\bbF)$ indicates the $i$-th largest eigenvalue of the matrix $\bbF$. For a non-symmetric matrix $\bbS$, we use $\sigma_j\left(\bbS\right)$ to denote the singular value of the matrix $\bbS$, which corresponds to the $j$-th largest eigenvalue of the matrix $\bbS\bbS^{\top}$. 
Let $\left|\left|\bbF\right|\right|$ be the square root of the maximum eigenvalue of $\bbF\bbF^{\top}$ and $\left|\left|\bbF\right|\right|_{\min}$ be the square root of the smallest nonzero eigenvalue of the matrix $\bbF\bbF^{\top}$. The notation $a\asymp b$ means that $a=O(b)$ and $b=O(a)$. 

\subsection*{Useful Lemmas}
We will introduce four lemmas that will be used in the proofs of Theorem \ref{thm1}. Lemma \ref{lem0}, Lemma \ref{lem1} and Lemma \ref{lem2} are available results on eigenvalues of matrices under various decomposition. Lemma \ref{lem3} provides the orders of eigenvalues of the matrix $\bbL_1$ and $\bbL_2$, and the proof follows up the statement of Lemma \ref{lem3}.

\begin{lem}[Weyl's Theorem]\label{lem0}
Let $\{\lambda_i\left(\bbS\right): i=1, \ldots, P\}$ be eigenvalues of the matrix $\bbS$ in descending order and $\{\lambda_i\left(\bbJ\right): i=1, \ldots, P\}$ be eigenvalues of the matrix $\bbJ$ in descending order. Then 
\begin{eqnarray}
\left|\lambda_i\left(\bbS\right)-\lambda_i\left(\bbJ\right)\right|\leq\left|\left|\bbS-\bbJ\right|\right|. 
\end{eqnarray}

\end{lem}

\begin{lem}[Lemma S.1 of \cite{LY2012}]\label{lem1}
Let $\bbF$ be a $k\times k$ symmetric matrix such that 
\begin{eqnarray}
\bbF=\left(\begin{matrix}
\bbG & \bbH\\
\bbH^{\top} & \bbD\\
\end{matrix}
\right)
\end{eqnarray}
with $\bbG: k_1\times k_1$, $\bbD: k_2\times k_2$ and $\lambda_{k_1}\left(\bbG\right)>\lambda_1\left(\bbD\right)$. Note that $k_1+k_2=k$. Then for $1\leq j\leq k_2$ , 
\begin{eqnarray}
0\leq \lambda_j\left(\bbD\right)-\lambda_{k_1+j}\left(\bbF\right)\leq\frac{\lambda_1\left(\bbH\bbH^{\top}\right)}{\lambda_{k_1}\left(\bbG\right)-\lambda_j\left(\bbD\right)}. 
\end{eqnarray}
\end{lem}

\begin{lem}[Lemma 3 of \cite{LYB2011}]\label{lem2}
Suppose $\bbF$ and $\bbF+\bbE$ are $P\times P$ symmetric matrices and that $\bbQ=\left(\bbQ_1, \bbQ_2\right)$, where $\bbQ_1$ has size $P\times k$ and $\bbQ_2$ has size $P\times (P-k)$, is an orthogonal matrix such that span$\left(\bbQ_1\right)$ is an invariant subspace for the matrix $\bbF$, that is, $\bbF\times span\left(\bbQ_1\right)\subset span\left(\bbF\right)$. Partition the matrices $\bbQ^{\top}\bbF\bbQ$ and $\bbQ^{T}\bbE\bbQ$ as follows. 
\begin{eqnarray}
\bbQ^{\top}\bbF\bbQ=\left(\begin{matrix}
\bbD_1 & \textbf{0}\\
\textbf{0} & \bbD_2\\
\end{matrix}
\right), \ \ 
\bbQ^{\top}\bbE\bbQ=\left(\begin{matrix}
\bbE_{11} & \bbE_{21}^{\top}\\
\bbE_{21} & \bbE_{22}\\
\end{matrix}
\right). 
\end{eqnarray}
If $sep\left(\bbD_1, \bbD_2\right):=\min_{\lambda\in\Lambda\left(\bbD_1\right), \mu\in\Lambda\left(\bbD_2\right)}\left|\lambda-\mu\right|>0$, where $\Lambda\left(\bbD_1\right)$ denotes the set of eigenvalues of the matrix $\bbD_1$ and $\left|\left|\bbE\right|\right|\leq sep\left(\bbD_1, \bbD_2\right)/5$, then there exists a matrix $\bbP: (P-k)\times k$ with 
\begin{eqnarray}
\left|\left|\bbP\right|\right|\leq \frac{4\left|\left|\bbE_{21}\right|\right|}{sep\left(\bbD_1, \bbD_2\right)}
\end{eqnarray}
such that the columns of the matrix $\widehat{\bbQ}_1=\left(\bbQ_1+\bbQ_2\bbP\right)\left(\bbI+\bbP^{\top}\bbP\right)^{-1/2}$ define an orthogonal basis for a subspace that is invariant for the matrix $\bbF+\bbE$. 
\end{lem}

\begin{lem}\label{lem3}
Under Assumptions \ref{assu1}-\ref{assu8}, we have 
\begin{eqnarray}
&&\lambda_j\left(\bbL_1\right)\asymp P^{2-2\delta_1}, \ \ j=1, \ldots, r_1.\label{sy31}\\
&&\lambda_{r_1+j}\left(\bbL_1\right)\asymp P^{2-2\delta_2}, \ \ j=1 , \ldots, r_2; \label{sy32}\\ 
&&\lambda_{r_1+r_2+i}\left(\bbL_1\right)=o_p\left(P^{1-\delta_1}\right), \ \ i=1, \ldots, P-(r_1+r_2). \label{sy33}\\ 
&&\lambda_i\left(\bbL_2\right)\asymp P^{2}, \ \ i=1, \ldots, r_2. \label{sy34}
\end{eqnarray}
\end{lem}

\begin{proof}[proof of Lemma \ref{lem3}]
Recall the FHFM
\begin{eqnarray}\label{sy1}
\bby_t=\bbB\bbk_t^{(1)}+\bbA\bbk_t^{(2)}+\boldsymbol{\varepsilon}_t. 
\end{eqnarray}
From the expression (\ref{sy1}), the population covariance matrix of $\bby_t$ has the following decomposition 
\begin{eqnarray}\label{sy30}
\boldsymbol{\Sigma}_{y}(1)=\bbB\bbM_1+\bbA\bbM_2+\boldsymbol{\Sigma}_{\varepsilon}(1), 
\end{eqnarray}
where 
\begin{eqnarray*}
\bbM_1=\boldsymbol{\Sigma}_k^{(1)}(1)\bbB^{\top}+\boldsymbol{\Sigma}_k^{(12)}(1)\bbA^{\top},\ \ \ \ \bbM_2=\boldsymbol{\Sigma}_k^{(2)}(1)\bbA^{\top}+\boldsymbol{\Sigma}_k^{(21)}(1)\bbB^{\top}. 
\end{eqnarray*}
Based on Lemma \ref{lem0}, we can evaluate the $j$-th eigenvalue of $\bbL_1$ below, $j=1, \ldots, r_1$, 
\begin{align*}
\lambda_j\left(\bbL_1\right) &=\sigma^2_j\left(\boldsymbol{\Sigma}_y(1)\right)
\ge\left[\sigma_j\left(\bbB\bbM_1\right)-\sigma_1\left(\bbA\bbM_2+\boldsymbol{\Sigma}_{\varepsilon}(1)\right)\right]^2\\
&\ge \left[\sigma_j\left(\bbB\bbM_1\right)-\sigma_1\left(\bbA\bbM_2\right)-\sigma_1\left(\boldsymbol{\Sigma}_{\varepsilon}(1)\right)\right]^2\\
&= \left[\sigma_j\left(\bbM_1\right)-\sigma_1\left(\bbM_2\right)-\sigma_1\left(\boldsymbol{\Sigma}_{\varepsilon}(1)\right)\right]^2\\
&\asymp \sigma^2\left(\bbM_1\right)
\ge\left[\sigma_j\left(\boldsymbol{\Sigma}_k^{(1)}(1)\bbB^{\top}\right)-\sigma_1\left(\boldsymbol{\Sigma}_k^{(12)}(1)\bbA^{\top}\right)\right]^2\\
&\asymp \sigma^2_j\left(\boldsymbol{\Sigma}_k^{(1)}(1)\bbB^{\top}\right)
=\sigma^2_j\left(\boldsymbol{\Sigma}_k^{(1)}(1)\right)
\ge\left|\left|\boldsymbol{\Sigma}_k^{(1)}(1)\right|\right|^2_{\min}=P^{2-2\delta_1}, 
\end{align*}
where the first and second inequalities use Lemma \ref{lem0}; the second equality uses the matrices $\bbB$ and $\bbA$ being orthonormal assumed in Assumption \ref{assu1}; and the last inequality and equality both utilize Assumption \ref{assu2}. 

Hence, the first $r_1$ largest eigenvalues of the matrix $\bbL_1$ have the order of $P^{2-2\delta_1}$.

Now we consider the order of the left $p-r_1$ eigenvalues of the matrix $\bbL_1$. In terms of Weyl's inequality in Lemma \ref{lem0}, we use the eigenvalues of the matrix $\widetilde{\bbL}_1=\widetilde{\boldsymbol{\Sigma}}_y(1)\widetilde{\boldsymbol{\Sigma}}_y(1)$ to approximate the eigenvalues of $\bbL_1$, where $\widetilde{\boldsymbol{\Sigma}}_y(1)=\bbB\bbM_1+\bbA\bbM_2$. In fact, 
\begin{align}\label{sy12}
&\quad \left|\lambda_{r_1+j}\left(\bbL_1\right)-\lambda_{r_1+j}\left(\widetilde{\bbL}_1\right)\right|
\le \left|\left|\bbL_1-\widetilde{\bbL}_1\right|\right| \nonumber \\
&\leq \left|\left|\bbB\bbM_1+\bbA\bbM_2+\boldsymbol{\Sigma}_{\varepsilon}(1)\right|\right|\cdot\left|\left|\boldsymbol{\Sigma}_{\varepsilon}(1)\right|\right|
+\left|\left|\bbB\bbM_1+\bbA\bbM_2\right|\right|\cdot\left|\left|\boldsymbol{\Sigma}_{\varepsilon}(1)\right|\right| \nonumber \\
&= o\left(P^{1-\delta_1}\right), 
\end{align} 
where the last equality uses Assumption \ref{assu2}. 

Now we evaluate the order of $\lambda_{r_1+j}\left(\widetilde{\bbL}_1\right)$. Note that the rank of $\widetilde{\bbL}_1$ is no larger than $r_1+r_2$. So, when $j>r_1+r_2$, $\lambda_{r_1+j}\left(\widetilde{\bbL}_1\right)=0$. Hence, next we investigate the case of $j=1, \ldots, r_2$. 

Decompose $\widetilde{\bbL}_1$ in the following way. 
\begin{align}
\widetilde{\bbL}_1 &= \left(\begin{matrix}
\bbB & \bbA\\
\end{matrix}\right)\left(\begin{matrix}
\bbM_1\\
\bbM_2\\
\end{matrix}
\right)\left(\begin{matrix}
\bbM_1^{\top} & \bbM_2^{\top}
\end{matrix}\right)\left(\begin{matrix}
\bbB^{\top}\\
\bbA^{\top}\\
\end{matrix}
\right)\nonumber \\
&= \left(\begin{matrix}
\bbB & \bbA
\end{matrix}
\right)\left(\begin{matrix}
\bbM_1\bbM_1^{\top} & \bbM_1\bbM_2^{\top}\\
\bbM_2\bbM_1^{\top} & \bbM_2\bbM_2^{\top}\\
\end{matrix}
\right)\left(\begin{matrix}
\bbB^{\top}\\
\bbA^{\top}\\
\end{matrix}
\right). 
\end{align}
Because 
\begin{eqnarray}
\left(\begin{matrix}
\bbB^{\top}\\
\bbC^{\top}\\
\end{matrix}
\right)\left(\begin{matrix}
\bbB & \bbC
\end{matrix}
\right)=\bbI, 
\end{eqnarray}
we have $\lambda_j\left(\bbL_1\right)=\lambda_j\left(\bbM\right)$, where 
\begin{eqnarray}
\bbM=\left(\begin{matrix}
\bbM_1\bbM_1^{\top} & \bbM_1\bbM_2^{\top}\\
\bbM_2\bbM_1^{\top} & \bbM_2\bbM_2^{\top}\\
\end{matrix}
\right). 
\end{eqnarray}
It follows from Lemma \ref{lem2} and Assumption \ref{assu2} that
\begin{eqnarray}\label{sy10}
\lambda_{r_1+j}\left(\bbM\right)\leq\lambda_j\left(\bbM_2\bbM_2^{\top}\right)
\asymp P^{2-2\delta_2},
\end{eqnarray}
and
\begin{eqnarray}\label{sy11}
\lambda_{r_1+j}\left(\bbM\right)\geq\sigma_j^2\left(\bbM_2\right)-\frac{\sigma^2_1\left(\bbM_1\bbM_2^{\top}\right)}{\sigma^2_{r_1}\left(\bbM_1\right)-\sigma^2_j\left(\bbM_2\right)}\asymp P^{2-2\delta_2}, 
\end{eqnarray}
where the last $\asymp$ above uses the fact that $\sigma_{r_1}^2\left(\bbM_1\right)\asymp P^{2-2\delta_1}$, $\sigma_{j}^2\left(\bbM_2\right)\asymp P^{2-2\delta_2}$ and $\sigma_1^2\left(\bbM_1\bbM_2^{\top}\right)=O\left(P^{4-2(\delta_1+\delta_2)}\right)$. 

Combining (\ref{sy10}) and (\ref{sy11}), we can get 
\begin{eqnarray}\label{sy13}
\lambda_{r_1+j}\left(\bbM\right)\asymp P^{2-2\delta_2}, \ \ j=1, \ldots, r_2. 
\end{eqnarray}

Then it follows from (\ref{sy12}), (\ref{sy13}) and Assumption \ref{assu2} that
\begin{eqnarray}
&&\lambda_{r_1+j}\left(\bbL_1\right)\asymp P^{2-2\delta_2}, \ \ j=1 , \ldots, r_2;\\ 
&&\lambda_{r_1+r_2+i}\left(\bbL_1\right)=o_p\left(P^{1-\delta_1}\right), \ \ i=1, \ldots, P-(r_1+r_2).  
\end{eqnarray}

Finally, the order of $\lambda_i\left(\bbL_2\right)$ can be derived from Proposition 2.1 of \cite{FLM2013} directly. 
\end{proof}

\subsection*{Proof of Theorem \ref{thm1}}

\begin{proof}
Let $\bbE_{L}^{(1)}=\widehat{\bbL}_1-\bbL_1$ with $\widehat{\bbL}_1=\widehat{\boldsymbol{\Sigma}}_y(1)\widehat{\boldsymbol{\Sigma}}_y(1)^{\top}$ and $\bbL_1=\boldsymbol{\Sigma}_y(1)\boldsymbol{\Sigma}_y(1)^{\top}$. First we evaluate the order of $\left|\left|\bbE_L^{(1)}\right|\right|$. In terms of simple calculations, we have 
\begin{eqnarray}\label{sy51}
\left|\left|\bbE_{L}^{(1)}\right|\right|\leq \left|\left|\widehat{\boldsymbol{\Sigma}}_y(1)-\boldsymbol{\Sigma}_y(1)\right|\right|^2+2\left|\left|\widehat{\boldsymbol{\Sigma}}_y(1)\right|\right|\cdot\left|\left|\widehat{\boldsymbol{\Sigma}}_y(1)-\boldsymbol{\Sigma}_y(1)\right|\right|. 
\end{eqnarray}
In terms of (\ref{sy31}) in Lemma \ref{lem3}, we have $\left|\left|\boldsymbol{\Sigma}_y(1)\right|\right|\asymp P^{1-\delta_1}$. From (\ref{sy30}), we can get 
\begin{align}\label{sy40}
\left|\left|\widehat{\boldsymbol{\Sigma}}_y(1)-\boldsymbol{\Sigma}_y(1)\right|\right|
&\leq \left|\left|\widehat{\bbM}_1-\bbM_1\right|\right|+\left|\left|\widehat{\bbM}_2-\bbM_2\right|\right| \nonumber \\
&\quad +\left|\left|\widehat{\boldsymbol{\Sigma}}_{\varepsilon}(1)-\boldsymbol{\Sigma}_{\varepsilon}(1)\right|\right|, 
\end{align}
where 
\begin{eqnarray*}
\widehat{\bbM}_1=\widehat{\boldsymbol{\Sigma}}_k^{(1)}(1)\bbB^{\top}+\widehat{\boldsymbol{\Sigma}}_k^{(12)}(1)\bbA^{\top},\ \ \ \ \widehat{\bbM}_2=\widehat{\boldsymbol{\Sigma}}_k^{(2)}(1)\bbA^{\top}+\widehat{\boldsymbol{\Sigma}}_k^{(21)}(1)\bbB^{\top},  
\end{eqnarray*}
with $\widehat{\boldsymbol{\Sigma}}_k^{(1)}(1)$, $\widehat{\boldsymbol{\Sigma}}_k^{(12)}(1)$, $\widehat{\boldsymbol{\Sigma}}_k^{(2)}(1)$ and $\widehat{\boldsymbol{\Sigma}}_k^{(21)}(1)$ are the sample covariances corresponding to the population covariances $\boldsymbol{\Sigma}_k^{(1)}(1)$, $\boldsymbol{\Sigma}_k^{(12)}(1)$, $\boldsymbol{\Sigma}_k^{(2)}(1)$ and $\boldsymbol{\Sigma}_k^{(21)}(1)$, respectively. 

Hence, we evaluate (\ref{sy40}) further
\begin{align}\label{sy50}
\left|\left|\widehat{\boldsymbol{\Sigma}}_y(1)-\boldsymbol{\Sigma}_y(1)\right|\right|
&\le \left|\left|\widehat{\boldsymbol{\Sigma}}_k^{(1)}(1)-\boldsymbol{\Sigma}_k^{(1)}(1)\right|\right|+\left|\left|\widehat{\boldsymbol{\Sigma}}_k^{(12)}(1)-\boldsymbol{\Sigma}_k^{(12)}(1)\right|\right|\nonumber \\
&\quad +\left|\left|\widehat{\boldsymbol{\Sigma}}_k^{(2)}(1)-\boldsymbol{\Sigma}_k^{(2)}(1)\right|\right|+\left|\left|\widehat{\boldsymbol{\Sigma}}_k^{(21)}(1)-\boldsymbol{\Sigma}_k^{(21)}(1)\right|\right|\nonumber \\
&\quad +\left|\left|\widehat{\boldsymbol{\Sigma}}_{\varepsilon}(1)-\boldsymbol{\Sigma}_{\varepsilon}(1)\right|\right|\nonumber \\
&= O_p\left(\frac{P^{1-\delta_1}}{T^{1/2}}\right)+O_p\left(\frac{P^{1-\delta_2}}{T^{1/2}}\right)
+O_p\left(\frac{P}{T}\right)\nonumber \\
&= O_p\left(\max\left(\frac{P}{T}, \frac{P^{1-\delta_1}}{T^{1/2}}\right)\right), 
\end{align}
where the last second equality uses (A8) of \cite{LYB2011} which demonstrates $\left|\left|\widehat{\boldsymbol{\Sigma}}_{\varepsilon}(1)-\boldsymbol{\Sigma}_{\varepsilon}(1)\right|\right|=O_p\left(\frac{P}{T}\right)$. 

Then it follows from (\ref{sy51}) and (\ref{sy50}) that
\begin{eqnarray}
\left|\left|\bbE_{L}^{(1)}\right|\right|&=&O_p\left(\max\left(\frac{P^2}{T^2}, \frac{P^{2-2\delta_1}}{T}, \frac{P^{2-\delta_1}}{T}, \frac{P^{2-2\delta_1}}{T^{1/2}}\right)\right)\\
&=&O_p\left(\frac{P^{2-2\delta_1}}{T^{1/2}}\right), 
\end{eqnarray}
where the last equality uses the assumption that $P^{\delta_1}=o\left(T^{1/2}\right)$. 

Now we use Lemma \ref{lem2} to get the order of estimated factor loadings. In Lemma \ref{lem2}, let $\bbF$ and $\bbE$ be $\bbL_1$ and $\widehat{\bbL}_1-\bbL_1$, respectively. Let $k$ in Lemma 3 equal to $r_1$. Then we have, from (\ref{sy31}), 
\begin{eqnarray}
sep(\bbD_1, \bbD_2)\asymp P^{2-2\delta_1}, 
\end{eqnarray} 
where the defition of $sep(\cdot, \cdot)$ is provided in Lemma \ref{lem2}. Then $\bbE_{L}^{(1)}$ and \\$sep(\bbD_1, \bbD_2)$ satisfies
\begin{eqnarray}
\left|\left|\bbE_L^{(1)}\right|\right|=o_p\left(sep(\bbD_1, \bbD_2)\right)\leq \frac{sep(\bbD_1, \bbD_2)}{5}. 
\end{eqnarray} 
Hence Lemma \ref{lem2} tells us that, there exists a matrix $\bbP: (P-r_1)\times r_1$ such that   
\begin{eqnarray}\label{sy61}
\left|\left|\bbP\right|\right|
\leq\frac{4}{sep\left(\bbD_1, \bbD_2\right)}\cdot\left|\left|\left(\bbE_{L}^{(1)}\right)_{21}\right|\right|\leq\frac{4\left|\left|\bbE_L^{(1)}\right|\right|}{sep(\bbD_1, \bbD_2)}
\end{eqnarray}
and then $\widehat{\bbB}=\left(\bbB+\bbB^{c}\bbP\right)\left(\bbI+\bbP^{\top}\bbP\right)^{-1/2}$ is an estimator of $\bbB$ with $\bbB^{c}$ being $\bbQ_2$ in Lemma \ref{lem2}. In view of this, the rate of convergence for $\widehat{\bbB}$ can be calculated as
\begin{align}\label{sy60}
\left|\left|\widehat{\bbB}-\bbB\right|\right|
&= \left|\left|\left(\bbB+\bbB^{c}\bbP\right)\left(\bbI+\bbP^{T}\bbP\right)^{-1/2}-\bbB\right|\right|\nonumber \\
&= \left|\left|\left[\left(\bbB+\bbB\bbP\right)-\bbB\left(\bbI+\bbP^{\top}\bbP\right)^{1/2}\right]\left(\bbI+\bbP^{\top}\bbP\right)^{-1/2}\right|\right|\nonumber \\
&= \left|\left|\left(\bbB\left[\bbI-\left(\bbI+\bbP^{\top}\bbP\right)^{1/2}\right]+\bbB^{c}\bbP\right)\left(\bbI+\bbP^{\top}\bbP\right)^{-1/2}\right|\right|\nonumber \\
&\le \left|\left|\bbI-\left(\bbI+\bbP^{\top}\bbP\right)^{1/2}\right|\right|+\left|\left|\bbP\right|\right|\le 2\left|\left|\bbP\right|\right|,  
\end{align}
where the last second equality uses the fact that $\bbB$ and $\bbB^{c}$ are orthonormal; and the last equality uses the fact that 
\begin{eqnarray}
\left|\left|\bbI-\left(\bbI+\bbP^{\top}\bbP\right)^{1/2}\right|\right|&=&1-\left(1+\lambda_{\min}\left(\bbP^{\top}\bbP\right)\right)^{1/2}\\
&\leq& \lambda^{1/2}_{\max}\left(\bbP^{\top}\bbP\right). 
\end{eqnarray}
Therefore, by (\ref{sy60}) and (\ref{sy61}), we obtain
\begin{eqnarray}\label{sy80}
\left|\left|\widehat{\bbB}-\bbB\right|\right|=O_P\left(\frac{\left|\left|\bbE_L^{(1)}\right|\right|}{sep(\bbD_1, \bbD_2)}\right)=O_p\left(\frac{1}{\sqrt{T}}\right). 
\end{eqnarray}

For the second factor model part, the estimation is to conduct principal component analysis on the residual of the first step, i.e. estimating the factor model 
\begin{eqnarray}\label{sy70}
\widehat{\bbu}_t=\bbA\bbk_t^{(2)}+\boldsymbol{\eta}_t, \ \ t=1, 2, \ldots, T, 
\end{eqnarray}
where $\widehat{\bbu}_t=\bby_t-\widehat{\bbB}\widehat{\bbk}_t^{(1)}$, $\boldsymbol{\eta}_t$ is the new error component in the estimation at the second step.   

In order to derive the rate of convergence for $\widehat{\bbA}$, we also utilize Lemma \ref{lem2}. Now let $\bbF$ and $\bbE$ in Lemma \ref{lem2} are $\bbL_2$ and $\bbE_{L}^{(2)}:=\widehat{\bbL}_2-\bbL_2$, respectively. Let $k$ in Lemma \ref{lem2} equal to $r_2$. 

First, we evaluate $\left|\left|\bbE_{L}^{(2)}\right|\right|$. Based on (\ref{sy70}), we have 
\begin{eqnarray}\label{sy96}
\left|\left|\bbE_L^{(2)}\right|\right|
\leq\left|\left|\widehat{\boldsymbol{\Sigma}}_{\widehat{u}}(0)-\boldsymbol{\Sigma}_u(0)\right|\right|^2+2\left|\left|\boldsymbol{\Sigma}_u(0)\right|\right|\cdot\left|\left|\widehat{\boldsymbol{\Sigma}}_{\widehat{u}}(0)-\boldsymbol{\Sigma}_u(0)\right|\right|, 
\end{eqnarray}
where $\boldsymbol{\Sigma}_u(0)$ is the population covariance matrix of $\bbu_t$ and $\widehat{\boldsymbol{\Sigma}}_{\widehat{u}}(0)$ is the sample covariance matrix of $\widehat{\bbu}_t$. 
Based on Assumption \ref{assu3} and Proposition 2.1 of \cite{FLM2013}, we know that $\left|\left|\boldsymbol{\Sigma}_u(0)\right|\right|\asymp p$. For the term $\left|\left|\widehat{\boldsymbol{\Sigma}}_{\widehat{u}}(0)-\boldsymbol{\Sigma}_u(0)\right|\right|$, we evaluate its order as follows. 
\begin{align}\label{sy92}
\left|\left|\widehat{\boldsymbol{\Sigma}}_{\widehat{u}}(0)-\boldsymbol{\Sigma}_u(0)\right|\right|
&\le \left|\left|\widehat{\boldsymbol{\Sigma}}_{\widehat{u}}(0)-\widehat{\boldsymbol{\Sigma}}_u(0)\right|\right|+\left|\left|\widehat{\boldsymbol{\Sigma}}_u(0)-\boldsymbol{\Sigma}_u(0)\right|\right|\nonumber \\
&\le \frac{1}{T}\left|\left|\widehat{\bbB}\widehat{\bbK}^{(1)}-\bbB\bbK^{(1)}\right|\right|^2
+\frac{2}{T}\left|\left|\bbB\bbK^{(1)}\right|\right|\cdot\left|\left|\widehat{\bbB}\widehat{\bbK}^{(1)}-\bbB\bbK^{(1)}\right|\right|\nonumber \\
&\quad +\left|\left|\widehat{\boldsymbol{\Sigma}}_u(0)-\boldsymbol{\Sigma}_u(0)\right|\right|, 
\end{align}
where $\widehat{\bbK}^{(1)}=\left(\widehat{\bbk}_1^{(1)}, \widehat{\bbk}_2^{(1)}, \ldots, \widehat{\bbk}_T^{(1)}\right)$ and $\bbK^{(1)}=\left(\bbk_1^{(1)}, \bbk_2^{(1)}, \ldots, \bbk_T^{(1)}\right)$.

From Assumption \ref{assu2} and (\ref{sy80}), it can be derived that
\begin{eqnarray}\label{sy90}
\frac{1}{\sqrt{T}}\left|\left|\widehat{\bbB}\widehat{\bbK}^{(1)}-\bbB\bbK^{(1)}\right|\right|&=&O_p\left(\frac{P^{1/2-\delta_1/2}}{\sqrt{T}}\right),\nonumber \\ 
\frac{1}{\sqrt{T}}\left|\left|\bbB\bbK^{(1)}\right|\right|&=&O_p\left(P^{1/2-\delta_1/2}\right). 
\end{eqnarray}
Similar to (\ref{sy50}), we can also get 
\begin{align}\label{sy91}
\left|\left|\widehat{\boldsymbol{\Sigma}}_u(0)-\boldsymbol{\Sigma}_u(0)\right|\right|
&\le \left|\left|\widehat{\boldsymbol{\Sigma}}_k^{(2)}(0)-\boldsymbol{\Sigma}_k^{(2)}(0)\right|\right|+\left|\left|\widehat{\boldsymbol{\Sigma}}_{\varepsilon}(0)-\boldsymbol{\Sigma}_{\varepsilon}(0)\right|\right|\nonumber \\
&= O_p\left(\frac{P}{T^{1/2}}\right)+O_p\left(\frac{P}{T}\right). 
\end{align}
In view of (\ref{sy90}), (\ref{sy91}) and (\ref{sy92}), we can get 
\begin{eqnarray}\label{sy97}
\left|\left|\widehat{\boldsymbol{\Sigma}}_{\widehat{u}}(0)-\boldsymbol{\Sigma}_u(0)\right|\right|
&&=O_P\left(\frac{p^{1-\delta_1}}{\sqrt{T}}\right)+O_p\left(\frac{P}{\sqrt{T}}\right)
+O_p\left(\frac{P}{T}\right)\nonumber \\
&&=O_p\left(\frac{P}{\sqrt{T}}\right). 
\end{eqnarray}
The order of $\left|\left|\bbE_L^{(2)}\right|\right|$ is obtained from (\ref{sy96}) and (\ref{sy97}), i.e. 
\begin{eqnarray}
\left|\left|\bbE_L^{(2)}\right|\right|=O_p\left(\frac{P^2}{\sqrt{T}}\right). 
\end{eqnarray}
Moreover, it follows from Proposition 2.1 of \cite{FLM2013} that 
\begin{eqnarray}
sep(\bbD_1, \bbD_2)\asymp P^2. 
\end{eqnarray}
Here $\bbD_1$ in Lemma \ref{lem2} is the diagonal matrix corresponding to the orthogonal matrix $\bbA$. Then we can get from Lemma \ref{lem2} that 
\begin{eqnarray}
\left|\left|\widehat{\bbA}-\bbA\right|\right|
=O_p\left(\frac{\left|\left|\bbE_L^{(2)}\right|\right|}{sep\left(\bbD_2, \bbD_1\right)}\right)
=O_p\left(\frac{1}{\sqrt{T}}\right). 
\end{eqnarray}
\end{proof}

\subsection*{Proof of Theorem \ref{thm2}}
\begin{proof}
It follows from the model (\ref{dpca:eq:y1}) and (\ref{dpca:eq:y2}) that 
\begin{eqnarray*}
&&\frac{1}{P}\left|\left|\widehat{\bby}_{t-1}(1)-\bby_t\right|\right|^2
=\frac{1}{P}\sum^{P}_{i=1}\left(\widehat{y}_{i,t-1}(1)-y_{it}\right)^2\\
&=&\frac{1}{P}\sum^{P}_{i=1}\left(\widehat{\bbb}^{\top}_i\widehat{\bbk}_{t-1}^{(1)}(1)+\widehat{\bba}_i^{\top}\widehat{\bbk}^{(2)}_{t-1}(1)-\bbb_i\bbk_t^{(1)}-\bba^{\top}_i\bbk_t^{(2)}-\varepsilon_{it}\right)^2\\
&\leq&\frac{1}{P}\sum^{P}_{i=1}\left(\left(\widehat{\bbb}_i-\bbb_i\right)^{\top}\widehat{\bbk}^{(1)}_{t-1}(1)\right)^2
+\frac{1}{P}\sum^{P}_{i=1}\left(\left(\widehat{\bba}_i-\bba_i\right)^{\top}\widehat{\bbk}^{(2)}_{t-1}(1)\right)^2\\
&&+\frac{1}{P}\sum^{P}_{i=1}\left[\bbb_i^{\top}\left(\widehat{\bbk}_{t-1}^{(1)}(1)-\bbk^{(1)}_t\right)\right]^2
+\frac{1}{P}\sum^{P}_{i=1}\left[\bba_i^{\top}\left(\widehat{\bbk}_{t-1}^{(2)}(1)-\bbk^{(2)}_t\right)\right]^2
+\frac{1}{P}\sum^{P}_{i=1}\varepsilon_{it}^2\\
&=:&\sum^{5}_{d=1}C_d. 
\end{eqnarray*}
From Theorem \ref{thm1}, we have the orders of $C_1$ and $C_2$ are $\frac{1}{T}$. The orders of $C_3$ and $C_4$ reply on the forecasting error of two kinds of factors $\widehat{\bbk}^{(j)}_{t-1}(1)$ with $j=1, 2$. The last term $\varepsilon_{it}$ comes from the second factor model error in the second step. 
\end{proof}

\end{document}